%% file: main.tex
\documentclass[11pt,a4paper]{article}
\usepackage[a4paper,margin=30mm]{geometry}

\usepackage[utf8]{inputenc}
\usepackage{lmodern}
\usepackage{microtype}

\usepackage{bbm}

\usepackage[english]{babel}

\usepackage[table]{xcolor}
\usepackage{enumitem}
\usepackage[most]{tcolorbox}
\usepackage[normalem]{ulem} 

\usepackage{amsmath,amssymb,amsthm}
\usepackage{mathrsfs}
\usepackage{thmtools}
\usepackage{bm}

\usepackage{graphicx}
\graphicspath{{figures/}{manuscript/figures/}}
\usepackage{pdfpages}
\usepackage{siunitx}
\usepackage{booktabs}
\usepackage{longtable}

\usepackage[hidelinks]{hyperref}
\usepackage[capitalize,nameinlink,noabbrev]{cleveref}

\usepackage{setspace}

\usepackage{algorithm}
\usepackage{algpseudocode}

\definecolor{TartareInk}{HTML}{20242A}
\definecolor{TartareFrame}{HTML}{6F7D8C}
\definecolor{TartareBack}{HTML}{F7F8F6}
\definecolor{TartareSoft}{HTML}{EEF2F1}
\definecolor{TartareAccent}{HTML}{4F6F64}

\newtcolorbox{tartarebox}[2][]{
  enhanced,
  breakable,
  colback=TartareBack,
  colframe=TartareFrame,
  coltitle=white,
  colbacktitle=TartareAccent,
  fonttitle=\bfseries,
  title={#2},
  sharp corners=south,
  arc=1.2mm,
  boxrule=0.45pt,
  left=7pt,
  right=7pt,
  top=7pt,
  bottom=7pt,
  before skip=9pt,
  after skip=9pt,
  #1
}

\newtcolorbox{tartarenote}[2][]{
  enhanced,
  breakable,
  colback=TartareSoft,
  colframe=TartareFrame,
  boxrule=0.35pt,
  borderline west={2pt}{0pt}{TartareAccent},
  title={#2},
  coltitle=TartareInk,
  fonttitle=\bfseries,
  attach boxed title to top left={xshift=5pt,yshift=-2pt},
  boxed title style={colback=TartareSoft,colframe=TartareSoft,boxrule=0pt},
  left=8pt,
  right=8pt,
  top=8pt,
  bottom=7pt,
  before skip=8pt,
  after skip=8pt,
  #1
}

\newtheorem{theorem}{Theorem}
\newtheorem{lemma}[theorem]{Lemma}
\newtheorem{proposition}{Proposition}
\newtheorem{corollary}{Corollary}
\newtheorem{fact}{Fact}
\theoremstyle{definition}
\newtheorem{definition}{Definition}
\newtheorem{assumption}{Assumption}
\theoremstyle{remark}

\crefname{assumption}{Assumption}{Assumptions}
\Crefname{assumption}{Assumption}{Assumptions}

\newenvironment{maintheorem}[1][]{\begin{theorem}[#1]}{\par\noindent\emph{Proof in Supplementary Material.}\end{theorem}}

\newenvironment{mainproposition}[1][]{\begin{proposition}[#1]}{\par\noindent\emph{Proof in Supplementary Material.}\end{proposition}}

\DeclareMathOperator{\E}{\mathbb{E}}

\DeclareMathOperator{\Cov}{Cov}

\DeclareMathOperator{\diag}{diag}

\newcommand{\norm}[1]{\left\lVert #1 \right\rVert}
\newcommand{\diff}{\,\mathrm d}
\newcommand{\indep}{\perp\!\!\!\perp}

\usepackage{comment}
\excludecomment{mainproof}

\usepackage{natbib}

\newif\ifshowchanges
\showchangesfalse

\title{Gaussian Process Differential Ensembles for Joint Inference on Curves, Derivatives, and Integrals}
\author{
  Andreas Kryger Jensen\thanks{Corresponding author: Andreas Kryger Jensen, Section of Biostatistics, Department of Public Health, University of Copenhagen, {\O}ster Farimagsgade 5, DK-1353 Copenhagen K, Denmark. Email: \href{mailto:aeje@sund.ku.dk}{aeje@sund.ku.dk}. Telephone: +45 35 33 35 28.}\, and Adam Gorm Hoffmann\\
  Section of Biostatistics, Department of Public Health\\
  University of Copenhagen}
\date{}

\begin{document}

\onehalfspacing

\maketitle



\begin{abstract}
Functional data are often modeled through one likelihood-linked curve, while the scientific target is a larger state containing rates, accumulated quantities, boundary values, or nonlinear functionals of several linked levels. These targets require more than smoothing the observed curve: derivative uncertainty, cross-level covariance, and integration constants must be handled jointly. We introduce anchored Gaussian process differential ensembles, embedding an anchor \(f_0\) in a joint Gaussian state with its mean-square derivatives and repeated integrals. Integral levels add explicit Gaussian integration constants. This separates the anchor-induced covariance from finite-dimensional boundary uncertainty and clarifies why anchor-only observations do not identify independent integration constants. For stationary one-dimensional kernels, we compute the ensemble with a transformed Hilbert space Gaussian process approximation that applies derivative and integral operators to Laplacian--Dirichlet basis functions while retaining the integration-constant covariance exactly. We establish operator-level approximation bounds and conditional finite-grid posterior convergence. We introduce TARTARE, a target-aware calibration procedure for finite-rank differential ensemble approximations, to address derivative under-resolution by anchor-calibrated bases. In second-order simulations, derivative-aware calibration improves derivative posterior recovery relative to anchor-only calibration while preserving anchor and integral summaries. A motorcycle crash analysis illustrates coherent posterior inference on a coupled kinematic state and short-horizon turning-point functionals.
\end{abstract}

\noindent%

\section{Introduction}
Functional data are often modeled through one latent curve that is linked directly to the observations. In many applications, however, that curve is only one coordinate of the scientific state. A curve may be observed because it is measured most directly, while the quantities of interest are rates, accelerations, accumulated effects, boundary states, or event probabilities defined by several different derivative and integral levels. In biomechanics, for example, the observed or modeled quantity may be acceleration while questions concern velocity, position, and turning points. More generally, the target may be a nonlinear functional of a state vector containing an anchor curve, its derivatives, and its repeated integrals. A Bayesian analysis of such targets should propagate posterior uncertainty through the differential relationships among the state components.

The statistical difficulty is not merely that derivatives and integrals can be computed after fitting a smooth curve. Post-processing a fitted curve can obscure the covariance structure needed for inference.  Differentiating a posterior mean ignores posterior covariance and therefore cannot provide calibrated uncertainty for rates or accelerations.  Finite differences depend on the chosen grid and become unstable when derivative order increases.  Antiderivatives are not uniquely determined without integration constants, so cumulative or boundary-state inference requires boundary information in addition to the anchor-level likelihood.  Standard Gaussian process derivative and linear operator constructions provide joint covariance identities for transformed processes, but positive-index boundary uncertainty is often left implicit.  Finite-rank approximations add a further difficulty because a basis calibrated for the anchor curve can under-resolve derivative targets because differentiation emphasizes high-frequency components.

This paper develops anchored Gaussian process differential ensembles for this setting. The starting point is a Gaussian process \(f_0\), called the anchor, because it is the component linked to the observation model. Negative-index components \(f_{-1},\ldots,f_{-r}\) are mean-square derivatives of the anchor. Positive-index components \(f_1,\ldots,f_r\) are repeated definite integrals of the anchor from a reference point \(t_0\), plus Gaussian integration constants \(\kappa_1,\ldots,\kappa_r\).  The resulting object is one coherent joint Gaussian law for the derivative--anchor--integral state.  The weak construction determines all finite-dimensional distributions and covariance blocks, and under additional regularity the same law can be represented by sample paths satisfying the corresponding differential chain.

The integration constants are part of the statistical model.  The canonical repeated integral from \(t_0\) is the antiderivative whose lower-order initial values vanish at \(t_0\).  A general antiderivative differs from this canonical representative by a polynomial, and the integration constants are precisely the finite-dimensional boundary values specifying that polynomial, echoing the polynomial null spaces that arise in spline and intrinsic-Gaussian-prior formulations \citep{kimeldorf1970correspondence,wahba1978improper,matheron1973intrinsic}.  Anchor-only observations do not identify independent integration constants, and inference above the anchor therefore depends on boundary observations, virtual constraints, or prior information. Boundary-constrained Gaussian process constructions impose such information as hard linear restrictions \citep{langehegermann2021linearly,gulian2022gaussian}, whereas our main construction treats the constants as explicit Gaussian variables.  In this paper we will focus on the independent-constant construction for clarity, but the online supplement contains a more detailed development of the dependent extension.

For scalable computation, we approximate the kernel-driven operator blocks for stationary one-dimensional anchor kernels using a transformed Hilbert space Gaussian process representation based on Laplacian--Dirichlet eigenfunctions \citep{solin2020hilbert}.  Derivatives and definite integrals are applied to the basis functions rather than to dense covariance matrices, while the covariance contribution from integration constants is retained exactly.  Thus, for any prescribed finite ensemble order allowed by the regularity assumptions, the required derivative and integral blocks are computed by transforming basis functions, without deriving covariance-family-specific derivative or integral formulas.  This computational representation must be calibrated for the targets actually reported. Differentiation emphasizes high frequencies, so a finite-rank basis that is accurate for the anchor covariance can still be inadequate for derivative blocks. We therefore introduce TARTARE, Target-Aware Range and Truncation for Approximate Representations of Ensembles, as a practical calibration procedure for finite-rank implementations. TARTARE extends single-process HSGP calibration rules by choosing the computational range and truncation level against a user-specified monitoring set \(\mathcal M\subseteq\{-r,\ldots,r\}\).

The paper makes three main contributions.
\begin{enumerate}[label=(\arabic*)]
\item We define the anchored Gaussian differential ensemble construction for joint posterior inference on a likelihood-linked anchor curve, its derivatives, and repeated integrals, clarifying the joint Gaussian law of the differential ensemble and its weak mean-square differential interpretation.
\item We make positive-index integration constants explicit as finite-dimensional boundary information, separate their covariance from the anchor-induced kernel part, and show why anchor-only observations do not identify independent boundary constants.
\item It gives a transformed Laplacian eigenspectrum implementation for stationary one-dimensional anchor kernels, establishes operator-level approximation bounds and finite-grid posterior convergence conditional on fixed hyperparameters, and introduces TARTARE calibration for derivative- and ensemble-aware finite-rank inference.
\end{enumerate}

\subsection{Related work}\label{subsec:related-work}

Gaussian process derivatives, gradient kriging, and linear operator Gaussian processes provide the closest mathematical antecedents.  The covariance identities for a process and its derivatives are standard in derivative Gaussian process modeling \citep{solak2002derivative}, and more general linear transformations of Gaussian processes have been used for operator-constrained prediction and inter-domain models \citep{graepel2003solving,murraysmith2005transformations,lazaro2009interdomain,sarkka2011linear}. Related latent force models use Gaussian process priors together with differential equation operators to construct physically motivated covariance functions \citep{alvarez2013linear}.  These constructions show that derivative and operator transformed quantities can be modeled jointly with an underlying Gaussian process.  The distinction here is the anchor organization where the likelihood is attached to a specified state component, while derivative levels, integral levels, and boundary constants are represented in a single posterior state.

Integrated processes and latent rate models emphasize the complementary direction in which a latent rate or derivative process is integrated to obtain the observed curve.  Examples include Gaussian process priors on derivatives that are integrated to recover a latent function \citep{holsclaw2013gaussian}, integrated latent rate models in environmental time series \citep{cahill2015modeling}, and integrated Wiener or stochastic differential equation smoothing priors used for joint inference on functions and derivatives \citep{yue2014bayesian,zhang2024model}.  Such models are natural when the rate process is the primitive object.  The present construction instead allows the observation linked curve to be chosen as the anchor and then represents both derivatives below the anchor and repeated integrals above it with integration constants kept explicit.

Derivative estimation also has a long functional data and smoothing spline tradition. A common approach is to fit a smooth basis expansion and differentiate the fitted representation \citep{silverman1985spline,ramsay2005functional}, and related generalized smoothing approaches encode differential equation structure through derivative penalties or profiling \citep{ramsay2007parameter}. These approaches are useful for exploratory derivative summaries, but the resulting uncertainty depends on how smoothing, derivative penalties, and boundary behavior are encoded in the fitted representation. Applied and methodological Gaussian process work likewise shows that rates, stationary points, and extrema can be primary inferential targets \citep{swain2016inferring,yu2023bayesian,li2024semiparametric}, and recent theory gives positive results for plug-in Gaussian process derivative inference \citep{liu2026optimal}. The comparison made here is not with a single estimator, but with the broader post-processing workflow. When uncertainty in derivatives, integrals, boundary states, or nonlinear functionals is central, those quantities should be part of the modeled posterior state.

Hilbert space Gaussian process (HSGP) approximations provide the computational background for the finite-rank representation used in this manuscript.  The HSGP construction of \citet{solin2020hilbert} approximates stationary kernels on compact domains through Laplacian eigenfunctions and spectral densities, with related finite-domain spectral-feature approximations developed by \citet{hensman2018variational}.  \Citet{riutort2023practical} developed practical calibration rules for probabilistic programming implementations.  Our use of this approximation is target-aware in which the basis is transformed to approximate derivative and integral covariance blocks, and calibration is tied to monitored differential targets rather than only to the anchor covariance.

The rest of the paper is organized as follows.  \Cref{sec:anchored-ensembles} develops the exact anchored Gaussian differential ensemble construction.  \Cref{sec:finite-rank} develops the transformed HSGP approximation, approximation bounds, and finite-grid posterior convergence result. \Cref{sec:tartare-calibration} introduces TARTARE calibration.  \Cref{sec:simulation} presents the simulation evidence.  \Cref{sec:data} gives the motorcycle crash application, and \cref{sec:discussion} concludes.  Proofs, further details on dependent integration constants, calibration details, simulation diagnostics, and application-specific details are collected in the online supplement.

\section{Methods}\label{sec:methods}

\paragraph{Notation.}
All processes are indexed by a fixed interval $[t_0,t_1]\subset\mathbb R$. When an ensemble order is fixed, \(r\in\mathbb N\) denotes that order and signed level indices \(p,q\) range over \(\{-r,\ldots,r\}\), with \(f_0\) the anchor, negative levels denoting derivatives, and positive levels denoting integrals. We write \(\mathbf 1_{\{\cdot\}}\) for indicators, \(\mathbb N_0=\{0,1,\ldots\}\), and \(\|\cdot\|_\infty\) for the supremum norm. For an integer $q\ge0$, write $D_t^q$ for the $q$th derivative, with $D_t^0=\mathrm{Id}$, and define the repeated integral operator by
\[
(\mathcal I_{t_0}^q f)(t)=
\begin{cases}
 f(t), & q=0,\\[3pt]
 \displaystyle \frac{1}{(q-1)!}\int_{t_0}^t (t-u)^{q-1}f(u)\,du, & q>0.
\end{cases}
\]
For $q>0$, $\mathcal I_{t_0}^q f$ is the $q$-fold antiderivative whose derivatives of orders $0,\ldots,q-1$ vanish at $t_0$. General antiderivatives are obtained by adding a polynomial of degree at most $q-1$. Those finite-dimensional degrees of freedom are represented below by explicit integration constants, and not by changing the operator $\mathcal I_{t_0}^q$.

For bivariate kernels, $D_1^q$ and $D_2^q$ differentiate the first and second arguments, while $\mathcal I_{1,t_0}^qH$ and $\mathcal I_{2,t_0}^qH$ apply $\mathcal I_{t_0}^q$ to $H(\cdot,t)$ and $H(s,\cdot)$, respectively. Define the signed-order operator
\[
\mathcal A_{t_0}^{[k]}=
\begin{cases}
D_t^{-k}, & k<0,\\
\mathrm{Id}, & k=0,\\
\mathcal I_{t_0}^{k}, & k>0,
\end{cases}
\]
with analogous definitions for $\mathcal A_{1,t_0}^{[k]}$ and $\mathcal A_{2,t_0}^{[k]}$.

A Gaussian process with mean $\mu$ and covariance $C_\theta$ is denoted $\mathcal{GP}(\mu,C_\theta)$ and is characterized by Gaussian finite-dimensional distributions. For a compact interval $I$ and $a,b\in\{0,\ldots,r\}$, $C^{a,b}(I^2)$ denotes functions with continuous mixed derivatives $D_1^uD_2^v$ for $0\le u\le a$ and $0\le v\le b$, endowed with the corresponding supremum norm. All stochastic integrals and derivative limits below are interpreted in $L^2$ unless a separate pathwise statement is explicitly invoked.

\subsection{Anchored Differential Ensembles}\label{sec:anchored-ensembles}

We begin with recalling the standard Gaussian process regression model for observed data $(Y_i, t_i)_{i=1}^n$,
\begin{align*}
Y_i \mid f, \sigma^2 \overset{\indep}{\sim} N(f(t_i), \sigma^2), \qquad f \mid \theta \sim \mathcal{GP}(0, C_\theta),
\end{align*}
where $f$ is the latent Gaussian process governing the expected value of the observations. When derivatives and repeated integrals of the latent curve are also of simultaneous scientific interest, we embed $f$ in a centered collection $(f_{-r}, \ldots, f_r)$ and refer to $f_0 := f$ as the anchor, because it is the component linked directly to the observations, and the component from which all other members in the collection are derived. Negative-index components are obtained by differentiating $f_0$, and positive-index components are obtained by taking canonical repeated definite integrals of $f_0$ from the reference point $t_0$ and then adding a random polynomial that restores the desired initial values. This decomposition is central throughout the paper, and it separates the kernel-driven functional part of the model from the statistically distinct uncertainty carried by the integration constants.

To fix ideas in a visual manner, \cref{fig:single-curve-ensemble} shows ten realizations of a second-order differential ensemble, $\mathfrak{f} := (f_{-2}, f_{-1},f_0,f_1, f_{2})$, for three covariance families. Within each row, the five panels correspond to successive derivative and integral levels built from the same underlying draw, so the figure makes the chain structure of the ensemble explicit: choosing a color and moving left within the same row corresponds to differentiation and moving right corresponds to repeated integration from \(t_0\) with stochastic initial values added above the anchor. A differential ensemble is formally defined in the following definition.

\begin{figure}[htbp]
\centering
\includegraphics[width=\textwidth]{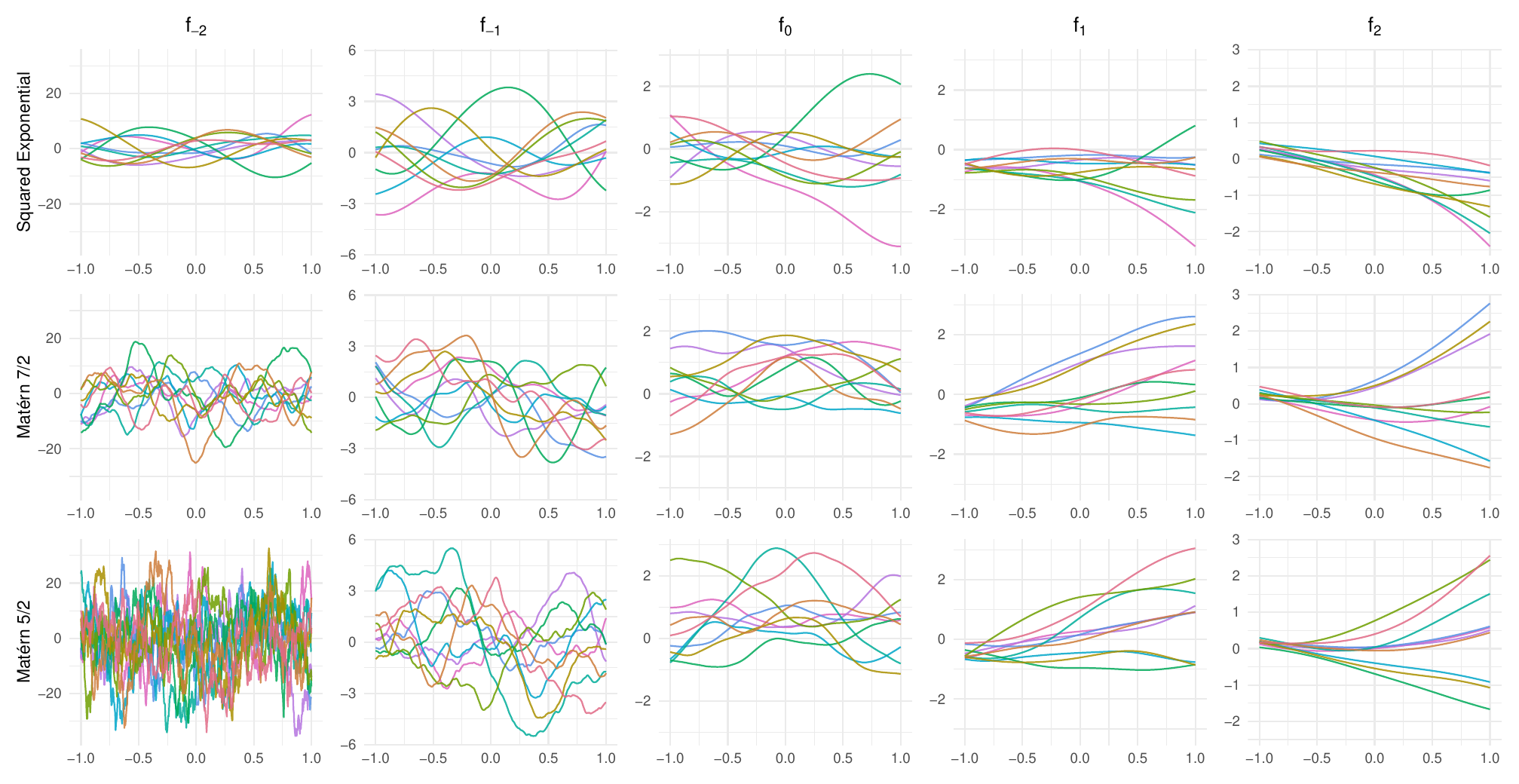}
\caption{Ten simulated realizations from a second-order differential ensemble under squared exponential, Mat\'ern $7/2$, and Mat\'ern $5/2$ covariance kernels. Each row shows the coupled levels $(f_{-2}, f_{-1}, f_0, f_1, f_2)$ from the same set of draws.}
\label{fig:single-curve-ensemble}
\end{figure}

\begin{definition}[Centered $r$-order differential ensemble]\label{def:ensemble}
Fix $r \in \mathbb{N}$. A centered $r$-order differential ensemble on $[t_0, t_1] \subset \mathbb{R}$ is a $(2r+1)$-tuple of real-valued processes
\begin{align*}
\mathfrak{f} := (f_{-r}, \ldots, f_{-1}, f_0, f_1, \ldots, f_r)
\end{align*}
indexed by $[t_0,t_1]$ such that $\E[f_p(t)] = 0$ and $\E[f_p(t)^2] < \infty$ for all $p \in \{-r,\ldots,r\}$ and $t \in [t_0,t_1]$, and satisfying:
\begin{enumerate}
\item[(i)] for each $p=-r+1,\ldots,r$ and $t\in(t_0,t_1)$,
\begin{align*}
\lim_{h \to 0} \E\left[\left(\frac{f_p(t+h) - f_p(t)}{h} - f_{p-1}(t)\right)^2\right] = 0;
\end{align*}
\item[(ii)] for \(p=1,\ldots,r\), write \(\kappa_p := f_p(t_0)\). These boundary values are the integration constants of the positive-index levels.
\end{enumerate}
\end{definition}

We first record the regularity assumptions needed for the weak Gaussian process construction.

\begin{assumption}[Mean-square regularity]\label{ass:ms-regularity}
Let $f_0 \sim \mathcal{GP}(\mu, C_\theta)$ on $[t_0, t_1]$. Assume that $\mu \in C^r([t_0, t_1])$ and $C_\theta \in C^{r,r}([t_0, t_1]^2)$.
\end{assumption}
These assumptions ensure that the mean-square derivatives of $f_0$ up to order $r$ exist and that the differentiated covariance function is well defined; see, for example, \citet{cramer1967stationary} and \citet[Sec.~9.4]{gpml}.

Throughout, we assume that the integration constants are independent of the anchor process. This keeps the covariance formula transparent and is natural when the constants represent initial-value uncertainty not determined by the anchor itself, but it is not a structural restriction. If \((f_0,\boldsymbol{\kappa})\) is instead jointly Gaussian, the defining representation and boundary identities are unchanged, while the covariance blocks acquire cross-covariance terms determined by \(\boldsymbol{\gamma}(s)=\Cov(f_0(s),\boldsymbol{\kappa})\). \Cref{appendix:dependent-integration-constants} in the Supplementary Material gives a more general construction in which the integration constants need not be independent of the anchor process. The idea is to introduce a second Gaussian process \(u_0\), coupled jointly with \(f_0\), and use the boundary derivative values of \(u_0\) as the integration constants for the ensemble built from \(f_0\), $\kappa_m = D_t^{m-1}u_0(t_0)$, $m=1,\ldots,r$. Equivalently, \(u_0\) may itself be expanded into its own differential ensemble, and the derivative side of that second ensemble supplies the boundary values needed to define the integral side of the first ensemble. This produces dependence between \(f_0\) and \(\boldsymbol{\kappa}\) through their joint Gaussian law, summarized by \(\boldsymbol{\gamma}(s)=\Cov(f_0(s),\boldsymbol{\kappa})\). The independent model used below is recovered by setting this cross-covariance to zero, \(\boldsymbol{\gamma}\equiv\mathbf 0\).

\begin{assumption}[Independent boundary constants]\label{ass:independent-boundary}
Conditional on fixed model hyperparameters, let $\boldsymbol\kappa=(\kappa_1,\dots,\kappa_r)^{\top}
\sim N(\boldsymbol\mu_\kappa,\boldsymbol\Sigma_\kappa)$ be independent of $f_0$.\end{assumption}

The next theorem constructs the Gaussian differential ensemble in the weak mean-square sense. It establishes the derivative chain, the boundary identities, and the joint Gaussian law.

\begin{maintheorem}[Weak Gaussian construction]\label{prop:gaussian-ensemble}
Assume \cref{ass:ms-regularity,ass:independent-boundary}. For $p=-r,\ldots,-1$, let $f_p$ be the $(-p)$th mean-square derivative process of $f_0$, and keep $f_0$ as the anchor process. For $p=1,\ldots,r$, define, for $t\in[t_0,t_1]$
\begin{align*}
  f_p(t)
:=
\frac{1}{(p-1)!}\int_{t_0}^t (t-u)^{p-1}f_0(u)\,du
+
\sum_{j=1}^{p}\frac{(t-t_0)^{p-j}}{(p-j)!}\kappa_j,
\end{align*}
where the integral is understood as an $L^2$-random variable.
The definite integral form separates the anchor-induced cumulative component from the finite-dimensional boundary information carried by \(\boldsymbol{\kappa}\), and the \(t_0\)-based definition yields the equivalent recursive initial-value formulation \(D_t f_p(t)=f_{p-1}(t)\), \(f_p(t_0)=\kappa_p\). Equivalent positive-index formulations are collected in \cref{appendix:alternative-positive-constructions} in the Supplementary Material.

Then \(\mathfrak f := (f_{-r},\ldots,f_r)\) satisfies the mean-square derivative chain and boundary identities in \cref{def:ensemble}. If, in addition, \(\mu = 0\) and \(\boldsymbol\mu_\kappa=\mathbf 0\), then \(\mathfrak f\) is a centered \(r\)th-order differential ensemble in the sense of \cref{def:ensemble}. Moreover, for every finite collection \((p_i,\tau_i)_{i=1}^m\) with \(p_i\in\{-r,\ldots,r\}\) and \(\tau_i\in[t_0,t_1]\), $\bigl(f_{p_1}(\tau_1),\ldots,f_{p_m}(\tau_m)\bigr)^\top$
is multivariate Gaussian. Equivalently, \(\mathbf f(t):=(f_{-r}(t), \ldots, f_r(t))^\top\) is an \(\mathbb{R}^{2r+1}\)-valued Gaussian process whose finite-dimensional distributions are determined by the mean functions
\begin{align*}
\mu_p(t)
&= (\mathcal A_{t_0}^{[p]}\mu)(t)
+ \mathbf{1}_{\{p>0\}}
\sum_{j=1}^{p}\frac{(t-t_0)^{p-j}}{(p-j)!}(\boldsymbol{\mu}_\kappa)_j,
\end{align*}
and by covariance entries $C_{pq}(s,t) = C^{\mathrm{anc}}_{pq}(s,t)+C^{\kappa}_{pq}(s,t)$ given by
\begin{align*}
C^{\mathrm{anc}}_{pq}(s,t)
&= (\mathcal A_{1,t_0}^{[p]}\mathcal A_{2,t_0}^{[q]}C_\theta)(s,t),\\
C^{\kappa}_{pq}(s,t)
&= \mathbf{1}_{\{p>0\}}\mathbf{1}_{\{q>0\}}
\sum_{j=1}^{p}\sum_{m=1}^{q}
\frac{(s-t_0)^{p-j}}{(p-j)!}
\frac{(t-t_0)^{q-m}}{(q-m)!}
(\boldsymbol{\Sigma}_\kappa)_{jm},
\end{align*}
for \(p,q\in\{-r,\ldots,r\}\). Thus \(C^{\mathrm{anc}}_{pq}\) is the anchor-induced kernel part, while \(C^{\kappa}_{pq}\) is the finite-dimensional polynomial covariance induced by the integration constants. Under \cref{ass:independent-boundary} there is no cross-covariance between these two parts.
\end{maintheorem}

\begin{mainproof}
Write
\begin{align*}
A_p(t) := (\mathcal A_{t_0}^{[p]}f_0)(t),
\qquad
B_p(t) := \mathbf{1}_{\{p>0\}}\sum_{j=1}^{p}\frac{(t-t_0)^{p-j}}{(p-j)!}\kappa_j,
\end{align*}
so that \(f_p(t)=A_p(t)+B_p(t)\). Standard mean-square differentiation theory under \cref{ass:ms-regularity} gives the derivative processes \(D_t^k f_0\), \(k=1,\ldots,r\), with
\begin{align*}
\E[D_t^k f_0(t)] = D_t^k\mu(t),
\qquad
\Cov(D_s^a f_0(s),D_t^b f_0(t))=(D_1^aD_2^b C_\theta)(s,t),
\end{align*}
for the relevant orders \(a,b,k\le r\) \citep[Theorem~2.5.1]{cramer1967stationary}. In particular, the components \(f_{-r},\ldots,f_{-1}\) are square-integrable, and \(f_{-r},\ldots,f_0\) satisfy the chain relations for \(p=-r+1,\ldots,0\).

Since \(\mu\) and \(C_\theta\) are continuous on compact sets,
\begin{align*}
M:=\sup_{u\in[t_0,t_1]}\E[f_0(u)^2]
=\sup_{u\in[t_0,t_1]}\left(\mu(u)^2+C_\theta(u,u)\right)<\infty.
\end{align*}
This gives \(A_0=f_0\in L^2\). For \(p>0\), Cauchy--Schwarz gives, for every \(t\in[t_0,t_1]\),
\begin{align*}
\E[A_p(t)^2]
&\le
\frac{1}{((p-1)!)^2}
\left(\int_{t_0}^t (t-u)^{2p-2}\,du\right)
\left(\int_{t_0}^t \E[f_0(u)^2]du\right)<\infty.
\end{align*}
The term \(B_p(t)\) is a finite linear combination of the Gaussian variables \(\kappa_1,\ldots,\kappa_p\), and is therefore square-integrable. Hence every \(f_p(t)\) is square-integrable.

It remains to verify the positive-index part of the derivative chain. Let \(G_0=f_0\) and \(G_p=\mathcal I_{t_0}^p f_0\) for \(p\ge1\). The continuity of \(\mu\) and \(C_\theta\) implies that \(f_0\) is mean-square continuous. Moreover, the repeated integral identity
\begin{align*}
G_p(t)=\int_{t_0}^t G_{p-1}(u)\,du
\end{align*}
holds in \(L^2\). If \(G_{p-1}\) is mean-square continuous, then, for \(t\in(t_0,t_1)\) and \(h\) such that \(t+h\in[t_0,t_1]\),
\begin{align*}
\frac{G_p(t+h)-G_p(t)}{h}-G_{p-1}(t)
=\frac{1}{h}\int_t^{t+h}\left(G_{p-1}(u)-G_{p-1}(t)\right)\,du,
\end{align*}
and Jensen's inequality yields
\begin{align*}
\E\!\left[\left|\frac{G_p(t+h)-G_p(t)}{h}-G_{p-1}(t)\right|^2\right]
&\le
\frac{1}{|h|}\int_{\min(t,t+h)}^{\max(t,t+h)}
\E\!\left[|G_{p-1}(u)-G_{p-1}(t)|^2\right]du
\to0.
\end{align*}
Thus \(G_p\) is mean-square differentiable with derivative \(G_{p-1}\), and hence is mean-square continuous. Induction gives this for all \(p=1,\ldots,r\). Since \(B_p\) has ordinary derivative \(B_{p-1}\) for \(p\ge1\), with \(B_0\equiv0\), the full processes satisfy the mean-square derivative relation \(D_t f_p=f_{p-1}\) for \(p=1,\ldots,r\). Finally, \(G_p(t_0)=0\) and $f_p(t_0)=B_p(t_0)=\kappa_p$ for $p=1,\ldots,r$ almost surely, so the boundary identities hold.

We next prove joint Gaussianity. Fix \((p_1,\tau_1),\ldots,(p_n,\tau_n)\) with \(p_i\in\{-r,\ldots,r\}\) and \(\tau_i\in[t_0,t_1]\), and set \(\mathbf A=(A_{p_1}(\tau_1),\ldots,A_{p_n}(\tau_n))^\top\). Each coordinate of \(\mathbf A\) is an \(L^2\)-limit of finite linear functionals of finite evaluation vectors of \(f_0\): finite differences for negative indices, the process value itself for index zero, and Riemann sums for positive indices. Finite evaluation vectors of \(f_0\) are Gaussian, finite linear transformations preserve Gaussianity, and the \(L^2\)-limit is Gaussian by the Cram\'er--Wold argument. The same approximations show that \(\mathbf A\) may be chosen measurable with respect to the completed \(\sigma\)-field generated by \(f_0\). The corresponding vector \(\mathbf B=(B_{p_1}(\tau_1),\ldots,B_{p_n}(\tau_n))^\top\) is a linear transformation of \(\boldsymbol\kappa\), hence Gaussian, and it is independent of \(\mathbf A\) because \(\boldsymbol\kappa\) is independent of \(f_0\). Therefore \(\mathbf A+\mathbf B\) is Gaussian. Since the finite collection was arbitrary, \(\mathbf f\) is a Gaussian process.

The mean formula follows by applying expectation to the same \(L^2\) limits for negative indices and by Fubini's theorem for positive indices: $\E[A_p(t)] = (\mathcal A_{t_0}^{[p]}\mu)(t)$ for $p=-r,\ldots,r$. Adding \(\E[B_p(t)]\) gives
\begin{align*}
\E[f_p(t)]
&=
(\mathcal A_{t_0}^{[p]}\mu)(t)
+
\mathbf{1}_{\{p>0\}}
\sum_{j=1}^{p}\frac{(t-t_0)^{p-j}}{(p-j)!}(\boldsymbol\mu_\kappa)_j.
\end{align*}
If \(\mu\equiv0\) and \(\boldsymbol\mu_\kappa=\mathbf0\), this mean is zero for every component and time point; together with square-integrability, the derivative chain, and the boundary identities proved above, this gives a centered \(r\)th-order differential ensemble in the sense of \cref{def:ensemble}.

It remains to compute the covariance. Independence of \(f_0\) and \(\boldsymbol\kappa\) gives
\begin{align*}
\Cov(f_p(s),f_q(t))=\Cov(A_p(s),A_q(t))+\Cov(B_p(s),B_q(t)).
\end{align*}
The polynomial part is
\begin{align*}
\Cov(B_p(s),B_q(t))
&=
\mathbf{1}_{\{p>0\}}\mathbf{1}_{\{q>0\}}
\sum_{j=1}^{p}\sum_{m=1}^{q}
\frac{(s-t_0)^{p-j}}{(p-j)!}
\frac{(t-t_0)^{q-m}}{(q-m)!}
(\boldsymbol\Sigma_\kappa)_{jm}.
\end{align*}
For the anchor part, the relevant kernels are continuous on compact sets, so the covariance--Fubini interchanges below are justified by the \(L^2\) bounds above. The standard derivative covariance formula gives the claim when \(p,q\le0\). When \(p>0\) and \(q\le0\), continuity of \(D_2^{-q}C_\theta\) and the \(L^2\)-integral definition give
\begin{align*}
\Cov(A_p(s),A_q(t))
&=\frac{1}{(p-1)!}\int_{t_0}^s (s-u)^{p-1}
\Cov(f_0(u),D_t^{-q}f_0(t))\,du\\
&=\frac{1}{(p-1)!}\int_{t_0}^s (s-u)^{p-1}(D_2^{-q}C_\theta)(u,t)\,du\\
&=(\mathcal I_{1,t_0}^pD_2^{-q}C_\theta)(s,t).
\end{align*}
The case \(p\le0\) and \(q>0\) similarly gives
\begin{align*}
\Cov(A_p(s),A_q(t))
&=\frac{1}{(q-1)!}\int_{t_0}^t (t-v)^{q-1}(D_1^{-p}C_\theta)(s,v)\,dv\\
&=(D_1^{-p}\mathcal I_{2,t_0}^qC_\theta)(s,t).
\end{align*}
Finally, when \(p,q>0\), continuity of \(C_\theta\) and Fubini's theorem give
\begin{align*}
\Cov(A_p(s),A_q(t))
&=\frac{1}{(p-1)!(q-1)!}
\int_{t_0}^s\int_{t_0}^t
(s-u)^{p-1}(t-v)^{q-1}C_\theta(u,v)\,dv\,du\\
&=(\mathcal I_{1,t_0}^p\mathcal I_{2,t_0}^qC_\theta)(s,t).
\end{align*}
All four cases are exactly
\begin{align*}
\Cov(A_p(s),A_q(t))=(\mathcal A_{1,t_0}^{[p]}\mathcal A_{2,t_0}^{[q]}C_\theta)(s,t).
\end{align*}
Combining this identity with the polynomial covariance gives the stated covariance function.
\end{mainproof}

\cref{prop:gaussian-ensemble} gives the inferential construction used in the rest of the paper, as it determines all joint finite-dimensional Gaussian distributions of the derivative--anchor--integral state with differential relations understood in the weak mean-square sense. To realize the same family as a literal chain of derivatives and repeated integrals holding simultaneously for all \(t\), we next impose pathwise regularity.

\begin{assumption}[Pathwise regularity]\label{ass:pathwise-regularity}
In addition to \cref{ass:ms-regularity}, assume that the anchor process has a version \(f_0^\ast\) whose sample paths lie in \(C^r([t_0, t_1])\) almost surely.
\end{assumption}

\begin{mainproposition}[Pathwise differential chain]\label{prop:samplepaths}
Assume \cref{ass:ms-regularity,ass:pathwise-regularity}, and let all processes be defined on a probability space \((\Omega,\mathcal F,\mathbb P)\). Let \(f_0^\ast\) be the \(C^r\) version of \(f_0\) from \cref{ass:pathwise-regularity}, and let \(\mathfrak f=(f_{-r},\ldots,f_r)\) be defined by the displayed construction in \cref{prop:gaussian-ensemble} with square-integrable boundary constants \(\boldsymbol\kappa\). Under these assumptions, \(D_t^k f_0^\ast\) is a version of the mean-square derivative process \(f_{-k}\) for each \(k=1,\ldots,r\). Assume further that \(f_0\) and \(f_0^\ast\) are jointly measurable on \([t_0,t_1]\times\Omega\) with
\[
\int_{t_0}^{t_1}\E\bigl[|f_0(u)|+|f_0^\ast(u)|\bigr]du<\infty.
\]
Then there are versions of \(f_{-r},\ldots,f_r\), denoted by the same symbols, and an event \(\Omega_\ast\) with \(\mathbb P(\Omega_\ast)=1\) such that for every \(\omega\in\Omega_\ast\), every \(p\in\{-r,\ldots,r\}\), and every \(t\in[t_0,t_1]\),
\begin{align*}
f_p(t,\omega)=(\mathcal A_{t_0}^{[p]}f_0^\ast)(t,\omega)+\mathbf 1_{\{p>0\}}\sum_{j=1}^{p}\frac{(t-t_0)^{p-j}}{(p-j)!}\kappa_j(\omega).
\end{align*}
For each \(p=-r+1,\ldots,r\), the map \(t\mapsto f_p(t,\omega)\) is differentiable on \((t_0,t_1)\) with derivative \(f_{p-1}(t,\omega)\). Thus, when \(\boldsymbol\kappa\) satisfies \cref{ass:independent-boundary}, the weak Gaussian law from \cref{prop:gaussian-ensemble} can be represented by a literal sample-path differential chain under these additional hypotheses.
\end{mainproposition}

\begin{mainproof}
Let \(\Omega_C\) be a probability-one event on which \(f_0^\ast(\cdot,\omega)\in C^r([t_0,t_1])\). Define pathwise candidate versions, with arbitrary values assigned off \(\Omega_C\) for negative indices, by
\begin{align*}
\tilde f_p(t)=
\begin{cases}
(\mathcal A_{t_0}^{[p]}f_0^\ast)(t), & p\le 0,\\
(\mathcal A_{t_0}^{[p]}f_0^\ast)(t)+\displaystyle\sum_{j=1}^{p}\frac{(t-t_0)^{p-j}}{(p-j)!}\kappa_j, & p>0.
\end{cases}
\end{align*}
First identify the negative-index versions. Fix \(k=1,\ldots,r\) and \(t\in[t_0,t_1]\). Under \cref{ass:ms-regularity}, \(f_{-k}(t)\) is the \(L^2\) limit of \(k\)th difference quotients of \(f_0\), using one-sided quotients at endpoints. Replacing \(f_0\) by its version \(f_0^\ast\) does not change these quotients almost surely for each fixed step. On \(\Omega_C\), the same quotients converge to \(D_t^k f_0^\ast(t)\). Uniqueness of limits in probability gives \(D_t^k f_0^\ast(t)=f_{-k}(t)\) almost surely. Hence \(D_t^k f_0^\ast\) is a version of \(f_{-k}\). For \(p=0\), \(\tilde f_0=f_0^\ast\) is a version of \(f_0\). For \(p=-k<0\), the preceding argument shows that \(D_t^k f_0^\ast=(\mathcal A_{t_0}^{[-k]}f_0^\ast)\) is a version of \(f_{-k}\).

It remains to identify the positive-index candidates with the \(L^2\)-defined processes in \cref{prop:gaussian-ensemble}. Fix \(p>0\) and \(t\in[t_0,t_1]\), and set
\begin{align*}
K_{p,t}(u):=\frac{(t-u)^{p-1}}{(p-1)!}\mathbf 1_{[t_0,t]}(u).
\end{align*}
Let
\begin{align*}
J_t:=\int_{t_0}^{t_1}K_{p,t}(u)f_0(u)\,du,
\qquad
J_t^\ast:=\int_{t_0}^{t_1}K_{p,t}(u)f_0^\ast(u)\,du.
\end{align*}
The joint-measurability and integrability assumptions make these samplewise integrals well defined. Since \(K_{p,t}\) is bounded and \(\sup_{u\in[t_0,t_1]}\E[f_0(u)^2]<\infty\) under \cref{ass:ms-regularity}, \(J_t\in L^2(\Omega)\). For every \(Z\in L^2(\Omega)\), Fubini's theorem and Cauchy--Schwarz yield
\begin{align*}
\E[ZJ_t]
&=\int_{t_0}^{t_1}K_{p,t}(u)\E[Zf_0(u)]\,du.
\end{align*}
Thus the samplewise integral \(J_t\) represents the weak \(L^2\)-integral \(\bigl(\mathcal I_{t_0}^{p}f_0\bigr)(t)\) used in \cref{prop:gaussian-ensemble}.

Because \(f_0^\ast\) is a version of \(f_0\), Tonelli's theorem gives
\begin{align*}
\E|J_t^\ast-J_t|
&\le \E\int_{t_0}^{t_1}|K_{p,t}(u)|\,|f_0^\ast(u)-f_0(u)|\,du\\
&=\int_{t_0}^{t_1}|K_{p,t}(u)|\,\E|f_0^\ast(u)-f_0(u)|\,du
=0.
\end{align*}
Hence \(J_t^\ast=J_t\) almost surely. Adding the same polynomial in \(\kappa_1,\ldots,\kappa_p\) shows that, for this fixed \(t\), \(\tilde f_p(t)\) equals the \(f_p(t)\) constructed in \cref{prop:gaussian-ensemble} almost surely. Therefore each \(\tilde f_p\) is a version of \(f_p\).

Replace \(f_p\) by \(\tilde f_p\) for all \(p=-r,\ldots,r\) and drop the tildes. On the full-probability event \(\Omega_\ast:=\Omega_C\), the displayed pathwise representation in the proposition holds for every \(p\) and every \(t\in[t_0,t_1]\) by construction. It remains only to check the pathwise derivative chain on this event. If \(p\le0\) and \(p\ge-r+1\), ordinary differentiation of the \(C^r\) path \(f_0^\ast(\cdot,\omega)\) gives
\[
\frac{d}{dt}(\mathcal A_{t_0}^{[p]}f_0^\ast)(t,\omega)=(\mathcal A_{t_0}^{[p-1]}f_0^\ast)(t,\omega).
\]
If \(p>0\), the fundamental theorem of calculus gives
\[
\frac{d}{dt}(\mathcal A_{t_0}^{[p]}f_0^\ast)(t,\omega)=(\mathcal A_{t_0}^{[p-1]}f_0^\ast)(t,\omega),
\]
where \(p=1\) gives \(f_0^\ast(t,\omega)\). Differentiating the polynomial boundary term gives
\[
\frac{d}{dt}\sum_{j=1}^{p}\frac{(t-t_0)^{p-j}}{(p-j)!}\kappa_j(\omega)
=\mathbf 1_{\{p>1\}}\sum_{j=1}^{p-1}\frac{(t-t_0)^{p-1-j}}{(p-1-j)!}\kappa_j(\omega),
\]
which is the boundary term in \(f_{p-1}\). At \(t=t_0\), the integral part vanishes and the same polynomial term gives \(f_p(t_0,\omega)=\kappa_p(\omega)\) for every \(p=1,\ldots,r\). Combining the operator and polynomial identities proves \(\frac{d}{dt}f_p(t,\omega)=f_{p-1}(t,\omega)\) for every \(p=-r+1,\ldots,r\) and \(t\in(t_0,t_1)\).
\end{mainproof}

\subsection{Spectral Approximation via Laplacian Eigenfunctions}\label{sec:finite-rank}
The construction in \cref{prop:gaussian-ensemble,prop:samplepaths} expresses each kernel-driven covariance entry in terms of the operators \(\mathcal A_{1,t_0}^{[p]}\mathcal A_{2,t_0}^{[q]}\) applied to the anchor covariance function \(C_\theta\).  For standard covariance functions such as the exponential quadratic or Mat\'ern families, the transforms can in principle be written in closed form, but the resulting expressions become cumbersome as the ensemble order \(r\) increases and need not be available in a form useful for implementation. Our proposal for a general construction is to replace these covariance operators by the Hilbert space Gaussian process (HSGP) approximation of \citet{solin2020hilbert}, which depends on \(C_\theta\) only through its spectral density and represents all required transformations through precomputable Laplacian basis functions. For any fixed finite order supported by the regularity assumptions, increasing the derivative--integral ladder changes the transformed basis, but it does not require new covariance-specific derivative or integral calculus.

Assume from here on that $[t_0,t_1]\subseteq[-L,L]$ where we refer to $[-L,L]$ as the computational domain. Let $C_\theta(s,t)=C_\theta(\tau)$ with $\tau=s-t$ be a stationary covariance function with spectral density
\[
S_\theta(\omega)=\int_{-\infty}^{\infty} C_\theta(\tau)e^{-\mathrm{i}\omega\tau}\,\mathrm{d}\tau.
\]
Following \citet{solin2020hilbert} we define the Hilbert space approximation of the anchor covariance function on $[-L,L]$ by
\begin{align*}
\widetilde{C}_{\theta,L}(s,t) := \sum_{k=1}^\infty S_\theta\left(\sqrt{\lambda_{k,L}}\right)\phi_{k,L}(s)\phi_{k,L}(t),
\end{align*}
where $\{(\lambda_{k,L}, \phi_{k,L})\}_{k \geq 1}$ is the orthonormal eigensystem of $-\mathrm d^2/\mathrm dt^2$ on $[-L,L]$ with Dirichlet boundary conditions $\phi_{k,L}(-L) = \phi_{k,L}(L) = 0$ given by
\begin{align*}
\lambda_{k,L} = \left(\frac{k\pi}{2L}\right)^2, \qquad
\phi_{k,L}(t) = \sqrt{\frac{1}{L}} \sin\left(\frac{k\pi (t + L)}{2L}\right),
\qquad k=1,2,\ldots.
\end{align*}
Truncating the sum at $K < \infty$ leads to the rank-reduced covariance and series representation of a Gaussian process
\begin{align*}
\widetilde{C}_{\theta,L,K}(s,t)
:=
\sum_{k=1}^K
S_\theta\left(\sqrt{\lambda_{k,L}}\right)\phi_{k,L}(s)\phi_{k,L}(t), \qquad \widetilde{f}_{L,K}(t) := \sum_{k=1}^K \zeta_k S_\theta\left(\sqrt{\lambda_{k,L}}\right)^{1/2}\phi_{k,L}(t)
\end{align*}
with $\zeta_k \overset{\indep}{\sim} N(0,1)$. For posterior computation, this finite-basis representation can be implemented either by sampling the coefficients \(\zeta_1,\ldots,\zeta_K\) directly or by integrating them out analytically. The Laplacian--Dirichlet eigenbasis implementation in Stan \citep{stan2026}, the marginalized likelihood, and posterior reconstruction are described in \cref{appendix:marginal-finite-basis-likelihood} in the Supplementary Material. The interval \([-L,L]\) is a computational domain and an enlargement of the observation domain. Its role is to place the artificial Dirichlet boundary away from the observation window while providing a countable basis that does not depend on the functional form of the covariance kernel nor its parameters. Dependence on the covariance function enters only through the spectral weights \(S_\theta(\sqrt{\lambda_{k,L}})\), whereas the time dependence is carried entirely by the precomputable eigenfunctions \(\phi_{k,L}\). Similarly, the lower limit of integration remains the observation domain reference point \(t_0\), not the computational boundary \(-L\).

Fix $K \in \mathbb N$, $p,q \in \{-r,\ldots,r\}$, and define $\psi_k^{[a]}(t):=(\mathcal A_{t_0}^{[a]}\phi_{k,L})(t)$ for $a\in\{p,q\}$ and $k=1,\ldots,K$. Then
\begin{align*}
(\mathcal A_{t_0}^{[p]} \widetilde f_{L,K})(t)
&=
\sum_{k=1}^K
\zeta_k S_\theta\left(\sqrt{\lambda_{k,L}}\right)^{1/2}\psi_k^{[p]}(t), \quad \widetilde C_{\theta,L,K}^{[p,q]}(s,t) = \sum_{k=1}^K
S_\theta\left(\sqrt{\lambda_{k,L}}\right)\psi_k^{[p]}(s)\psi_k^{[q]}(t).
\end{align*}
The finite-rank construction now separates into two reusable ingredients. The covariance family enters only through the spectral weights \(S_\theta(\sqrt{\lambda_{k,L}})\), while the derivative or integral level enters through the transformed basis rows \(\psi_k^{[p]}\). Thus all covariance blocks can be assembled from the same spectral weights once the signed-order transforms of the Laplacian--Dirichlet eigenfunctions are available. The next proposition records these transforms explicitly.

\begin{mainproposition}[Transformed eigenfunctions]\label{prop:integrals}
For \(L>0\) and \(k\in\mathbb N\), write
\[
\omega_{k,L}:=\frac{k\pi}{2L},
\qquad
\phi_{k,L}(t)=L^{-1/2}\sin\bigl(\omega_{k,L}(t+L)\bigr).
\]
For every \(m\in\mathbb Z\), the transformed eigenfunction
\(\psi_k^{[m]}(t):=(\mathcal A_{t_0}^{[m]}\phi_{k,L})(t)\) is given by
\[
\psi_k^{[m]}(t)=
\begin{cases}
L^{-1/2}\omega_{k,L}^{-m}
\sin\bigl(\omega_{k,L}(t+L)-m\pi/2\bigr),
& m\le0,\\[5pt]
\displaystyle
\frac{L^{-1/2}}{\omega_{k,L}}
\left[
\cos\bigl(\omega_{k,L}(t_0+L)\bigr)-\cos\bigl(\omega_{k,L}(t+L)\bigr)
\right],
& m=1,\\[8pt]
\begin{aligned}
&\displaystyle
\frac{(t-t_0)^{m-1}L^{-1/2}}{(m-1)!\,\omega_{k,L}}
\cos\bigl(\omega_{k,L}(t_0+L)\bigr)\\
&\displaystyle\quad+
\frac{(t-t_0)^{m-2}L^{-1/2}}{(m-2)!\,\omega_{k,L}^{2}}
\sin\bigl(\omega_{k,L}(t_0+L)\bigr)
-\omega_{k,L}^{-2}\psi_k^{[m-2]}(t),
\end{aligned}
& m=2,3,\ldots.
\end{cases}
\]
\end{mainproposition}

\begin{mainproof}
The derivative identity follows by differentiating the sine function \(a\) times:
\[
\frac{d^a}{dt^a}\sin\{\omega_{k,L}(t+L)\}
=
\omega_{k,L}^{a}\sin\{\omega_{k,L}(t+L)+a\pi/2\}.
\]
Multiplying by \(L^{-1/2}\) gives the displayed formula for signed orders \(m=-a\le0\), including \(m=0\).

For positive signed orders, direct integration gives
\begin{align*}
(\mathcal I_{t_0}^{1}\phi_{k,L})(t)
&=
L^{-1/2}\int_{t_0}^t \sin\{\omega_{k,L}(s+L)\}\,\mathrm ds\\
&=
\frac{L^{-1/2}}{\omega_{k,L}}
\left[
\cos\{\omega_{k,L}(t_0+L)\}-\cos\{\omega_{k,L}(t+L)\}
\right],
\end{align*}
since \(L>0\) and \(k\in\mathbb N\) imply \(\omega_{k,L}\neq0\).

It remains to prove the recursion for \(q\ge2\). Write \(f=\phi_{k,L}\) and \(\omega=\omega_{k,L}\). Then \(f''=-\omega^2 f\), and Cauchy's formula for repeated integration gives
\[
(\mathcal I_{t_0}^q f'')(t)
=
\frac{1}{(q-1)!}\int_{t_0}^t (t-s)^{q-1}f''(s)\,\mathrm ds.
\]
Two integrations by parts yield
\begin{align*}
(\mathcal I_{t_0}^q f'')(t)
&=
-\frac{(t-t_0)^{q-1}}{(q-1)!}f'(t_0)
+
\frac{1}{(q-2)!}\int_{t_0}^t (t-s)^{q-2}f'(s)\,\mathrm ds\\
&=
(\mathcal I_{t_0}^{q-2}f)(t)
-
\frac{(t-t_0)^{q-2}}{(q-2)!}f(t_0)
-
\frac{(t-t_0)^{q-1}}{(q-1)!}f'(t_0),
\end{align*}
where the second equality also covers \(q=2\), with \(\mathcal I_{t_0}^0\) interpreted as the identity. Since \(f''=-\omega^2 f\),
\[
(\mathcal I_{t_0}^q f)(t)
=
-\omega^{-2}(\mathcal I_{t_0}^{q-2}f)(t)
+
\frac{(t-t_0)^{q-2}}{(q-2)!\,\omega^2}f(t_0)
+
\frac{(t-t_0)^{q-1}}{(q-1)!\,\omega^2}f'(t_0).
\]
Finally,
\[
f(t_0)=L^{-1/2}\sin\{\omega_{k,L}(t_0+L)\},
\qquad
f'(t_0)=L^{-1/2}\omega_{k,L}\cos\{\omega_{k,L}(t_0+L)\}.
\]
Substituting these two endpoint values and returning to \(f=\phi_{k,L}\), \(\omega=\omega_{k,L}\), and \(\psi_k^{[q]}=\mathcal A_{t_0}^{[q]}\phi_{k,L}\) gives the displayed recursion for every \(q\ge2\).
\end{mainproof}

It remains to quantify what is lost by replacing the exact stationary kernel with its truncated Laplacian--Dirichlet expansion on the enlarged computational interval. The most stringent requirements arise for derivative blocks, since differentiation weights the spectrum by powers of frequency. Once those blocks are controlled, the integral blocks follow by bounded repeated integration on \([t_0,t_1]\), and the integration constant covariance is carried exactly. The following condition and bounds make this operator-level error control explicit.

\begin{assumption}[Derivative-weighted spectral regularity]\label{ass:subsection-weighted-spectral}
For each pair \((a,b)\) to which the approximation bounds below are applied, the weighted spectral density \(S_\theta^{[a,b]}(\omega):=\omega^{a+b}S_\theta(\omega)\), \(\omega>0\), is uniformly bounded on \((0,\infty)\), is differentiable with bounded derivative on \((0,\infty)\), and both \(S_\theta^{[a,b]}\) and the derivative of \(S_\theta^{[a,b]}\) with respect to \(\omega\) are integrable on \((0,\infty)\).
\end{assumption}

\Cref{prop:spectral-approximation-bounds} presents two bounds. The first is the one-dimensional HSGP approximation bound of \citet{solin2020hilbert} where the error is controlled by a finite-domain term of order \(L^{-1}\) and the spectral tail beyond the largest retained frequency. The second is the extension needed here because the differential ensemble is built from derivative covariance blocks, not only from the anchor covariance. Differentiating \(a\) times in the first argument and \(b\) times in the second argument weights the spectral density by \(\omega^{a+b}\), so the same argument must be applied to \(S_\theta^{[a,b]}\) to obtain uniform control of the derivative blocks used by the ensemble construction.

\begin{maintheorem}[Spectral approximation of derivative covariance blocks]\label{prop:spectral-approximation-bounds}
Fix \(\Lambda>0\) such that $[t_0,t_1]\subset[-\Lambda,\Lambda]$, and let \(L>\Lambda\). Assume that \(C_\theta\) is the restriction to \([t_0,t_1]^2\) of a real stationary covariance function on \(\mathbb R\), with spectral density \(S_\theta\). Assume also that \(S_\theta\) is uniformly bounded on \((0,\infty)\), is differentiable with bounded derivative on \((0,\infty)\), and that both \(S_\theta\) and the derivative of \(S_\theta\) with respect to \(\omega\) are integrable on \((0,\infty)\). Then there exists a constant \(E<\infty\), depending at most on \(S_\theta\) and \(\Lambda\) but independent of \(K\), \(L\), \(s\), and \(t\), such that
\begin{align*}
\sup_{(s,t)\in [t_0,t_1]^2}
\left|
C_\theta(s,t)-\widetilde C_{\theta,L,K}(s,t)
\right|
\le
\frac{E}{L}
+
\frac{1}{\pi}
\int_{\pi K/(2L)}^\infty S_\theta(\omega)\,\mathrm d\omega.
\end{align*}
Moreover, if \cref{ass:subsection-weighted-spectral} holds for a fixed pair \((a,b)\in \mathbb N_0^2\), then there exists a constant \(E_{a,b}<\infty\), depending at most on \(S_\theta\), \((a,b)\), and \(\Lambda\) but independent of \(K\), \(L\), \(s\), and \(t\), such that
\begin{align*}
\sup_{(s,t)\in [t_0,t_1]^2}
\left|
D_1^a D_2^b C_\theta(s,t)
-
D_1^a D_2^b \widetilde C_{\theta,L,K}(s,t)
\right|
\le
\frac{E_{a,b}}{L}
+
\frac{1}{\pi}
\int_{\pi K/(2L)}^\infty \omega^{a+b}S_\theta(\omega)\,\mathrm d\omega.
\end{align*}
Consequently, $\widetilde C_{\theta,L,K}\to C_\theta$ uniformly on $[t_0,t_1]^2$ as \(K,L\to\infty\) with \(K/L\to\infty\), and, for the same fixed pair \((a,b)\) for which \cref{ass:subsection-weighted-spectral} holds, $D_1^a D_2^b \widetilde C_{\theta,L,K} \to D_1^a D_2^b C_\theta$ uniformly on $[t_0,t_1]^2$.
\end{maintheorem}

\begin{mainproof}
We prove a common estimate for the anchor case and for the differentiated case. Fix nonnegative integers $a,b$ and put $h:=\frac{\pi}{2L}$ and $\omega_{k,L}=kh$ with $a=b=0$ and $g=S_\theta$ for the anchor-kernel bound, and with $g(\omega)=\omega^{a+b}S_\theta(\omega)$ and $\omega>0$ for the derivative-block bound. In the latter case, the required boundedness, differentiability, and integrability properties of $g$ are exactly those in \cref{ass:subsection-weighted-spectral}; in the anchor case they are the explicit assumptions on $S_\theta$. Since $S_\theta$ is the spectral density of a real stationary covariance function, it is even and nonnegative, and hence so is the weight $g$ on $(0,\infty)$. All constants below may depend on \(g\), \((a,b)\), and the fixed envelope \(\Lambda\), but not on \(K\), \(L\), \(s\), or \(t\).

The assumptions imply that $g$ has a continuous one-sided extension at zero and is absolutely continuous on compact subintervals of $[0,\infty)$. Moreover,
\[
    \|g\|_\infty<\infty,\qquad
    \int_0^\infty g(\omega)\,\mathrm d\omega<\infty,
    \qquad
    \int_0^\infty |g'(\omega)|\,\mathrm d\omega<\infty .
\]
Under the Fourier inversion representation of the derivative block, dominated convergence gives
\[
D_1^aD_2^b C_\theta(s,t)
=
\frac{1}{2\pi}\int_{\mathbb R}
(\mathrm i\omega)^a(-\mathrm i\omega)^b e^{\mathrm i\omega(s-t)}S_\theta(\omega)\,\mathrm d\omega.
\]
Using the evenness of $S_\theta$, this becomes
\[
F_{a,b}(s,t)
:=
\frac{1}{\pi}\int_0^\infty
    g(\omega)
    \cos\!\left(\omega(s-t)+\frac{(a-b)\pi}{2}\right)
\,\mathrm d\omega .
\]
For $a=b=0$ this is simply the Fourier representation of $C_\theta(s,t)$. Thus \(F_{a,b}\) is the exact derivative covariance block represented in half-line cosine form.

Since the series defining $\widetilde C_{\theta,L,K}$ is finite, termwise differentiation is immediate. By \cref{prop:integrals},
\[
D_t^a\phi_{k,L}(t)
=
L^{-1/2}\omega_{k,L}^a
\sin\!\left(\omega_{k,L}(t+L)+\frac{a\pi}{2}\right),
\]
and the product-to-sum identity gives
\[
D_1^aD_2^b\widetilde C_{\theta,L,K}(s,t)
=
A_{L,K}^{[a,b]}(s,t)-B_{L,K}^{[a,b]}(s,t),
\]
where
\[
A_{L,K}^{[a,b]}(s,t)
:=
\frac{h}{\pi}\sum_{k=1}^K
    g(kh)
    \cos\!\left(kh(s-t)+\frac{(a-b)\pi}{2}\right)
\]
and, since $2Lh=\pi$,
\[
B_{L,K}^{[a,b]}(s,t)
:=
\frac{h}{\pi}\sum_{k=1}^K
    (-1)^k g(kh)
    \cos\!\left(kh(s+t)+\frac{(a+b)\pi}{2}\right).
\]

It remains to bound the Riemann sum error in $A_{L,K}^{[a,b]}$ and the reflected term $B_{L,K}^{[a,b]}$. The fixed compact-domain envelope is used twice: the main Riemann term is controlled through \(|s-t|\), while the reflected alternating term is controlled through \(|s+t|\). Let $\tau=s-t$ and define
\[
    r_\tau(\omega)
    :=g(\omega)\cos\!\left(\omega\tau+\frac{(a-b)\pi}{2}\right).
\]
Because $(s,t)\in[t_0,t_1]^2\subset[-\Lambda,\Lambda]^2$, we have $|\tau|\le 2\Lambda$, and hence
\[
\int_0^\infty |r_\tau'(\omega)|\,\mathrm d\omega
\le
\int_0^\infty |g'(\omega)|\,\mathrm d\omega
+2\Lambda\int_0^\infty g(\omega)\,\mathrm d\omega
=:M_{g,\Lambda}<\infty .
\]
For each $K$,
\begin{align*}
\left|F_{a,b}(s,t)-A_{L,K}^{[a,b]}(s,t)\right|
&\le
\frac{1}{\pi}
\left|\int_0^{Kh} r_\tau(\omega)\,\mathrm d\omega
      -h\sum_{k=1}^K r_\tau(kh)\right|
+
\frac{1}{\pi}\int_{Kh}^\infty g(\omega)\,\mathrm d\omega  \\
&\le
\frac{hM_{g,\Lambda}}{\pi}
+
\frac{1}{\pi}\int_{Kh}^\infty g(\omega)\,\mathrm d\omega .
\end{align*}
The second inequality is the elementary right endpoint Riemann sum bound obtained by applying absolute continuity on each interval $[(k-1)h,kh]$.

For the reflected term, set $x=s+t$ and
\[
    q_x(\omega)
    :=g(\omega)\cos\!\left(\omega x+\frac{(a+b)\pi}{2}\right).
\]
Again $|x|\le 2\Lambda$, so $\int_0^\infty |q_x'|\le M_{g,\Lambda}$ and $\|q_x\|_\infty\le\|g\|_\infty$. Pair consecutive terms in the alternating sum, with \(N=\lfloor K/2\rfloor\). If \(K\) is even, all terms are paired. If \(K\) is odd, one endpoint remains and is bounded by \(\|g\|_\infty\). Hence
\[
\left|\sum_{k=1}^K(-1)^k q_x(kh)\right|
\le
\sum_{j=1}^N |q_x(2jh)-q_x((2j-1)h)|+\|g\|_\infty
\le
M_{g,\Lambda}+\|g\|_\infty .
\]
Therefore
\[
\sup_{(s,t)\in[t_0,t_1]^2}
\left|B_{L,K}^{[a,b]}(s,t)\right|
\le
\frac{h}{\pi}(M_{g,\Lambda}+\|g\|_\infty)
\le
\frac{C_{g,\Lambda}}{L},
\]
where $C_{g,\Lambda}<\infty$ is independent of $K$, $L$, $s$, and $t$.

Combining the two estimates and using $Kh=\pi K/(2L)$ gives
\[
\sup_{(s,t)\in[t_0,t_1]^2}
\left|F_{a,b}(s,t)-D_1^aD_2^b\widetilde C_{\theta,L,K}(s,t)\right|
\le
\frac{E_{a,b}}{L}
+
\frac{1}{\pi}\int_{\pi K/(2L)}^\infty g(\omega)\,\mathrm d\omega,
\]
for a finite constant $E_{a,b}$ that may depend on \(g\), \((a,b)\), and \(\Lambda\), but not on $K$, $L$, $s$, and $t$. The reflected-term bound is absorbed into the \(E_{a,b}/L\) term, while the remaining half-line spectral tail has coefficient \(1/\pi\). Taking $a=b=0$ gives the anchor-kernel bound, with $E=E_{0,0}$ and $g=S_\theta$. For a fixed pair $(a,b)$ satisfying \cref{ass:subsection-weighted-spectral}, the same display gives the derivative-block bound.

Finally, if $K,L\to\infty$ and $K/L\to\infty$, then $L^{-1}\to0$ and the lower limit $\pi K/(2L)$ in the tail integral tends to infinity. Since the relevant weight $g$ is integrable, the tail integral tends to zero. This proves both uniform convergence statements.
\end{mainproof}

The preceding approximation theorem also implies convergence for the whole finite ensemble. To see this, keep the assumptions of \cref{prop:spectral-approximation-bounds} and suppose that \cref{ass:subsection-weighted-spectral} holds for every pair \((a,b)\in\{0,\ldots,r\}^2\). For \(m\in\{-r,\ldots,r\}\), let \(a_m:=\max(-m,0)\) and \(i_m:=\max(m,0)\). The signed-order representation of the exact and finite-rank kernel-driven blocks gives their difference as \(\mathcal I_{1,t_0}^{i_p}\mathcal I_{2,t_0}^{i_q}\bigl(D_1^{a_p}D_2^{a_q}\widetilde C_{\theta,L,K}-D_1^{a_p}D_2^{a_q}C_\theta\bigr)\). Since repeated integration is a bounded map in the sup norm on the compact square, there is a finite constant \(M_r\), depending only on \(r\) and \(t_1-t_0\), such that
\[
\max_{p,q\in\{-r,\ldots,r\}}
\left\|
\widetilde C_{\theta,L,K}^{[p,q]}-C_\theta^{[p,q]}
\right\|_\infty
\le
M_r
\max_{0\le a,b\le r}
\left\|
D_1^aD_2^b\widetilde C_{\theta,L,K}
-
D_1^aD_2^bC_\theta
\right\|_\infty.
\]
The right-hand side converges to zero by \cref{prop:spectral-approximation-bounds}, uniformly over the finitely many derivative pairs. Consequently, as \(K,L\to\infty\) with \(K/L\to\infty\) and with the computational interval containing \([t_0,t_1]\),
\[
\max_{p,q\in\{-r,\ldots,r\}}
\sup_{(s,t)\in[t_0,t_1]^2}
\left|
\widetilde C_{\theta,L,K}^{[p,q]}(s,t)-C_\theta^{[p,q]}(s,t)
\right|
\to 0.
\]

The finite-rank construction enters posterior computation only through finite Gaussian mean vectors and covariance blocks. The preceding convergence statement gives sufficient conditions for the kernel-driven blocks to converge on any fixed collection of evaluation grids. The following proposition records the resulting finite-grid posterior convergence directly.

\begin{mainproposition}[Finite-grid posterior convergence]\label{prop:finite-grid-posterior-convergence}
Fix hyperparameters \(\vartheta=(\theta,\sigma^2,\boldsymbol\mu_\kappa,\boldsymbol\Sigma_\kappa)\) with \(\sigma^2>0\). Let \(Y=(Y_1,\ldots,Y_n)^\top\), with \(Y_i=f_0(x_i)+\epsilon_i\) at fixed points \(x_1,\ldots,x_n\), where \(\epsilon_i\overset{\indep}{\sim}N(0,\sigma^2)\) and the errors are independent of the ensemble, and let \(G\in\mathbb R^d\) be any fixed finite vector of ensemble evaluations or fixed finite linear combinations of them. For the exact model and a finite-rank approximation indexed by \(\ell\), write the joint laws in \(Y,G\) block order as
\[
\binom{Y}{G}\mid\vartheta
\sim
N\!\left(
\binom{\mu_Y}{\mu_G},
\bigl[\Sigma_{AB}\bigr]_{A,B\in\{Y,G\}}
\right),
\qquad
\binom{\widetilde Y^{(\ell)}}{\widetilde G^{(\ell)}}\mid\vartheta
\sim
N\!\left(
\binom{\widetilde\mu_Y^{(\ell)}}{\widetilde\mu_G^{(\ell)}},
\bigl[\widetilde\Sigma_{AB}^{(\ell)}\bigr]_{A,B\in\{Y,G\}}
\right),
\]
where the \(Y,Y\) blocks include the independent noise variance \(\sigma^2 I_n\). If \(\widetilde\mu_A^{(\ell)}\to\mu_A\) for \(A\in\{Y,G\}\) and \(\widetilde\Sigma_{AB}^{(\ell)}\to\Sigma_{AB}\) for all \(A,B\in\{Y,G\}\), then, for every fixed observed data vector \(y\in\mathbb R^n\), the conditional mean and covariance of \(\widetilde G^{(\ell)}\mid \widetilde Y^{(\ell)}=y\) satisfy
\[
\widetilde\mu_{G\mid Y}^{(\ell)}(y) \to
\mu_G+\Sigma_{GY}\Sigma_{YY}^{-1}(y-\mu_Y),\qquad
\widetilde\Sigma_{G\mid Y}^{(\ell)} \to \Sigma_{GG}-\Sigma_{GY}\Sigma_{YY}^{-1}\Sigma_{YG}.
\]
\end{mainproposition}

\begin{mainproof}
The \(Y,Y\) block is positive definite because it includes \(\sigma^2 I_n\) with \(\sigma^2>0\). In particular, \(\det(\Sigma_{YY})\neq0\). Since \(\widetilde\Sigma_{YY}^{(\ell)}\to\Sigma_{YY}\), continuity of the determinant implies that \(\det(\widetilde\Sigma_{YY}^{(\ell)})\neq0\) for all sufficiently large \(\ell\), and continuity of matrix inversion gives \((\widetilde\Sigma_{YY}^{(\ell)})^{-1}\to\Sigma_{YY}^{-1}\). Applying the finite-dimensional Gaussian conditioning formulas and using continuity of addition and matrix multiplication gives the stated limits.
\end{mainproof}

This convergence is conditional on the fixed hyperparameters \(\vartheta\). It concerns only the displayed finite grids and monitored vector \(G\) and does not assert continuum convergence, convergence for unmonitored targets, or convergence after hyperparameter integration. The conditioning argument only inverts the noisy observation block \(\Sigma_{YY}\), so singular or deterministic integration constants are allowed because \(\boldsymbol\Sigma_\kappa\) is not inverted.

The asymptotic result does not by itself choose a finite representation. The two approximation parameters \(L\) and \(K\) control different errors: \(L\) determines how far the artificial Dirichlet boundary is placed from the observation interval, while \(K\) determines the largest retained frequency. Practical HSGP calibration rules, such as those of \citet{riutort2023practical}, choose these quantities for a single Gaussian process. For a differential ensemble, however, calibrating only the anchor process \(f_0\) is not enough. Derivative levels amplify high frequencies, integral levels depend on the boundary and lower integration limit, and posterior summaries may depend on cross-covariances among derivative, anchor, and integral levels.

\subsection{Target-aware calibration of the Differential Ensemble}\label{sec:tartare-calibration}

TARTARE (Target-Aware Range and Truncation for Approximate Representations of Ensembles) is the corresponding extension of the calibration strategy of \citet{riutort2023practical} from a single Gaussian process to a differential ensemble. Rather than calibrating \(L\) and \(K\) only for the anchor covariance, TARTARE calibrates the finite-rank approximation against a user-specified set of monitored ensemble levels. Let \(\mathcal M\subseteq\{-r,\ldots,r\}\) be this monitoring set, where \(p\in\mathcal M\) denotes the ensemble element \(f_p\). For example, \(\mathcal M=\{0\}\) monitors the anchor, \(\mathcal M=\{-2\}\) monitors the second derivative, and \(\mathcal M=\{-2,-1,0,1,2\}\) monitors a full second-order ensemble.

TARTARE sets \(L\) and \(K\) through the normalized length-scale \(u=\rho/W\) (the covariance length-scale \(\rho\) relative to the half-width \(W\) of the observation window) and two calibration constants: a spectral constant \(m\), which fixes the retained frequency range and hence \(K\); and a range coefficient \(c_{\mathcal M}\), which fixes the distance to the artificial Dirichlet boundary and hence \(L\). These constants are chosen by comparing exact and HSGP covariance matrices on fixed calibration grids for the monitored levels. The exact Gaussian ensemble remains the target object. The polynomial covariance from integration constants is retained exactly and is not part of the offline covariance approximation. Offline calibration builds reusable entries indexed by kernel family, smoothness, ensemble order, and monitored set. Online use applies one such entry to a dataset, fits the model, and enlarges the finite representation until the reported posterior summaries are stable.

\paragraph{Offline calibration.}
Offline calibration is done on the normalized observation window with half-width \(W=1\). For a dataset with half-width \(W\), the normalized length-scale is \(u=\rho/W\), and an entry with spectral constant \(m\), range coefficient \(c_{\mathcal M}\), and floor \(c_{\mathrm{floor}}\) is used through
\[
u=\frac{\rho}{W},
\qquad
c=c(u)=\max\{c_{\mathrm{floor}},c_{\mathcal M}u\},
\qquad
L=cW,
\qquad
K=\left\lceil m\,\frac{c}{u}\right\rceil,
\qquad
\Omega_K=\frac{\pi K}{2L}.
\]
Here \(\rho\) is the covariance length-scale, \(W\) is the observation-window half-width, \(u\) is the normalized length-scale, \(c\) is the range factor, \(c_{\mathrm{floor}}\) is the minimum allowed range factor (\(c_{\mathrm{floor}}=1.2\) in this implementation), \(c_{\mathcal M}\) is the calibrated coefficient for the monitored set, \(m\) is the spectral constant, \(L\) is the computational half-width, \(K\) is the number of retained eigenfunctions, and \(\Omega_K\) is the angular frequency of the \(K\)th retained Laplacian--Dirichlet eigenfunction. The constant \(m\) controls the operator-weighted spectral tail, while \(c\) controls the distance from the observation window to the artificial Dirichlet boundary.

The spectral constant is tied to a calibration order \(q_\star\). For ensemble calibration entries, \(q_\star\) is set to the ensemble order \(r\). Thus the spectral cutoff is chosen for the highest derivative order represented by the ensemble, even when \(\mathcal M\) monitors only lower-order levels. Derivative-only entries use their derivative order. To choose the frequency cutoff, TARTARE measures how much of the derivative-weighted spectral mass lies above a candidate cutoff \(\Omega\). For a \(q\)th derivative, differentiation weights the spectral density by \(\omega^{2q}\) in the variance of that derivative. We therefore define the normalized tail fraction
\[
R_q(\Omega;\rho)
=
\frac{
\int_{\Omega}^{\infty} \omega^{2q}S_\theta(\omega;\rho)\,\mathrm d\omega
}{
\int_0^{\infty} \omega^{2q}S_\theta(\omega;\rho)\,\mathrm d\omega
},
\]
whenever the denominator is finite. Small \(R_q(\Omega;\rho)\) means that the retained frequencies up to \(\Omega\) capture nearly all of the spectral mass relevant for the \(q\)th derivative. Given a pre-specified tail tolerance \(\varepsilon_{\mathrm{tail}}\), the inverse-tail requirement for \(m\) is
\[
R_{q_\star}\!\left(\frac{\pi m}{2\rho};\rho\right)
\leq \varepsilon_{\mathrm{tail}},
\]
The stored value of \(m\) includes a conservative multiplier \(s_m\) used by the implementation with default \(s_m=1.25\). For the squared exponential kernel,
\[
R_q^{\mathrm{SE}}(\Omega;\rho)
=
P\!\left(\chi^2_{2q+1}>\rho^2\Omega^2\right),
\qquad
m_q^{\mathrm{SE}}(\varepsilon_{\mathrm{tail}})
=
\frac{2}{\pi}
\sqrt{\chi^2_{2q+1;\,1-\varepsilon_{\mathrm{tail}}}},
\]
so \(m_q^{\mathrm{SE}}(\varepsilon_{\mathrm{tail}})\) is the smallest spectral constant satisfying the tail tolerance. For a Mat\'ern kernel with smoothness \(\nu\), let \(z(\Omega)=2\nu/(2\nu+\rho^2\Omega^2)\). For admissible \(q<\nu\),
\[
R_{q,\nu}^{\mathrm{Mat}}(\Omega;\rho)
=
I_{z(\Omega)}\!\left(\nu-q,q+\frac12\right),
\]
where \(I_z(a,b)\) is the regularized incomplete beta function. If \(z_\varepsilon\) solves \(I_{z_\varepsilon}(\nu-q,q+1/2)=\varepsilon_{\mathrm{tail}}\), then
\[
m_{q,\nu}^{\mathrm{Mat}}(\varepsilon_{\mathrm{tail}})
=
\frac{2\sqrt{2\nu}}{\pi}
\sqrt{\frac{1-z_\varepsilon}{z_\varepsilon}}.
\]
This step is defined only for mean-square admissible derivative orders. Squared exponential kernels support every finite order. A Mat\'ern kernel with smoothness \(\nu\) supports order \(q\) only when \(q<\nu\), equivalently when \(\omega^{2q}S_\nu(\omega)\) is integrable \citep{gpml}. Near this boundary the weighted tail decays slowly, so calibration entries at the largest admissible Mat\'ern order are excluded from the joint ensemble table together with inadmissible higher orders.

After \(m\) has fixed the retained spectral range, the remaining calibration task is to choose the domain range factor \(c_{\mathcal M}\), which determines how far the artificial Dirichlet boundary is placed from the normalized observation window for the monitored ensemble levels. Fix \(\mathcal M=\{p_1,\ldots,p_J\}\) and a calibration grid \(\xi_1,\ldots,\xi_G\). The exact covariance matrix \(K_{\mathrm{ex}}(u)\) is built from the kernel-driven blocks $\left(\mathcal A_{1,t_0}^{[p]}\mathcal A_{2,t_0}^{[q]}C_\theta\right)(\xi_g,\xi_h)$ for $p,q\in\mathcal M$.
The HSGP matrix \(K_{\mathrm H}(u,c)\) is built from the finite Laplacian--Dirichlet expansion on \([-c,c]\), with the same \(\mathcal A\) operators applied to the basis functions as in \cref{prop:integrals}. Thus both matrices correspond to the stacked monitored evaluations $\left((\mathcal A_{t_0}^{[p_j]}f_0)(\xi_g):j=1,\ldots,J,\ g=1,\ldots,G\right)$.
The comparison includes all monitored same-level and cross-level covariance blocks. It excludes the covariance contribution of integration constants, which is added exactly in the model. With \(V_{\mathrm{diag}}\) denoting the diagonal matrix of exact marginal standard deviations for the stacked evaluations, the primary joint error is
\[
E_{\mathrm{joint}}(u,c)
=
\frac{
\left\|V_{\mathrm{diag}}^{-1}\left(K_{\mathrm{H}}(u,c)-K_{\mathrm{ex}}(u)\right)V_{\mathrm{diag}}^{-1}\right\|_{\mathrm F}
}{
\left\|V_{\mathrm{diag}}^{-1}K_{\mathrm{ex}}(u)V_{\mathrm{diag}}^{-1}\right\|_{\mathrm F}
},
\]
where \(\|\cdot\|_{\mathrm F}\) is the Frobenius norm. Secondary guards use the same blocks. \(E_{\mathrm{pair}}(u,c)\) is the largest normalized Frobenius error over monitored block pairs. \(E_{\mathrm{var}}(u,c)\) is the largest relative marginal-variance error on the calibration grid. If positive levels are monitored, \(E_{t_0}(u,c)\) is the one-sided marginal-variance error near the lower integration limit \(t_0=-W\). A candidate \(c\) is accepted at a given \(u\) only if these errors are below their prescribed tolerances. The search records the smallest acceptable \(c\) for each grid value of \(u\), and these minima are summarized by the conservative rule \(c(u)=\max\{c_{\mathrm{floor}},c_{\mathcal M}u\}\).

The offline check is finite and conditional. A completed calibration entry validates the calibration grid for the monitored set \(\mathcal M\). It is not a continuum approximation theorem and it does not guarantee accuracy for ensemble components outside \(\mathcal M\). This check is separate from \cref{prop:finite-grid-posterior-convergence}, which gives conditional posterior convergence once the required finite covariance blocks converge. The calibration constants are reported in \cref{tab:tartare-calibration-constants} in the Supplementary Material.

\paragraph{Online use.}
The online part applies a completed calibration entry to a dataset. The covariance length-scale \(\rho\) remains an inferred model parameter with its fixed prior. TARTARE uses posterior summaries of \(\rho\) only to choose the finite HSGP representation. Smaller values of \(\rho\) place more spectral mass at high frequencies and therefore require larger \(K\) to satisfy the derivative-weighted tail criterion. Phase A guards against posterior mass at such short length-scales, and Phase B checks whether the reported posterior summaries are stable under further basis enlargement.

\emph{Phase A: entry search.} All Phase A refits use the same prior on \(\rho\). Only \(K\) and \(L\) are changed. The first pass sets a working length-scale \(\rho_{\mathrm{work}}=\rho_{\mathrm{start}}\), with \(\rho_{\mathrm{start}}=0.5W\). The working value \(\rho_{\mathrm{work}}\) is not a parameter estimate. It is only the length-scale value used to design the next basis. For any current \(\rho_{\mathrm{work}}\), set \(u=\rho_{\mathrm{work}}/W\) and compute
\[
c=\max\{c_{\mathrm{floor}},c_{\mathcal M}u\},
\qquad
L=cW,
\qquad
K=\left\lceil m\,\frac{c}{u}\right\rceil .
\]
The model is then fitted with that prior. Let \(\rho_{\alpha}\) denote the posterior \(\alpha_\rho\)-quantile of \(\rho\), where \(\alpha_\rho\) is a small configured probability level, and let \(\delta_\rho\) be the configured length-scale tolerance. Phase A computes \(\ell_{\min}=mcW/K\), the smallest length-scale for which the current basis count attains the stored spectral-tail calibration, and requires \(\rho_{\alpha}-\delta_\rho\geq\ell_{\min}\). It also requires the operator-weighted spectral tail at \(\Omega_K\), evaluated at \(\rho_{\alpha}\), to be no larger than \(\varepsilon_{\mathrm{tail}}\). If either check fails, the next basis is designed with \(\rho_{\mathrm{work}}=\rho_{\alpha}\), and the same model with the same prior is refit with the enlarged representation. The loop stops when both checks hold or when the Phase A iteration cap is reached.

\emph{Phase B: diagnostic refinement.} Phase B starts from the accepted Phase A fit. Each refinement increases \(K\), recomputes \(c\) from the configured posterior length-scale summary of \(\rho\), and enlarges \(c\) further if the observation or prediction grid lies too close to the computational boundary. A refit is accepted only if the following checks pass. First, the Phase A length-scale and spectral-tail checks must still hold. The high-frequency fraction of derivative-weighted posterior basis energy, computed from squared basis coefficients weighted by \(\lambda_k^{q_\star}\), must be below its tolerance. If that fraction is unavailable, the total derivative-weighted posterior basis energy must change by less than its relative tolerance from the previous fit. Second, the posterior length-scale stability summary and leave-one-out expected log predictive density must change by less than their configured tolerances, and sampler diagnostics must pass. Third, for every monitored level in \(\mathcal M\), the posterior mean and standard-deviation curves must change by less than their prescribed tolerances from the previous fit. The same check is applied to posterior means and standard deviations of any unfixed integration constants. The curve comparison is first made on a trimmed interior prediction grid, obtained by dropping a small fraction of points near each endpoint, and then checked on the full grid to detect endpoint instability. Phase B terminates only after a configured number of consecutive enlarged-basis refits pass all checks. This requirement guards against reporting summaries that are artifacts of a particular truncation.

\section{Simulation study}\label{sec:simulation}

The simulation study evaluates whether TARTARE's target-aware ensemble calibration improves finite-rank HSGP approximations for differential ensemble inference. The comparison is against anchor-level calibration in the spirit of \citet{riutort2023practical}, which chooses the HSGP basis for the observed anchor process \(f_0\) alone. The central question is whether calibrating the basis to derivative or full-ensemble targets improves posterior recovery of the monitored ensemble levels, especially the second derivative \(f_{-2}\), without sacrificing anchor or integral summaries.

We simulate data from a second-order differential ensemble \((f_{-2},f_{-1},f_0,f_1,f_2)\) and observe only the noisy anchor \(f_0\). The simulation uses \(1000\) replicated datasets for each configuration and evaluates posterior summaries on a \(101\)-point prediction grid. Within each replicate, data were generated on \([-1,1]\) from \(y_i=f_0(t_i)+\epsilon_i\), \(\epsilon_i\sim N(0,0.10^2)\), at \(n\in\{25,50\}\) equally spaced design points. The design crossed squared exponential (SE) and Mat\'ern \(7/2\) (Mat7/2) anchor kernels with normalized length-scales \(\rho\in\{0.35,0.65\}\). The positive-index constants were fixed at their generating values, \(\kappa_1=0.60\) and \(\kappa_2=-0.40\), to isolate the kernel-driven HSGP approximation. Boundary-constant sensitivity is reported in \cref{appendix:boundary-constant-sensitivity}.

\Cref{tab:tartare-simulation-summary} summarizes the simulation results. Exact is the analytic exact-covariance reference at the generating hyperparameters. HSGP denotes an anchor-level practical Hilbert space Gaussian process rule in the style of \citet{riutort2023practical}. The TARTARE rules are T-f monitoring \(\mathcal M=\{0\}\), T-D2 monitoring \(\mathcal M=\{-2\}\), and T-all monitoring \(\mathcal M=\{-2,-1,0,1,2\}\). Spline is a smoothing spline point estimate comparator. The table reports failures from the online phase of the TARTARE calibration procedure, selected basis size \(K\), range factor \(c\), integrated root mean squared error (IRMSE), pointwise \(95\%\) coverage, and the geometric-mean IRMSE ratio relative to Exact within each configuration. Coverage is the percentage, over prediction-grid locations and available replicated fits, for which the simulated latent value \(f_p(t)\) lies inside the pointwise central \(95\%\) posterior interval for the corresponding ensemble level \(p\).

\begin{table}[htbp]
\centering
\caption{\label{tab:tartare-simulation-summary}Simulation summary for second-order differential ensemble fits. Each row aggregates 1000 replicated datasets for the listed configuration. Numerical summaries other than Failed are computed among completed fits. Coverage is not reported for the point estimate spline comparator.}
{\catcode`\%=12
\resizebox{\textwidth}{!}{\input{tables/tartare_simulation_summary.tex}}}
\end{table}

\Cref{tab:tartare-simulation-summary} shows the same qualitative ordering across sample sizes and length-scales. Exact gives the Monte Carlo reference and has coverage close to the nominal pointwise level. In the SE cases, the anchor-level HSGP rule uses the smallest bases but under-resolves derivatives: \(f_{-2}\) IRMSE is several times the Exact value and \(f_{-2}\) coverage is consistently low, while the anchor and integral levels remain much closer to Exact. The anchor-targeted T-f rule improves SE derivative error relative to HSGP, but it is not uniformly adequate for derivative uncertainty. In the Mat7/2 cases this distinction is sharper: T-f has geometric-mean IRMSE close to Exact, yet \(f_{-2}\) coverage remains well below nominal.

Monitoring \(f_{-2}\) directly through T-D2 moves derivative coverage close to the Exact reference and keeps anchor and integral summaries close to Exact. T-all gives essentially the same pointwise IRMSE and coverage as T-D2 in this table, but usually needs a larger basis, especially for Mat7/2. The smoothing-spline point estimates keep anchor and integral errors small in several settings. Their geometric-mean IRMSE ratios are larger than the corresponding TARTARE ratios, they provide no posterior coverage, and \(f_{-2}\) recovery is weak across the kernel settings. Additional exact-vs-HSGP posterior covariance diagnostics from the same simulation run compare T-D2 and T-all in \cref{appendix:additional-simulation-diagnostics} in the Supplementary Material.

The reported online calibration failures are nonzero but modest, and they are most visible for the derivative-aware and full-ensemble rules with \(n=25\). These failures refer to the Phase B online refinement contract described in \cref{sec:tartare-calibration}. Passing Phase B requires length-scale and spectral-tail checks, posterior basis-energy diagnostics, leave-one-out and sampler diagnostics, monitored-curve stability, and when integration constants are estimated rather than fixed, stability of their posterior means and standard deviations.

Overall, the table suggests that target-aware TARTARE calibration is a necessary computational safeguard when finite-rank HSGP approximations are used for derivative or ensemble inference. Anchor-only calibration can be adequate for anchor-level summaries, but derivative and ensemble summaries require derivative-aware calibration, often with larger finite-rank representations or enlarged computational domains to reduce boundary-induced approximation error from the Laplacian--Dirichlet eigen-representation.

\section{Application -- Motorcycle crash data}\label{sec:data}

We apply the method to motorcycle crash data to infer a coherent kinematic state from noisy head-acceleration measurements, using short-horizon turning-point behavior as a focused example of a joint ensemble functional.

The \texttt{mcycle} data comprise $133$ measurements of head acceleration from a simulated $60\,\mathrm{ms}$ motorcycle crash experiment and are a well-known example in nonparametric regression \citep{silverman1985spline,desouza2014switching}. For each observation, \texttt{times} gives elapsed time in milliseconds after impact and \texttt{accel} gives the corresponding head acceleration in units of \(g\). We use acceleration as the anchor of a differential ensemble and ask whether the inferred kinematic state predicts a near-future reversal of velocity during impact. This target is not a feature of a smoothed acceleration curve alone but depends jointly on acceleration, jerk, snap, velocity, and position.

The mean structure is a second-order differential ensemble with acceleration as the anchor \(f_0\). The negative levels are thus jerk \(f_{-1}\) and snap \(f_{-2}\), while the positive levels are velocity and position,
\[
f_1(t)=\int_{t_0}^{t}f_0(u)\,du+\kappa_1,\qquad
f_2(t)=\int_{t_0}^{t}\int_{t_0}^{v}f_0(u)\,du\,dv+(t-t_0)\kappa_1+\kappa_2.
\]
Here \(\kappa_1\) and \(\kappa_2\) represent initial velocity and initial position at crash time \(t_0\). We set \(\kappa_1\sim N(8\,\mathrm{m}/\mathrm{s},\,2.5^2\,\mathrm{m}^2/\mathrm{s}^2)\) to encode plausible forward motion at the start of the impact window without fixing the velocity level. We set \(\kappa_2\sim N(0\,\mathrm{m},\,0.05^2\,\mathrm{m}^2)\) independently of \(\kappa_1\), so that the position origin is centered at the crash-time reference while still allowing modest uncertainty in that origin. The data are also substantially heteroscedastic, so we model the log measurement standard deviation with a first-order differential ensemble \((h_{-1},h_0,h_1)\), where \(h_0(t)=\log\sigma(t)\), and \(\sigma(t)\) is the acceleration measurement standard deviation. Its cumulative component is anchored at \(t_0\) by fixing the integration constant implicitly as \(\kappa_h=h_0(t_0)\), so the cumulative log standard deviation level starts from the modeled crash-time noise scale rather than from a separate free parameter.

Both differential ensembles used Mat\'ern \(7/2\) anchor kernels with the Laplacian--Dirichlet basis representation. A common basis count and computational domain were selected by a TARTARE calibration on the motorcycle time grid with monitoring set \(\mathcal M=\{-2\}\), corresponding to the snap level of the acceleration ensemble, and the resulting basis was reused for the log-standard-deviation ensemble.  With centered observation half-width \(W=27.6\,\mathrm{ms}\), the two-phase calibration terminated with stable phase B diagnostics and selected \(K_{\mathrm{TARTARE}}=181\), \(c_{\mathrm{TARTARE}}=1.42\), and \(L_{\mathrm{TARTARE}}=39.31\,\mathrm{ms}\). Posterior trajectories were evaluated on a grid of \(500\) equidistant prediction points spanning the observed time range, and posterior draws for velocity and position were converted from \(g\,\mathrm{ms}\) and \(g\,\mathrm{ms}^2\) to \(\mathrm{m}/\mathrm{s}\) and \(\mathrm{m}\). Covariance parameter priors, posterior summaries, and trace plots are given in \cref{appendix:additional-application-information} in the Supplementary Material.

Posterior means and pointwise 95\% credible bands for the two ensembles are shown in \cref{fig:motorcycle-posterior-ensemble}. The ensemble for the mean acceleration captures the sharp deceleration phase and propagates it to jerk, snap, velocity, and position, with uncertainty increasing most visibly for the differentiated levels. The log standard deviation ensemble shows time-varying measurement variability rather than a constant noise scale, together with the associated noise dynamics and cumulative log uncertainty.

\begin{figure}[htbp]
\centering
\includegraphics[width=\textwidth]{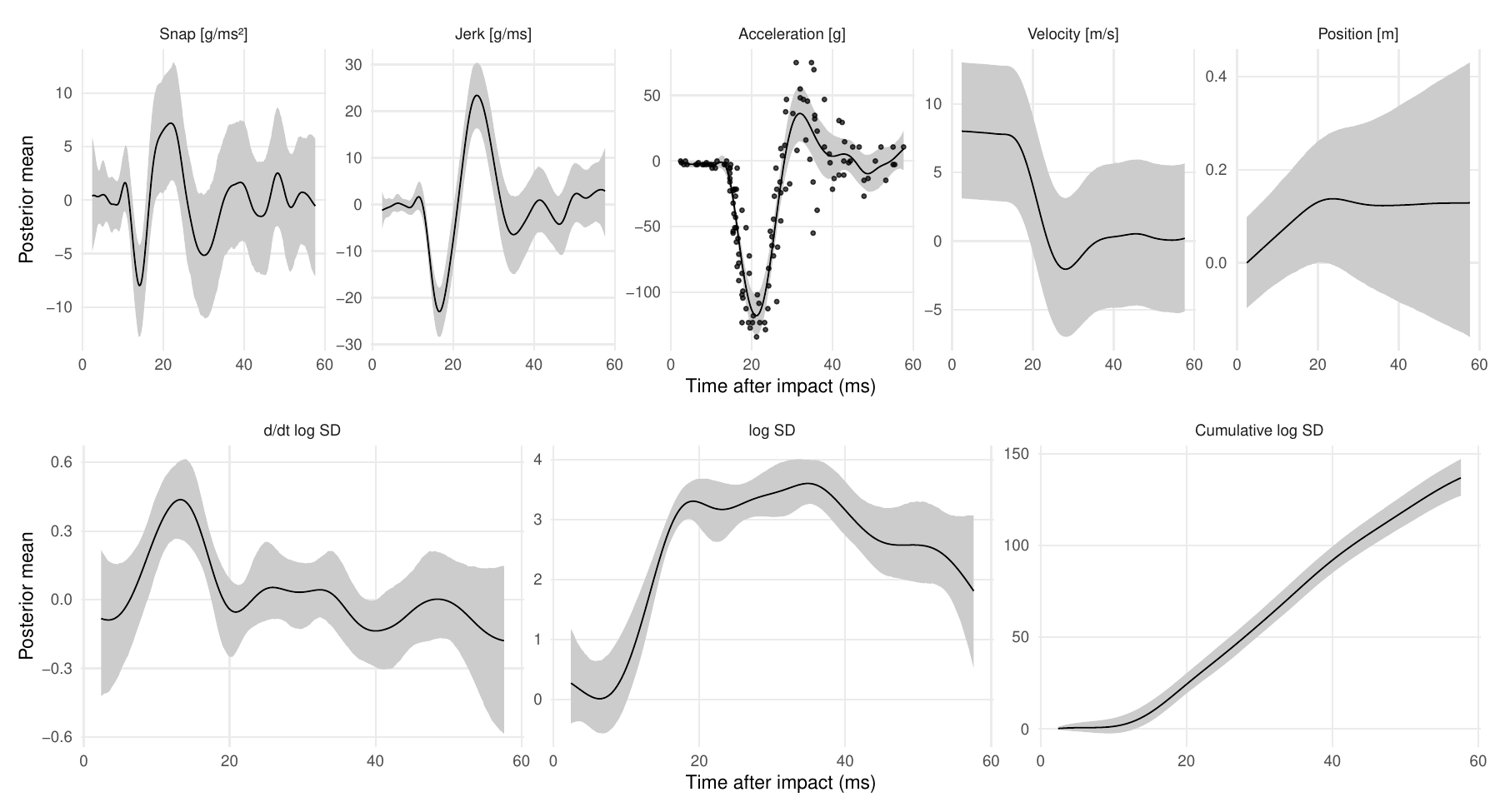}
\caption{Posterior ensemble summaries for the motorcycle crash data. The top row shows posterior means and 95\% credible bands for the mean differential ensemble levels (snap, jerk, acceleration, velocity, and position), with observed acceleration points overlaid. The bottom row shows posterior means and 95\% credible bands for the log standard deviation ensemble levels.}
\label{fig:motorcycle-posterior-ensemble}
\end{figure}

One way to use the joint inference over the full ensemble is to ask a short-horizon rebound question. Let \(I\) denote the observed crash window. At each time in \(I\), the question is whether the current kinematic state predicts a velocity reversal within the next \(5\,\mathrm{ms}\), and if so how far the head would have moved by that turning point. The reversal probability uses the posterior relation among velocity, acceleration, jerk, and snap, and the displacement threshold probability additionally uses the posterior position level. Thus the target is not an acceleration curve summary, but a posterior functional of the coupled snap--jerk--acceleration--velocity--position state.

For a starting time \(t\), define the local velocity forecast through snap,
\[
V_t(\tau)=
f_1(t)+\tau f_0(t)+\frac{\tau^2}{2}f_{-1}(t)+\frac{\tau^3}{6}f_{-2}(t),
\]
and let \(\tau^{\ast}(t)\) be the earliest positive root of \(V_t(\tau)\) in \((0,\tau_{\max}]\), setting \(\tau^{\ast}(t)=\infty\) if no such root exists. We take \(\tau_{\max}=5\,\mathrm{ms}\), so \(\tau^{\ast}(t)\le \tau_{\max}\) is the event that the local forecast predicts a velocity reversal within the next \(5\,\mathrm{ms}\). For finite \(\tau^{\ast}(t)\), the corresponding displacement is
\[
x_{\mathrm{turn}}(t)=
f_2(t)+\tau^{\ast}(t)f_1(t)+\frac{\tau^{\ast}(t)^2}{2}f_0(t)+\frac{\tau^{\ast}(t)^3}{6}f_{-1}(t)+\frac{\tau^{\ast}(t)^4}{24}f_{-2}(t).
\]
Let \(\mathcal D=(\texttt{times}_i,\texttt{accel}_i)_{i=1}^n\) denote the observed motorcycle data. We summarize this target by the marginal turning probability and by displacement-threshold probabilities,
\[
p_{\mathrm{turn}}(t)=P(\tau^\ast(t)\le\tau_{\max}\mid \mathcal D),
\qquad
p_d(t)=P(x_{\mathrm{turn}}(t)>d,\ \tau^\ast(t)\le\tau_{\max}\mid \mathcal D).
\]

We estimate these posterior quantities from sampled trajectories evaluated on a finite prediction grid \(\mathcal G\). Write \(\widehat f_p^{(s)}\) for the \(s\)th sampled trajectory from the fitted posterior at level \(p\), evaluated on \(\mathcal G\). For each grid time \(t\in\mathcal G\cap I\), the draw-level quantities \(V_t^{(s)}\), \(\tau^{\ast(s)}(t)\), and \(x_{\mathrm{turn}}^{(s)}(t)\) are obtained by replacing \(f_p\) with \(\widehat f_p^{(s)}\) in the definitions above. The probabilities \(p_{\mathrm{turn}}(t)\) and \(p_d(t)\) are then estimated by the corresponding posterior draw proportions, so \(\mathcal G\) is used to summarize the posterior functional rather than to construct finite-difference derivatives.

\Cref{fig:motorcycle-turning-point-composite} summarizes the posterior evidence for short-horizon rebound over start time and displacement threshold. The posterior turning probability \(p_{\mathrm{turn}}(t)\) peaks at \(0.612\) around \(20.1\,\mathrm{ms}\). At that time, among draws with \(\tau^\ast(t)\le\tau_{\max}\), the posterior median of \(x_{\mathrm{turn}}(t)\) is \(11.8\,\mathrm{cm}\), with central 95\% interval \((0.7,23.2)\,\mathrm{cm}\), and the threshold probability \(p_d(t)\) at \(d=10\,\mathrm{cm}\) is \(0.383\). The figure displays \(p_d(t)\) as a time-by-threshold heatmap, posterior probability trajectories for four selected thresholds, and the marginal turning probability \(p_{\mathrm{turn}}(t)\), identifying when the inferred state most strongly supports a substantial near-future rebound displacement.

\begin{figure}[htbp]
\centering
\includegraphics[width=\textwidth]{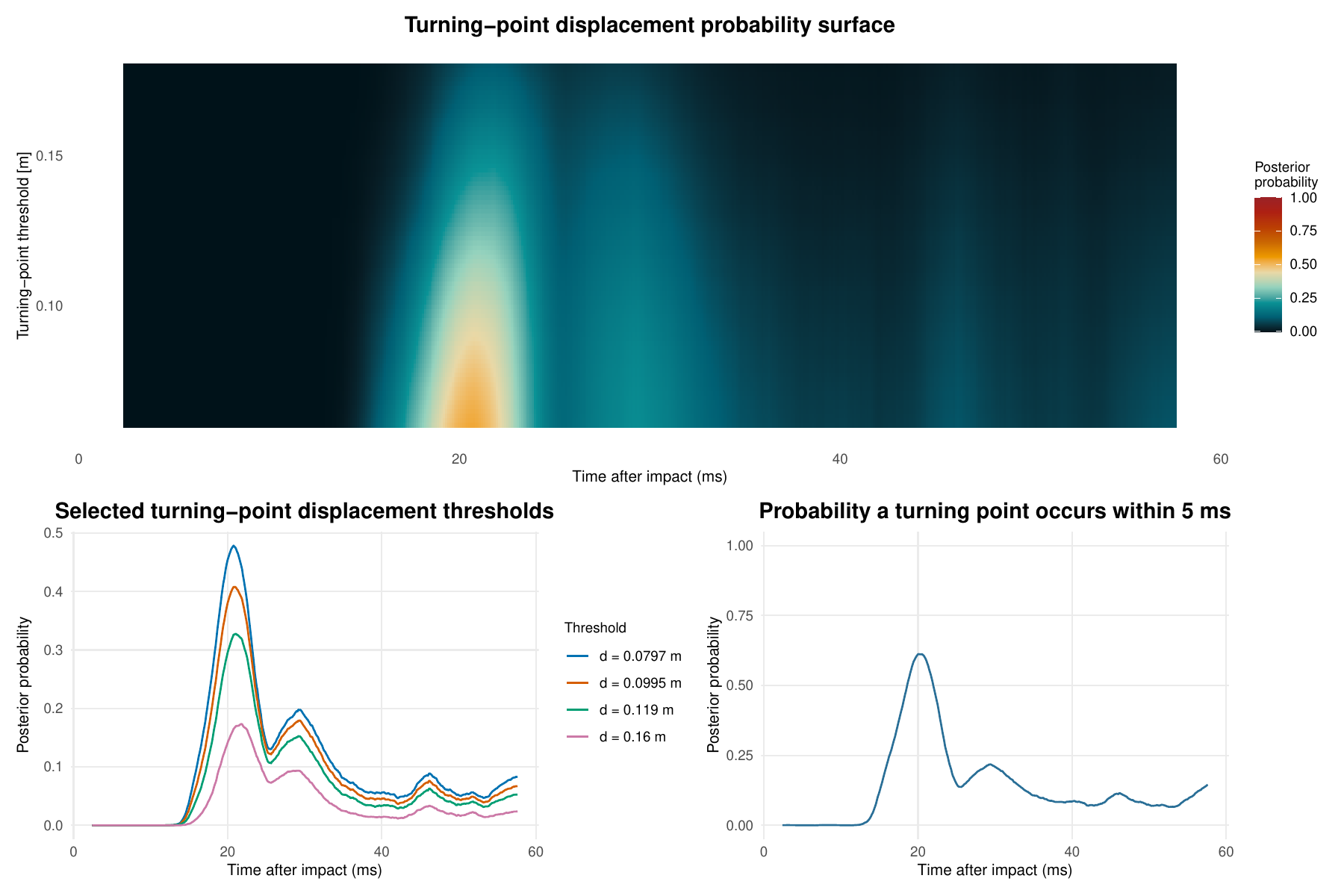}
\caption{Posterior summary of the short-horizon turning-point functional in the motorcycle crash application. The top panel shows \(p_d(t)=P(x_{\mathrm{turn}}(t)>d,\ \tau^\ast(t)\le 5\,\mathrm{ms}\mid \mathcal{D})\) over start time \(t\) and displacement threshold \(d\). The bottom-left panel shows \(p_d(t)\) for four selected thresholds, and the bottom-right panel shows \(p_{\mathrm{turn}}(t)=P(\tau^\ast(t)\le 5\,\mathrm{ms}\mid \mathcal{D})\). Together the panels show when the jointly inferred kinematic state most strongly supports a substantial near-future rebound displacement.}
\label{fig:motorcycle-turning-point-composite}
\end{figure}

\section{Discussion}\label{sec:discussion}

Observed functional data often enter a model through one anchor curve, but the scientific target is rarely limited to that curve.  Rates, accelerations, accumulated exposure, boundary states, and turning-point or threshold functionals all depend on several linked levels of a differential state.  The contribution is not the observation that Gaussian processes can be differentiated or integrated. It is the organization of these operations around a chosen anchor, with covariance propagated across all levels and with positive-index boundary information represented explicitly. TARTARE is the necessary calibration procedure that makes the finite-rank implementation target-aware when derivative or ensemble summaries are reported.  This organization keeps the inferential target visible from model specification to posterior reporting.

The exact construction separates two sources of uncertainty. The anchor covariance induces the derivative, anchor, integral, and cross-level covariance blocks through linear operators. The integration constants contribute a finite-dimensional polynomial covariance above the anchor. This distinction is statistically important.

The weak Gaussian construction is the inferential foundation of the paper. It gives the finite-dimensional Gaussian laws and covariance blocks needed for posterior computation. The pathwise differential chain is a stronger representation that requires additional regularity and version conditions.

The transformed HSGP implementation provides a reusable finite-rank representation for stationary one-dimensional anchor kernels. Derivatives and integrals are applied to Laplacian--Dirichlet basis functions, while the covariance family enters through spectral weights.  For stationary one-dimensional anchor kernels, the transformed HSGP representation approximates the kernel-driven operator blocks by applying derivatives and definite integrals to Laplacian--Dirichlet basis functions and weighting them by the anchor spectral density.  This makes differential ensembles computable to any prescribed finite order permitted by the regularity assumptions without deriving new covariance-specific derivative or integral formulas for each order.  The covariance contribution from integration constants is finite dimensional and is retained exactly, so truncation error is confined to the kernel-driven part.  Operator-aware calibration then asks whether a chosen finite-rank representation is adequate for a prespecified monitored set on a calibration grid.  The approximation theory establishes uniform bounds for derivative covariance blocks and finite-grid posterior convergence conditional on fixed hyperparameters. It does not establish a universal finite-\(K,L\) error bound for all posterior summaries, continuum convergence, or convergence after hyperparameter integration.

TARTARE addresses the practical gap between asymptotic approximation theory and a finite basis used in a fitted model. Its role is diagnostic and target-aware: it asks whether a chosen HSGP representation is accurate enough for specified monitored ensemble levels and whether reported posterior summaries are stable under basis enlargement.  Derivative-aware and ensemble-aware calibration correct this under-resolution in the studied scenarios; when derivative and full-ensemble calibration are compared jointly, their posterior discrepancies are broadly similar rather than uniformly ordered.  Recorded calibration failures in the production study are therefore best read as failures of the finite search contract for the monitored approximation, separate from Markov chain Monte Carlo diagnostics.

The motorcycle application demonstrates interpretable posterior state summaries from the same anchored construction.  Acceleration is the observed anchor, but the posterior summaries involve jerk, snap, velocity, position, and a short-horizon turning-point functional that depends on all these levels.  The analysis is useful because the state is sampled jointly: velocity changes, threshold exceedance probabilities, and turning-point displacements inherit the covariance implied by the acceleration anchor and by the initial-state priors.  This example shows how the framework reports coherent state uncertainty in an applied model; it is not a claim of broad empirical superiority over all competing functional-data workflows.

Several limitations define the current scope. The method is one-dimensional in its developed computational form and uses stationary anchor kernels for the transformed HSGP approximation. High-order derivative targets require kernels with sufficient mean-square smoothness and increase numerical demands because differentiation weights high frequencies. Positive-index inference is only as credible as the boundary specification, since independent integration constants remain prior-driven without additional information.  High-order derivatives require kernels smooth enough for the requested mean-square derivatives, and they place increasing numerical stress on finite-rank approximations because differentiation emphasizes high frequencies.  Mat\'ern kernels near their differentiability limits are especially demanding, both analytically and computationally.  The calibration results are finite-grid and target-specific, and the present simulations provide scenario evidence for the calibration workflow rather than a general finite-\(K,L\) posterior error theorem.

Future work should sharpen finite-basis posterior error theory, especially with inferred hyperparameters and nonlinear state functionals. Adaptive calibration rules that select monitoring sets from the posterior questions actually reported would make TARTARE less dependent on precomputed lookup tables.  Adaptive or joint calibration criteria could choose monitored targets from the posterior questions actually reported, while preserving the separation between approximation certification and sampler diagnostics.  Broader comparisons with spline derivatives, finite differences of posterior draws, derivative-GP workflows, and general linear-operator GP models would clarify where the ensemble structure is most consequential in practice.  Applied models would also benefit from the optional dependent integration-constant extension in the online supplement, particularly when initial states are partially informed by external measurements or physical constraints.  Finally, extending anchored differential ensembles to multidimensional and spatio-temporal domains would require operator-specific eigensystems and calibration rules, but would keep the same central principle: choose the anchor observed by the data, model the surrounding differential state jointly, and make boundary information explicit.

\input{bibliography.tex}
\appendix
\renewcommand{\thesection}{S}
\renewcommand{\thesubsection}{S.\arabic{subsection}}
\renewcommand{\thesubsubsection}{S.\arabic{subsection}.\arabic{subsubsection}}
\renewcommand{\thefigure}{S.\arabic{figure}}
\renewcommand{\thetable}{S.\arabic{table}}
\renewcommand{\theHsection}{S}
\renewcommand{\theHsubsection}{S.\arabic{subsection}}
\renewcommand{\theHsubsubsection}{S.\arabic{subsection}.\arabic{subsubsection}}
\renewcommand{\theHfigure}{S.\arabic{figure}}
\renewcommand{\theHtable}{S.\arabic{table}}
\setcounter{figure}{0}
\setcounter{table}{0}

\section{Supplementary Material}

\subsection{Proofs for main text results}\label{appendix:main-text-proofs}
\input{main-proofs.tex}

\subsection{Equivalent constructions above the anchor}\label{appendix:alternative-positive-constructions}
Fix \(r \in \mathbb{N}\) and let \(f_0\) be an anchor process satisfying the regularity assumptions of the main construction. This subsection compares three equivalent ways to define the positive-index components \(f_1,\ldots,f_r\): the \(t_0\)-based definite-integral construction used in the main text, a recursive construction through initial
  values, and a formulation based on indefinite antiderivatives.

The main-text construction defines the positive-index components by
\[
f_p(t)= (\mathcal I_{t_0}^p f_0)(t)+\sum_{j=1}^p \frac{(t-t_0)^{p-j}}{(p-j)!}\kappa_j,
\]
where $\boldsymbol{\kappa}:=(\kappa_1,\ldots,\kappa_r)^\top \sim N(\boldsymbol{\mu}_\kappa,\boldsymbol{\Sigma}_\kappa)$ is independent of $f_0$. The repeated definite integral is the canonical contribution generated by \(f_0\). Once \(f_0\) is fixed, this term is fixed. All remaining freedom is the finite-dimensional polynomial determined by \(\kappa_1,\ldots,\kappa_p\).

The same process class can be defined recursively by
\[
f_1(t):=\int_{t_0}^t f_0(u)\,du+\kappa_1,\quad
f_p(t):=\int_{t_0}^t f_{p-1}(u)\,du+\kappa_p,\quad p=2,\ldots,r,
\]
which gives \(D_t f_p(t)=f_{p-1}(t)\) and \(f_p(t_0)=\kappa_p\) for \(p=1,\ldots,r\). An elementary induction yields the definite-integral representation above. Each recursion adds the initial value \(f_p(t_0)=\kappa_p\), while previously introduced constants are integrated into lower-degree polynomial terms. The recursive and definite integral forms therefore differ only in presentation. The same quantities \(\kappa_p\) may be interpreted either as initial values or as coefficients of the polynomial correction.

The indefinite antiderivative view starts from the general solution of \(\frac{d^p}{dt^p}f_p(t)=f_0(t)\). For fixed \(p\geq 1\),
\[
f_p(t)= (\mathcal I_{t_0}^p f_0)(t)+\sum_{m=0}^{p-1}\frac{(t-t_0)^m}{m!}a_{p,m},
\]
where \(a_{p,0},\ldots,a_{p,p-1}\) are arbitrary constants, or arbitrary random variables in a stochastic formulation. Thus \(p\)-fold indefinite integration determines \(f_p\) only up to a polynomial of degree \(p-1\).

Most choices of these polynomial coefficients do not define a coherent differential ensemble. The ensemble also requires the chain \(D_t f_p(t)=f_{p-1}(t)\) for \(p=1,\ldots,r\). This forces the compatibility relations
$a_{p,m+1}=a_{p-1,m}$ for $m=0,\ldots,p-2$. Writing $\kappa_j:=a_{j,0}=f_j(t_0)$ for $j=1,\ldots,r$, these relations reduce the polynomial term to
\[
\sum_{j=1}^p \frac{(t-t_0)^{p-j}}{(p-j)!}\kappa_j.
\]
Hence the indefinite antiderivative view becomes equivalent to the recursive and definite integral constructions once the differential chain is imposed. The advantage of the definite integral formulation is that it selects the unique \(p\)-fold antiderivative whose derivatives of orders \(0,\ldots,p-1\) vanish at \(t_0\), and places all remaining finite-dimensional freedom explicitly in \(\kappa_1,\ldots,\kappa_p\).

Conditional on \(f_0\), only \(\boldsymbol{\kappa}\) remains random above the anchor. One may therefore retain the integration constants as Gaussian parameters, constrain them through direct or virtual observations, or fix selected constants at zero. Fixing selected constants at zero is a modeling restriction, not a harmless change of parameterization, because it removes part of the finite-dimensional uncertainty above the anchor.

\subsection{Dependent integration constants and coupled differential ensembles}\label{appendix:dependent-integration-constants}

Use the notation and operator conventions of \cref{sec:methods}. Let \(\mathfrak{f}:=(f_{-r},\ldots,f_r)\) be defined from the anchor process \(f_0\) as in \cref{prop:gaussian-ensemble}, except that \(\boldsymbol{\kappa}:=(\kappa_1,\ldots,\kappa_r)^\top\) may depend on \(f_0\). We first derive covariance blocks for a valid jointly Gaussian pair \((f_0,\boldsymbol{\kappa})\), then characterize admissibility, and finally give constructions that realize such pairs.

Write
\begin{align*}
\boldsymbol{\gamma}(s):=\Cov\bigl(f_0(s),\boldsymbol{\kappa}\bigr)=\bigl(\gamma_1(s),\ldots,\gamma_r(s)\bigr)^\top,
\qquad
\gamma_m(s):=\Cov\bigl(f_0(s),\kappa_m\bigr),
\end{align*}
and let \(\boldsymbol{\Sigma}_{\kappa}:=\Cov(\boldsymbol{\kappa})\). Assume \cref{ass:ms-regularity}. For \(m=1,\ldots,r\), assume \(\gamma_m\in C^r([t_0,t_1])\). Also assume that all displayed cross-covariance derivatives below exist as continuous functions on their stated domains with endpoint derivatives interpreted one-sided. These assumptions justify differentiating covariance kernels and passing repeated \(L^2\)-integrals through covariance by Cauchy--Schwarz and Fubini. In particular, an evaluation of \(\mathcal A_{t_0}^{[p]}\) at an endpoint for \(p<0\) denotes the corresponding one-sided mean-square derivative value. Whenever a pathwise statement is used, also assume \cref{ass:pathwise-regularity} together with the joint measurability and integrability hypotheses stated in \cref{prop:samplepaths}.

The ensemble representation is unchanged. The only difference is that \((f_0,\boldsymbol{\kappa})\) is jointly Gaussian rather than independent. Consequently, the mean functions remain those from \cref{prop:gaussian-ensemble}. Whenever the pathwise assumptions above are imposed, the pathwise differential chain relations from \cref{prop:samplepaths} continue to hold. In particular, for every \(p=1,\ldots,r\) and \(j=0,\ldots,p-1\), one has \(\frac{d^j}{dt^j}f_p(t_0)=\kappa_{p-j}\) almost surely. Thus dependence changes only the covariance blocks through \(\boldsymbol{\gamma}\), while the representation and boundary identities remain unchanged. Those additional terms are derived below.

\begin{lemma}[Cross-covariance with the integration constants]\label{lem:appendix-cross-basic}
For each \(m\in\{1,\ldots,r\}\), each \(p\in\{-r,\ldots,r\}\), and each \(s\in[t_0,t_1]\) when \(p\ge 0\) or \(s\in(t_0,t_1)\) when \(p<0\),
\begin{align*}
\Cov\bigl((\mathcal A_{t_0}^{[p]}f_0)(s),\kappa_m\bigr)=(\mathcal A_{t_0}^{[p]}\gamma_m)(s).
\end{align*}
\end{lemma}

\begin{proof}
For $p=0$ this is the definition of $\gamma_m$. For \(p=\ell\geq 1\), the assumed covariance--integration interchange gives
\begin{align*}
\Cov\bigl((\mathcal I_{t_0}^{\ell}f_0)(s),\kappa_m\bigr)
&=\Cov\left(\int_{t_0}^{s}\frac{(s-u)^{\ell-1}}{(\ell-1)!}f_0(u)\,du,\kappa_m\right)\\
&=\int_{t_0}^{s}\frac{(s-u)^{\ell-1}}{(\ell-1)!}\gamma_m(u)\,du
=(\mathcal I_{t_0}^{\ell}\gamma_m)(s).
\end{align*}
For $p=-\ell\leq -1$, the covariance--differentiation interchange assumed above gives
\begin{align*}
\Cov\bigl((D_s^{\ell}f_0)(s),\kappa_m\bigr)
&=\frac{d^\ell}{ds^\ell}\Cov\bigl(f_0(s),\kappa_m\bigr)
=(D_s^{\ell}\gamma_m)(s).
\end{align*}
These three cases exhaust all $p\in\{-r,\ldots,r\}$.
\end{proof}

Combining \cref{lem:appendix-cross-basic} with the representation of \(f_p\) gives the full covariance formula.

\begin{proposition}[General covariance formula for the dependent extension]\label{thm:appendix-general-cov}
Retain the standing assumptions of this subsection. For \(p,q\in\{-r,\ldots,r\}\), let \(C_{pq}(s,t):=\Cov\bigl(f_p(s),f_q(t)\bigr)\). Then
\begin{align*}
C_{pq}(s,t)
&=(\mathcal A_{1,t_0}^{[p]}\mathcal A_{2,t_0}^{[q]}C_\theta)(s,t)\\
&\quad+\mathbf{1}_{\{q>0\}}\sum_{m=1}^{q}\frac{(t-t_0)^{q-m}}{(q-m)!}(\mathcal A_{t_0}^{[p]}\gamma_m)(s)\\
&\quad+\mathbf{1}_{\{p>0\}}\sum_{j=1}^{p}\frac{(s-t_0)^{p-j}}{(p-j)!}(\mathcal A_{t_0}^{[q]}\gamma_j)(t)\\
&\quad+\mathbf{1}_{\{p>0\}}\mathbf{1}_{\{q>0\}}\sum_{j=1}^{p}\sum_{m=1}^{q}
\frac{(s-t_0)^{p-j}}{(p-j)!}\frac{(t-t_0)^{q-m}}{(q-m)!}(\boldsymbol{\Sigma}_{\kappa})_{jm}.
\end{align*}
If $\boldsymbol{\gamma}\equiv \mathbf{0}$, then the two middle terms vanish and one recovers the covariance formula in \cref{prop:gaussian-ensemble}.
\end{proposition}

\begin{proof}
Write
\begin{align*}
A_p(t):=(\mathcal A_{t_0}^{[p]}f_0)(t),
\qquad
B_p(t):=\mathbf{1}_{\{p>0\}}\sum_{j=1}^{p}\frac{(t-t_0)^{p-j}}{(p-j)!}\kappa_j,
\end{align*}
so that $f_p(t)=A_p(t)+B_p(t)$. Then
\begin{align*}
\Cov\bigl(f_p(s),f_q(t)\bigr)
&=\Cov\bigl(A_p(s),A_q(t)\bigr)+\Cov\bigl(A_p(s),B_q(t)\bigr)\\
&\quad+\Cov\bigl(B_p(s),A_q(t)\bigr)+\Cov\bigl(B_p(s),B_q(t)\bigr).
\end{align*}
By the covariance formula for the anchor part established in the proof of \cref{prop:gaussian-ensemble},
\begin{align*}
\Cov\bigl(A_p(s),A_q(t)\bigr)=(\mathcal A_{1,t_0}^{[p]}\mathcal A_{2,t_0}^{[q]}C_\theta)(s,t).
\end{align*}
Moreover,
\begin{align*}
\Cov\bigl(A_p(s),B_q(t)\bigr)
&=\mathbf{1}_{\{q>0\}}\sum_{m=1}^{q}\frac{(t-t_0)^{q-m}}{(q-m)!}\Cov\bigl(A_p(s),\kappa_m\bigr),\\
\Cov\bigl(B_p(s),A_q(t)\bigr)
&=\mathbf{1}_{\{p>0\}}\sum_{j=1}^{p}\frac{(s-t_0)^{p-j}}{(p-j)!}\Cov\bigl(\kappa_j,A_q(t)\bigr).
\end{align*}
By \cref{lem:appendix-cross-basic}, together with symmetry of covariance in the second term, these two contributions equal the two middle terms in the stated formula. Finally,
\begin{align*}
\Cov\bigl(B_p(s),B_q(t)\bigr)
&=\mathbf{1}_{\{p>0\}}\mathbf{1}_{\{q>0\}}\sum_{j=1}^{p}\sum_{m=1}^{q}
\frac{(s-t_0)^{p-j}}{(p-j)!}\frac{(t-t_0)^{q-m}}{(q-m)!}\Cov(\kappa_j,\kappa_m),
\end{align*}
which is the last term. Adding the four contributions proves the formula.
\end{proof}

\subsubsection{Validity and conditional characterization of the dependent extension}
We now characterize when \(C_\theta\), \(\boldsymbol{\gamma}\), and \(\boldsymbol{\Sigma}_{\kappa}\) define a valid jointly Gaussian pair \((f_0,\boldsymbol{\kappa})\).

For any $n\in\mathbb{N}$ and any $u_1,\ldots,u_n\in[t_0,t_1]$, the covariance matrix of
\begin{align*}
\bigl(f_0(u_1),\ldots,f_0(u_n),\kappa_1,\ldots,\kappa_r\bigr)^\top
\end{align*}
has the block form
\begin{align*}
\boldsymbol{\Omega}_n:=
\begin{pmatrix}
\mathbf{C}_n & \mathbf{G}_n\\
\mathbf{G}_n^\top & \boldsymbol{\Sigma}_{\kappa}
\end{pmatrix},
\end{align*}
where \((\mathbf{C}_n)_{ij}=C_\theta(u_i,u_j)\) and \((\mathbf{G}_n)_{im}=\gamma_m(u_i)\). The admissibility of the dependent extension is therefore determined by the positivity of these finite-dimensional covariance matrices. The next theorem characterizes admissibility of the jointly Gaussian pair \((f_0,\boldsymbol{\kappa})\). The standing smoothness assumptions are used afterward to extend this pair across the full differential ladder.

\begin{theorem}[Covariance function criterion for the dependent extension]\label{thm:appendix-schur}
Let \(\boldsymbol{\Sigma}_{\kappa}\in\mathbb R^{r\times r}\) be symmetric positive definite, and define
\[
C_{\theta\mid\kappa}(s,t):=C_\theta(s,t)-\boldsymbol\gamma(s)^\top\boldsymbol{\Sigma}_{\kappa}^{-1}\boldsymbol{\gamma}(t).
\]
Then the following are equivalent:
\begin{enumerate}
\item There exists a jointly Gaussian pair \((f_0,\boldsymbol{\kappa})\) such that
\[
\E[f_0(t)]=\mu(t),\quad t\in[t_0,t_1],
\qquad
\E[\boldsymbol{\kappa}]=\boldsymbol{\mu}_\kappa,
\]
and
\begin{align*}
\Cov(f_0(s),f_0(t))&=C_\theta(s,t), && s,t\in[t_0,t_1],\\
\Cov(f_0(s),\boldsymbol{\kappa})&=\boldsymbol{\gamma}(s), && s\in[t_0,t_1],\\
\Cov(\boldsymbol{\kappa})&=\boldsymbol{\Sigma}_{\kappa}.
\end{align*}
\item The covariance function \(C_{\theta\mid\kappa}\) is positive semidefinite on \([t_0,t_1]\), that is, for every \(n\in\mathbb{N}\) and every \(u_1,\ldots,u_n\in[t_0,t_1]\), the matrix \(\bigl(C_{\theta\mid\kappa}(u_i,u_j)\bigr)_{i,j=1}^n\) is positive semidefinite.
\end{enumerate}
\end{theorem}

\begin{proof}
Assume first that $(f_0,\boldsymbol{\kappa})$ defines a valid jointly Gaussian model. Then, for every $n\in\mathbb{N}$ and every $u_1,\ldots,u_n\in[t_0,t_1]$, the matrix
\begin{align*}
\boldsymbol{\Omega}_n=
\begin{pmatrix}
\mathbf{C}_n & \mathbf{G}_n\\
\mathbf{G}_n^\top & \boldsymbol{\Sigma}_{\kappa}
\end{pmatrix}
\end{align*}
is positive semidefinite. Since $\boldsymbol{\Sigma}_{\kappa}$ is positive definite, the standard Schur complement criterion for positive semidefinite block matrices yields
\begin{align*}
\boldsymbol{\Omega}_n\succeq 0
\quad\Longleftrightarrow\quad
\mathbf{C}_n-\mathbf{G}_n\boldsymbol{\Sigma}_{\kappa}^{-1}\mathbf{G}_n^\top\succeq 0.
\end{align*}
The finite-dimensional admissibility condition is therefore Schur-complement positivity.
The $(i,j)$ entry of the Schur complement is
\begin{align*}
C_\theta(u_i,u_j)-\boldsymbol{\gamma}(u_i)^\top\boldsymbol{\Sigma}_{\kappa}^{-1}\boldsymbol{\gamma}(u_j)
=
C_{\theta\mid\kappa}(u_i,u_j).
\end{align*}
Hence the matrix $\bigl(C_{\theta\mid\kappa}(u_i,u_j)\bigr)_{i,j=1}^n$ is positive semidefinite for every finite choice of points, so $C_{\theta\mid\kappa}$ is positive semidefinite on $[t_0,t_1]$.

Conversely, assume that $C_{\theta\mid\kappa}$ is positive semidefinite on $[t_0,t_1]$. Then for every $n\in\mathbb{N}$ and every $u_1,\ldots,u_n\in[t_0,t_1]$, the matrix
\begin{align*}
\mathbf{C}_n-\mathbf{G}_n\boldsymbol{\Sigma}_{\kappa}^{-1}\mathbf{G}_n^\top
=
\bigl(C_{\theta\mid\kappa}(u_i,u_j)\bigr)_{i,j=1}^n
\end{align*}
is positive semidefinite. Since $\boldsymbol{\Sigma}_{\kappa}$ is positive definite, the Schur complement criterion implies that $\boldsymbol{\Omega}_n\succeq 0$ for every finite set of evaluation points. Assign to
\[
\bigl(f_0(u_1),\ldots,f_0(u_n),\boldsymbol{\kappa}^\top\bigr)^\top
\]
the Gaussian law with the prescribed mean vector and covariance matrix $\boldsymbol{\Omega}_n$. These finite-dimensional Gaussian laws are consistent under deletion and permutation of evaluation points because the prescribed means and covariances restrict to the corresponding sub-vectors and principal submatrices. Kolmogorov's extension theorem then supplies a process on the product index set \([t_0,t_1]\) together with the \(r\) distinguished coordinates of \(\boldsymbol{\kappa}\). Since every finite subvector was assigned a Gaussian law, the resulting pair \((f_0,\boldsymbol{\kappa})\) is jointly Gaussian with the displayed mean and covariance blocks.
\end{proof}

\(C_{\theta\mid\kappa}\) is the residual anchor covariance after conditioning on the integration constants. The function \(\boldsymbol{\gamma}\) is the covariance link between \(f_0\) and \(\boldsymbol{\kappa}\), and \(\boldsymbol{\gamma}=\mathbf{0}\) recovers \cref{prop:gaussian-ensemble}. The next proposition propagates this conditional Gaussian structure from the anchor process to the full differential ladder.

\begin{proposition}[Conditional laws given the integration constants]\label{prop:appendix-conditional-laws}
Retain the standing assumptions, including the definition of \(\mathfrak f\) and joint Gaussianity of \((f_0,\boldsymbol{\kappa})\), and assume that \(\boldsymbol{\Sigma}_{\kappa}\) is positive definite. Write
\begin{align*}
\boldsymbol{\gamma}(s):=\Cov\bigl(f_0(s),\boldsymbol{\kappa}\bigr)=\bigl(\gamma_1(s),\ldots,\gamma_r(s)\bigr)^\top,
\end{align*}
and define
\begin{align*}
C_{\theta\mid\kappa}(s,t):=C_\theta(s,t)-\boldsymbol{\gamma}(s)^\top\boldsymbol{\Sigma}_{\kappa}^{-1}\boldsymbol{\gamma}(t),
\end{align*}
and, for \(p=-r,\ldots,r\),
\begin{align*}
(\mathcal A_{t_0}^{[p]}\boldsymbol{\gamma})(t):=\bigl((\mathcal A_{t_0}^{[p]}\gamma_1)(t),\ldots,(\mathcal A_{t_0}^{[p]}\gamma_r)(t)\bigr)^\top.
\end{align*}

\begin{enumerate}
\item Conditional on \(\boldsymbol{\kappa}\), the anchor process is Gaussian with conditional mean and covariance
\begin{align*}
\E\bigl[f_0(s)\mid \boldsymbol{\kappa}\bigr]
&=\mu(s)+\boldsymbol{\gamma}(s)^\top\boldsymbol{\Sigma}_{\kappa}^{-1}(\boldsymbol{\kappa}-\boldsymbol{\mu}_{\kappa}), && s\in[t_0,t_1],\\
\Cov\bigl(f_0(s),f_0(t)\mid \boldsymbol{\kappa}\bigr)
&=C_{\theta\mid\kappa}(s,t), && s,t\in[t_0,t_1].
\end{align*}

\item Conditional on \(\boldsymbol{\kappa}\), the family \(\mathfrak{f}\) is jointly Gaussian in the sense that, for every finite collection \((p_1,\tau_1),\ldots,(p_n,\tau_n)\) with \(p_i\in\{-r,\ldots,r\}\) and \(\tau_i\in[t_0,t_1]\), the vector
\begin{align*}
\bigl(f_{p_1}(\tau_1),\ldots,f_{p_n}(\tau_n)\bigr)^\top
\end{align*}
is multivariate Gaussian conditional on \(\boldsymbol{\kappa}\). Its conditional mean is
\begin{align*}
\E\bigl[f_p(t)\mid \boldsymbol{\kappa}\bigr]
&=(\mathcal A_{t_0}^{[p]}\mu)(t)+(\mathcal A_{t_0}^{[p]}\boldsymbol{\gamma})(t)^\top\boldsymbol{\Sigma}_{\kappa}^{-1}(\boldsymbol{\kappa}-\boldsymbol{\mu}_{\kappa})\\
&\quad+\mathbf{1}_{\{p>0\}}\sum_{j=1}^{p}\frac{(t-t_0)^{p-j}}{(p-j)!}\kappa_j,
\qquad p=-r,\ldots,r,\quad t\in[t_0,t_1],
\end{align*}
and its conditional covariance entries are
\begin{align*}
\Cov\bigl(f_p(s),f_q(t)\mid \boldsymbol{\kappa}\bigr)
&=(\mathcal A_{1,t_0}^{[p]}\mathcal A_{2,t_0}^{[q]}C_{\theta\mid\kappa})(s,t), && p,q\in\{-r,\ldots,r\},\quad s,t\in[t_0,t_1].
\end{align*}
\end{enumerate}
\end{proposition}

\begin{proof}
For every finite collection \(\tau_1,\ldots,\tau_n\), the vector
\begin{align*}
\bigl(f_0(\tau_1),\ldots,f_0(\tau_n),\boldsymbol{\kappa}^\top\bigr)^\top
\end{align*}
is multivariate Gaussian by assumption, with the block covariance matrix displayed above \cref{thm:appendix-schur}. Since \(\boldsymbol{\Sigma}_{\kappa}\) is positive definite, the formulas in part (i) are the standard conditional mean and conditional covariance of a Gaussian block vector, and consistency over finite collections yields the stated conditional Gaussian process.

For part (ii), fix an arbitrary finite collection \((p_1,\tau_1),\ldots,(p_n,\tau_n)\) with \(p_i\in\{-r,\ldots,r\}\) and \(\tau_i\in[t_0,t_1]\). For \(p\in\{-r,\ldots,r\}\), write
\[
A_p(t):=(\mathcal A_{t_0}^{[p]}f_0)(t),
\qquad
B_p(t):=\mathbf{1}_{\{p>0\}}\sum_{j=1}^{p}\frac{(t-t_0)^{p-j}}{(p-j)!}\kappa_j,
\]
so that \(f_p(t)=A_p(t)+B_p(t)\). Define
\begin{align*}
\mathbf A :=\bigl(A_{p_1}(\tau_1),\ldots,A_{p_n}(\tau_n)\bigr)^\top,\qquad \mathbf B :=\bigl(B_{p_1}(\tau_1),\ldots,B_{p_n}(\tau_n)\bigr)^\top.
\end{align*}
Conditional on \(\boldsymbol{\kappa}\), the vector \(\mathbf B\) is deterministic. All conditional randomness therefore comes from the vector \(\mathbf A\), whose entries are determined by \(f_0\) through differentiation, evaluation, or integration.
By the same approximation argument used in the proof of \cref{prop:gaussian-ensemble}, each \(A_{p_i}(\tau_i)\) is an \(L^2\)-limit of finite linear combinations of values of \(f_0\). Therefore \((\mathbf A^\top,\boldsymbol{\kappa}^\top)^\top\) is jointly Gaussian, because its approximating vectors are finite linear transformations of finite vectors of the form
\begin{align*}
\bigl(f_0(u_1),\ldots,f_0(u_M),\boldsymbol{\kappa}^\top\bigr)^\top,
\end{align*}
which are jointly Gaussian by assumption, and joint Gaussianity is preserved under \(L^2\)-limits by the Cram\'er--Wold argument used in \cref{prop:gaussian-ensemble}. Since \(\mathbf B\) is a deterministic linear transformation of \(\boldsymbol{\kappa}\), the vector
\begin{align*}
\bigl(\mathbf A+\mathbf B,\boldsymbol{\kappa}\bigr)
=
\bigl(f_{p_1}(\tau_1),\ldots,f_{p_n}(\tau_n),\boldsymbol{\kappa}\bigr)
\end{align*}
is jointly Gaussian. Conditioning on \(\boldsymbol{\kappa}\) therefore shows that
\begin{align*}
\bigl(f_{p_1}(\tau_1),\ldots,f_{p_n}(\tau_n)\bigr)^\top
\end{align*}
is multivariate Gaussian conditional on \(\boldsymbol{\kappa}\). Since the finite collection was arbitrary, the family \(\mathfrak f\) is jointly Gaussian conditional on \(\boldsymbol{\kappa}\).

It remains to identify the conditional mean and covariance entries. Fix \(p,q\in\{-r,\ldots,r\}\). By the mean calculation in the proof of \cref{prop:gaussian-ensemble} and by \cref{lem:appendix-cross-basic} applied componentwise,
\begin{align*}
\E[A_p(t)]&=(\mathcal A_{t_0}^{[p]}\mu)(t),\\
\Cov\bigl(A_p(t),\boldsymbol{\kappa}\bigr)
&=\bigl((\mathcal A_{t_0}^{[p]}\gamma_1)(t),\ldots,(\mathcal A_{t_0}^{[p]}\gamma_r)(t)\bigr)^\top
=(\mathcal A_{t_0}^{[p]}\boldsymbol{\gamma})(t).
\end{align*}
Hence the Gaussian conditioning formula gives
\begin{align*}
\E\bigl[A_p(t)\mid \boldsymbol{\kappa}\bigr]
&=
(\mathcal A_{t_0}^{[p]}\mu)(t)
+
(\mathcal A_{t_0}^{[p]}\boldsymbol{\gamma})(t)^\top
\boldsymbol{\Sigma}_{\kappa}^{-1}
(\boldsymbol{\kappa}-\boldsymbol{\mu}_{\kappa}),
\end{align*}
and adding the \(\boldsymbol{\kappa}\)-measurable term \(B_p(t)\) yields the stated conditional mean of \(f_p(t)\).

For the conditional covariance, the polynomial terms contribute no conditional randomness, so
\[
\Cov\bigl(f_p(s),f_q(t)\mid \boldsymbol{\kappa}\bigr)
=
\Cov\bigl(A_p(s),A_q(t)\mid \boldsymbol{\kappa}\bigr).
\]
Applying the Gaussian conditioning formula once more, and using the covariance calculation from the proof of \cref{prop:gaussian-ensemble} together with \cref{lem:appendix-cross-basic}, gives
\begin{align*}
\Cov\bigl(A_p(s),A_q(t)\mid \boldsymbol{\kappa}\bigr)
&=
\Cov\bigl(A_p(s),A_q(t)\bigr)
-
\Cov\bigl(A_p(s),\boldsymbol{\kappa}\bigr)
\boldsymbol{\Sigma}_{\kappa}^{-1}
\Cov\bigl(\boldsymbol{\kappa},A_q(t)\bigr)\\
&=
(\mathcal A_{1,t_0}^{[p]}\mathcal A_{2,t_0}^{[q]}C_\theta)(s,t)
-
(\mathcal A_{t_0}^{[p]}\boldsymbol{\gamma})(s)^\top
\boldsymbol{\Sigma}_{\kappa}^{-1}
(\mathcal A_{t_0}^{[q]}\boldsymbol{\gamma})(t).
\end{align*}
Because \(\boldsymbol{\Sigma}_{\kappa}^{-1}\) is constant, the signed-order operators act on the two variables separately. The regularity assumptions make these operators well defined on \(C_{\theta\mid\kappa}\), so the last display equals $(\mathcal A_{1,t_0}^{[p]}\mathcal A_{2,t_0}^{[q]}C_{\theta\mid\kappa})(s,t)$, which is the stated conditional covariance.
\end{proof}

\subsubsection{Constructive representations and coupled ensembles}
The admissibility criterion is complemented by two constructions. One defines \(\boldsymbol{\kappa}\) directly as Gaussian functionals of \(f_0\). The other uses a second process \(u_0\), jointly Gaussian with \(f_0\), whose derivative jet at \(t_0\) supplies \(\boldsymbol{\kappa}\). These constructions give the same covariance form but represent different sources of dependence.

\begin{proposition}[Linear functional construction]\label{prop:appendix-linear-functional}
Let \(w_1,\ldots,w_r\) be fixed deterministic functions in \(L^2([t_0,t_1])\), used as weights in the linear functionals below. Let $\boldsymbol{\varepsilon}:=(\varepsilon_1,\ldots,\varepsilon_r)^\top\sim N(\mathbf{0},\mathbf{T})$, where $\mathbf{T}\in\mathbb{R}^{r\times r}$ is positive semidefinite, and assume that $\boldsymbol{\varepsilon}$ is independent of $f_0$. Define, as $L^2$ random variables,
\begin{align*}
\kappa_m:=\int_{t_0}^{t_1}w_m(u)f_0(u)\mathrm{d}u+\varepsilon_m,
\qquad m=1,\ldots,r.
\end{align*}
Write $\boldsymbol{\kappa}:=(\kappa_1,\ldots,\kappa_r)^\top$, $\gamma_m(s):=\Cov(f_0(s),\kappa_m)$, and $(\boldsymbol{\Sigma}_{\kappa})_{mj}:=\Cov(\kappa_m,\kappa_j)$. Then $(f_0,\boldsymbol{\kappa})$ is jointly Gaussian, each \(\gamma_m\) belongs to \(C^r([t_0,t_1])\), and
\begin{align*}
\gamma_m(s)&=\int_{t_0}^{t_1}w_m(u)C_\theta(s,u)\mathrm{d}u, && m=1,\ldots,r,\quad s\in[t_0,t_1],\\
(\boldsymbol{\Sigma}_{\kappa})_{mj}&=\int_{t_0}^{t_1}\int_{t_0}^{t_1}w_m(u)C_\theta(u,v)w_j(v)\mathrm{d}u\mathrm{d}v+\mathbf{T}_{mj}, && m,j=1,\ldots,r.
\end{align*}
If, in addition, the resulting matrix $\boldsymbol{\Sigma}_{\kappa}$ is positive definite, then the compatibility condition in \cref{thm:appendix-schur} holds automatically.
\end{proposition}

\begin{proof}
Let $M:=\sup_{u\in[t_0,t_1]}\E[f_0(u)^2]<\infty$. Since $[t_0,t_1]$ has finite length, $w_m\in L^2([t_0,t_1])$ implies $w_m\in L^1([t_0,t_1])$. Hence the integral $\int_{t_0}^{t_1}w_m(u)f_0(u)\mathrm{d}u$ is well defined as an $L^2$ limit of step-function sums, with the bound
\[
\left\|\int_{t_0}^{t_1}w_m(u)f_0(u)\mathrm{d}u\right\|_{L^2(\Omega)}
\le M^{1/2}\|w_m\|_{L^1}.
\]
For any $\tau_1,\ldots,\tau_n\in[t_0,t_1]$, approximate all $r$ integrals simultaneously by such step-function sums, and approximate the resulting integrals of $f_0$ over subintervals by Riemann sums in $L^2$. The approximating vectors are linear transformations of finite Gaussian vectors formed from values of $f_0$ together with the independent Gaussian vector $\boldsymbol{\varepsilon}$. Their $L^2$ limits are therefore Gaussian by the Cram\'er--Wold argument, so $(f_0,\boldsymbol{\kappa})$ is jointly Gaussian.

The covariance identities follow by applying Cauchy--Schwarz to pass covariance through the $L^2$ limits and then using Fubini and independence of $\boldsymbol{\varepsilon}$ from $f_0$:
\begin{align*}
\gamma_m(s)
&=\Cov\bigl(f_0(s),\kappa_m\bigr)
=\int_{t_0}^{t_1}w_m(u)\Cov\bigl(f_0(s),f_0(u)\bigr)\mathrm{d}u
=\int_{t_0}^{t_1}w_m(u)C_\theta(s,u)\mathrm{d}u,\\
(\boldsymbol{\Sigma}_{\kappa})_{mj}
&=\Cov(\kappa_m,\kappa_j)
=\int_{t_0}^{t_1}\int_{t_0}^{t_1}w_m(u)C_\theta(u,v)w_j(v)\mathrm{d}u\mathrm{d}v+\mathbf{T}_{mj}.
\end{align*}
For $a=0,\ldots,r$, dominated convergence gives
\[
D_s^a\gamma_m(s)=\int_{t_0}^{t_1}w_m(u)D_1^aC_\theta(s,u)\mathrm{d}u,
\]
because $D_1^aC_\theta$ is continuous and bounded on the compact square and $w_m\in L^1([t_0,t_1])$. Thus $\gamma_m\in C^r([t_0,t_1])$.

If $\boldsymbol{\Sigma}_{\kappa}$ is positive definite, then the construction gives a jointly Gaussian pair $(f_0,\boldsymbol{\kappa})$ with the displayed covariance blocks. Thus condition~(i) of \cref{thm:appendix-schur} holds, and the compatibility claim follows.
\end{proof}

The second construction makes \(\boldsymbol{\kappa}\) the derivative jet of a jointly Gaussian process \(u_0\). The cross-covariances are then determined by \(C_{fg}\) and \(C_g\).

\begin{proposition}[Construction from a second Gaussian process]\label{prop:appendix-second-process}
Let $u_0\sim\mathcal{GP}(\nu,C_g)$ be a second Gaussian process on $[t_0,t_1]$, jointly Gaussian with $f_0$, and write
\[
C_{fg}(s,t):=\Cov\bigl(f_0(s),u_0(t)\bigr).
\]
Assume that \(u_0\) is mean-square differentiable up to order \(r-1\) and that, for \(m=1,\ldots,r\), the right-sided endpoint mean-square derivative \(D_t^{m-1}u_0(t_0)\) exists. Define
\[
\kappa_m:=D_t^{m-1}u_0(t_0),\qquad m=1,\ldots,r,
\]
set $\boldsymbol{\kappa}:=(\kappa_1,\ldots,\kappa_r)^\top$, and write
\[
\gamma_m(s):=\Cov\bigl(f_0(s),\kappa_m\bigr),\qquad
(\boldsymbol{\Sigma}_{\kappa})_{mj}:=\Cov\bigl(\kappa_m,\kappa_j\bigr).
\]
Then \((f_0,\boldsymbol{\kappa})\) is jointly Gaussian and
\begin{align*}
\gamma_m(s)&=D_2^{m-1}C_{fg}(s,t)\big|_{t=t_0}, && m=1,\ldots,r,\quad s\in[t_0,t_1],\\
(\boldsymbol{\Sigma}_{\kappa})_{mj}&=D_1^{m-1}D_2^{j-1}C_g(s,t)\big|_{s=t=t_0}, && m,j=1,\ldots,r.
\end{align*}
In particular, this covariance specification is admissible because the pair \((f_0,\boldsymbol{\kappa})\) is constructed jointly on a common probability space.
\end{proposition}

\begin{proof}
For each $m$, choose a one-sided finite-difference approximant $Q_{m,h}(t_0)$ to $D_t^{m-1}u_0(t_0)$, based on values of $u_0$ near $t_0$, such that $Q_{m,h}(t_0)\to\kappa_m$ in $L^2$ as $h\downarrow0$; for $m=1$ take $Q_{1,h}(t_0)=u_0(t_0)$. For any finite collection $s_1,\ldots,s_n\in[t_0,t_1]$, the approximating vector
\[
\bigl(f_0(s_1),\ldots,f_0(s_n),Q_{1,h}(t_0),\ldots,Q_{r,h}(t_0)\bigr)
\]
is jointly Gaussian because $(f_0,u_0)$ is jointly Gaussian. Its $L^2$ limit is
\[
\bigl(f_0(s_1),\ldots,f_0(s_n),\kappa_1,\ldots,\kappa_r\bigr),
\]
which is Gaussian by the Cram\'er--Wold argument. Hence $(f_0,\boldsymbol{\kappa})$ is jointly Gaussian.

For fixed $s\in[t_0,t_1]$, Cauchy--Schwarz gives
\[
\Cov\bigl(f_0(s),Q_{m,h}(t_0)\bigr)\to \Cov\bigl(f_0(s),\kappa_m\bigr).
\]
The covariance on the left is the same finite-difference operator applied to the second argument of $C_{fg}(s,t)$ at $t=t_0$. Passing to the limit gives
\[
\gamma_m(s)=\Cov\bigl(f_0(s),\kappa_m\bigr)=D_2^{m-1}C_{fg}(s,t)\big|_{t=t_0}.
\]
Similarly, if $Q_{j,\eta}(t_0)\to\kappa_j$ is the corresponding approximant for $D_t^{j-1}u_0(t_0)$, then
\[
\Cov\bigl(Q_{m,h}(t_0),Q_{j,\eta}(t_0)\bigr)\to \Cov\bigl(\kappa_m,\kappa_j\bigr)
\]
as $h,\eta\downarrow0$. The left side is the mixed finite-difference quotient of $C_g(s,t)$ at $(t_0,t_0)$, and hence
\[
(\boldsymbol{\Sigma}_{\kappa})_{mj}
=\Cov\bigl(\kappa_m,\kappa_j\bigr)
=D_1^{m-1}D_2^{j-1}C_g(s,t)\big|_{s=t=t_0}.
\]
The final admissibility claim follows from the joint construction of $(f_0,\boldsymbol{\kappa})$.
\end{proof}

The next corollary specializes \cref{prop:appendix-second-process} to the case \(u_0=f_0\), with derivatives interpreted in the same endpoint mean-square sense.

\begin{corollary}[Derivative jet construction from the anchor process]\label{cor:appendix-self-referential}
If, in \cref{prop:appendix-second-process}, one takes \(u_0=f_0\), then
\begin{align*}
\kappa_m&=D_t^{m-1}f_0(t_0), && m=1,\ldots,r,\\
\gamma_m(s)&=D_2^{m-1}C_\theta(s,t)\big|_{t=t_0}, && m=1,\ldots,r,\quad s\in[t_0,t_1],
\end{align*}
and
\begin{align*}
(\boldsymbol{\Sigma}_{\kappa})_{mj}=D_1^{m-1}D_2^{j-1}C_\theta(s,t)\big|_{s=t=t_0},
\qquad m,j=1,\ldots,r.
\end{align*}
In that case,
\begin{align*}
f_p(t)=(\mathcal A_{t_0}^{[p]}f_0)(t)+\sum_{j=1}^{p}\frac{(t-t_0)^{p-j}}{(p-j)!}D_t^{j-1}f_0(t_0),
\qquad p=1,\ldots,r,\quad t\in[t_0,t_1].
\end{align*}
\end{corollary}

\begin{proof}
Taking $u_0=f_0$ in \cref{prop:appendix-second-process} gives $C_{fg}=C_g=C_\theta$, and therefore
\[
\kappa_m=D_t^{m-1}f_0(t_0),\qquad
\gamma_m(s)=D_2^{m-1}C_\theta(s,t)\big|_{t=t_0},
\]
together with
\[
(\boldsymbol{\Sigma}_{\kappa})_{mj}=D_1^{m-1}D_2^{j-1}C_\theta(s,t)\big|_{s=t=t_0}.
\]
Since the defining positive-index representation in the dependent extension is unchanged from the main construction, the positive-index representation remains
\[
f_p(t)=(\mathcal A_{t_0}^{[p]}f_0)(t)+\sum_{j=1}^{p}\frac{(t-t_0)^{p-j}}{(p-j)!}\kappa_j,
\qquad p=1,\ldots,r.
\]
Substituting \(\kappa_j=D_t^{j-1}f_0(t_0)\) yields the stated formula.
\end{proof}

The preceding constructions define dependence at the level of \((f_0,\boldsymbol{\kappa})\). We now track how that dependence propagates through the ensemble.

\begin{proposition}[Cross-covariances across ensemble levels]\label{prop:appendix-cross-chain}
Retain the standing assumptions of this subsection. For each \(p\in\{-r,\ldots,r\}\), define the vector-valued cross-covariance function
\begin{align*}
\boldsymbol{\Gamma}_p(t):=\Cov\bigl(f_p(t),\boldsymbol{\kappa}\bigr)
=\bigl(\Gamma_{p,1}(t),\ldots,\Gamma_{p,r}(t)\bigr)^\top,
\end{align*}
where the covariance with \(\boldsymbol{\kappa}\) is understood componentwise. Then, for \(m=1,\ldots,r\),
\begin{align*}
\Gamma_{p,m}(t)
&=(\mathcal A_{t_0}^{[p]}\gamma_m)(t)+\mathbf{1}_{\{p>0\}}\sum_{j=1}^{p}\frac{(t-t_0)^{p-j}}{(p-j)!}(\boldsymbol{\Sigma}_{\kappa})_{jm},
\qquad p=-r,\ldots,r,
\end{align*}
and, for \(p=-r+1,\ldots,r\) and \(t\in(t_0,t_1)\),
\begin{align*}
\frac{d}{dt}\Gamma_{p,m}(t)=\Gamma_{p-1,m}(t).
\end{align*}
Thus \(\Gamma_{p,m}\) is obtained by applying the same signed-order operator to \(\gamma_m\), with an added polynomial correction from \(\boldsymbol{\Sigma}_{\kappa}\) when \(p>0\). More precisely, if \(p=-q<0\), then \(\Gamma_{p,m}=\gamma_m^{(q)}\), so \(\gamma_m\) is recovered from \(\Gamma_{p,m}\) by \(q\) integrations together with the boundary values \(\gamma_m(t_0),\gamma_m'(t_0),\ldots,\gamma_m^{(q-1)}(t_0)\). If \(p=0\), then \(\Gamma_{0,m}=\gamma_m\). If \(p>0\), then
\begin{align*}
\Gamma_{p,m}(t)-\sum_{j=1}^{p}\frac{(t-t_0)^{p-j}}{(p-j)!}(\boldsymbol{\Sigma}_{\kappa})_{jm}
=(\mathcal A_{t_0}^{[p]}\gamma_m)(t),
\end{align*}
so \(\gamma_m\) is recovered by subtracting this polynomial correction term and differentiating \(p\) times.
\end{proposition}

\begin{proof}
For fixed \(m\), the right-hand side is the same positive-level deterministic construction as in the main ensemble, with anchor element \(\gamma_m\) and boundary constants \((\boldsymbol{\Sigma}_\kappa)_{1m},\ldots,(\boldsymbol{\Sigma}_\kappa)_{rm}\).

By the defining representation of \(f_p\) inherited from \cref{prop:gaussian-ensemble},
\begin{align*}
\Gamma_{p,m}(t)
=\Cov\bigl((\mathcal A_{t_0}^{[p]}f_0)(t),\kappa_m\bigr)+\mathbf{1}_{\{p>0\}}\sum_{j=1}^{p}\frac{(t-t_0)^{p-j}}{(p-j)!}\Cov(\kappa_j,\kappa_m).
\end{align*}
Now \cref{lem:appendix-cross-basic} gives
\begin{align*}
\Cov\bigl((\mathcal A_{t_0}^{[p]}f_0)(t),\kappa_m\bigr)=(\mathcal A_{t_0}^{[p]}\gamma_m)(t),
\end{align*}
and \(\Cov(\kappa_j,\kappa_m)=(\boldsymbol{\Sigma}_{\kappa})_{jm}\), which proves the displayed formula for \(\Gamma_{p,m}\).

For \(t\in(t_0,t_1)\), differentiating the displayed expression yields
\begin{align*}
\frac{d}{dt}\Gamma_{p,m}(t)=\Gamma_{p-1,m}(t),
\end{align*}
since \(\frac{d}{dt}(\mathcal A_{t_0}^{[p]}\gamma_m)(t)=(\mathcal A_{t_0}^{[p-1]}\gamma_m)(t)\) and the derivative of the polynomial correction term is exactly the polynomial correction term at level \(p-1\).

For the recovery statement, if \(p=-q<0\), then \(\Gamma_{p,m}=\gamma_m^{(q)}\), so repeated integration together with the boundary values \(\gamma_m(t_0),\gamma_m'(t_0),\ldots,\gamma_m^{(q-1)}(t_0)\) recovers \(\gamma_m\). If \(p=0\), there is nothing to prove because \(\Gamma_{0,m}=\gamma_m\). If \(p>0\), subtracting the polynomial correction term gives \((\mathcal A_{t_0}^{[p]}\gamma_m)(t)=\mathcal I_{t_0}^{p}\gamma_m(t)\), and differentiating \(p\) times recovers \(\gamma_m\).
\end{proof}

Although deterministic, the cross-covariance chain \(\bigl(\Gamma_{-r,m},\ldots,\Gamma_{r,m}\bigr)\) evolves like an ensemble for each fixed \(m\). It can be viewed as a weak second ensemble because it follows the same derivative--integral transformation law without adding randomness.

A coupled second differential ensemble gives a stronger probabilistic realization: its derivative branch at \(t_0\) supplies the integration constants of \(\mathfrak f\). The following statement gives the corresponding ensemble-level construction.

\begin{proposition}[Coupled second differential ensemble]\label{prop:appendix-coupled-second-ensemble}
Assume the setting of \cref{prop:appendix-second-process}, and assume in addition that the anchor process \(u_0\) satisfies the order-\(r\) analogue of \cref{ass:ms-regularity}: its mean belongs to \(C^r([t_0,t_1])\) and its covariance function \(C_g\) belongs to \(C^{r,r}([t_0,t_1]^2)\). Let \(\boldsymbol{\chi}:=(\chi_1,\ldots,\chi_r)^\top\) be any Gaussian vector jointly Gaussian with \((f_0,u_0)\), and define
\begin{align*}
u_\ell(t):=(\mathcal A_{t_0}^{[\ell]}u_0)(t)+\mathbf{1}_{\{\ell>0\}}\sum_{m=1}^{\ell}\frac{(t-t_0)^{\ell-m}}{(\ell-m)!}\chi_m,
\qquad \ell=-r,\ldots,r.
\end{align*}
Then \(\mathfrak{u}:=(u_{-r},\ldots,u_r)\) satisfies the mean-square derivative and boundary clauses of \cref{def:ensemble}. If, in addition, \(\nu\equiv 0\) and \(\E[\boldsymbol{\chi}]=\mathbf 0\), then \(\mathfrak{u}\) is a centered differential ensemble in the sense of \cref{def:ensemble}. Moreover, by \cref{prop:appendix-second-process}, the integration constants of \(\mathfrak{f}\) satisfy
\begin{align*}
\kappa_m=D_t^{m-1}u_0(t_0),
\qquad m=1,\ldots,r.
\end{align*}
Consequently, for \(p=1,\ldots,r\),
\begin{align*}
f_p(t)=(\mathcal A_{t_0}^{[p]}f_0)(t)+\sum_{j=1}^{p}\frac{(t-t_0)^{p-j}}{(p-j)!}D_t^{j-1}u_0(t_0).
\end{align*}
\end{proposition}

\begin{proof}
The analogue of \cref{ass:ms-regularity} for \(u_0\) gives the derivative processes \(u_{-r},\ldots,u_{-1}\), bounded second moments on the compact interval, and the same mean-square continuity used for the positive-index integral branch in \cref{prop:gaussian-ensemble}. The finite polynomial terms involving \(\boldsymbol{\chi}\) are square-integrable. The proof of the derivative recursion and boundary identities in \cref{prop:gaussian-ensemble} therefore applies verbatim to \(u_0\) and \(\boldsymbol{\chi}\); that part of the argument does not use independence between the anchor process and the boundary vector. Hence \(\mathfrak u\) satisfies the mean-square derivative relations and \(u_\ell(t_0)=\chi_\ell\) for \(\ell=1,\ldots,r\). If \(\nu\equiv0\) and \(\E[\boldsymbol{\chi}]=\mathbf0\), the derivative, anchor, integral, and polynomial components all have mean zero, so \(\mathfrak u\) is centered.

The displayed identity for \(\kappa_m\) is the defining identity in \cref{prop:appendix-second-process}. Substituting it into the positive-index representation of \(f_p\) gives the stated formula.
\end{proof}

The final result links the coupled second ensemble back to \(\Gamma_{p,m}\): the full cross-covariance with the derivative branch specializes at \(t_0\) to the chain above.

\begin{proposition}[Cross-covariance with the derivative branch of the second ensemble]\label{prop:appendix-crosscov-second-ensemble}
Assume the setting of \cref{prop:appendix-coupled-second-ensemble}, and interpret endpoint derivatives with the same convention as in \cref{prop:appendix-second-process}. Then, for every \(p\in\{-r,\ldots,r\}\), every \(\ell\in\{0,\ldots,r\}\), and every \(s,t\in[t_0,t_1]\),
\begin{align*}
\Cov\bigl(f_p(s),u_{-\ell}(t)\bigr)
&=(\mathcal A_{1,t_0}^{[p]}D_2^{\ell}C_{fg})(s,t)\\
&\quad+\mathbf{1}_{\{p>0\}}\sum_{j=1}^{p}\frac{(s-t_0)^{p-j}}{(p-j)!}D_1^{j-1}D_2^{\ell}C_g(u,v)\big|_{u=t_0,v=t}.
\end{align*}
In particular, with \(\gamma_m\), \(\Gamma_{p,m}\), and \(\boldsymbol{\Sigma}_{\kappa}\) as in \cref{prop:appendix-second-process,prop:appendix-cross-chain}, for every \(p\in\{-r,\ldots,r\}\), every \(m\in\{1,\ldots,r\}\), and every \(s\in[t_0,t_1]\),
\begin{align*}
\Cov\bigl(f_p(s),u_{-(m-1)}(t_0)\bigr)
&=(\mathcal A_{t_0}^{[p]}\gamma_m)(s)+\mathbf{1}_{\{p>0\}}\sum_{j=1}^{p}\frac{(s-t_0)^{p-j}}{(p-j)!}(\boldsymbol{\Sigma}_{\kappa})_{jm}\\
&=\Gamma_{p,m}(s).
\end{align*}
\end{proposition}

\begin{proof}
Write \(f_p(s)=A_p(s)+B_p(s)\), where
\begin{align*}
A_p(s)=(\mathcal A_{t_0}^{[p]}f_0)(s),
\qquad
B_p(s)=\mathbf{1}_{\{p>0\}}\sum_{j=1}^{p}\frac{(s-t_0)^{p-j}}{(p-j)!}\kappa_j.
\end{align*}
For fixed \(\ell\) and \(t\), the mean-square derivative identity for \(u_0\) gives
\[
\Cov\bigl(f_0(v),u_{-\ell}(t)\bigr)=D_2^\ell C_{fg}(v,t),
\qquad v\in[t_0,t_1].
\]
Applying the same covariance--integration and covariance--differentiation interchanges as in \cref{lem:appendix-cross-basic}, now with the square-integrable variable \(u_{-\ell}(t)\) in place of \(\kappa_m\), yields
\[
\Cov\bigl(A_p(s),u_{-\ell}(t)\bigr)
=(\mathcal A_{1,t_0}^{[p]}D_2^\ell C_{fg})(s,t).
\]
For the polynomial part,
\begin{align*}
\Cov\bigl(B_p(s),u_{-\ell}(t)\bigr)
&=\mathbf{1}_{\{p>0\}}\sum_{j=1}^{p}\frac{(s-t_0)^{p-j}}{(p-j)!}\Cov\bigl(\kappa_j,u_{-\ell}(t)\bigr)\\
&=\mathbf{1}_{\{p>0\}}\sum_{j=1}^{p}\frac{(s-t_0)^{p-j}}{(p-j)!}D_1^{j-1}D_2^\ell C_g(u,v)\big|_{u=t_0,v=t},
\end{align*}
where the last equality uses \(\kappa_j=D_t^{j-1}u_0(t_0)\) and the same finite-difference covariance argument as in \cref{prop:appendix-second-process}. Adding the two terms proves the general formula.

Setting \(\ell=m-1\) and \(t=t_0\) gives \(u_{-(m-1)}(t_0)=\kappa_m\). Moreover,
\[
D_2^{m-1}C_{fg}(s,t)\big|_{t=t_0}=\gamma_m(s),
\qquad
D_1^{j-1}D_2^{m-1}C_g(u,v)\big|_{u=v=t_0}=(\boldsymbol{\Sigma}_{\kappa})_{jm}.
\]
The displayed specialization is therefore exactly the formula for \(\Gamma_{p,m}\) in \cref{prop:appendix-cross-chain}.
\end{proof}

The construction has three cases. If \(u_0=f_0\), no additional randomness is introduced. If \(f_0\) and \(u_0\) are independent, then \(C_{fg}\equiv0\), \(\boldsymbol{\gamma}\equiv0\), and the independent model in \cref{prop:gaussian-ensemble} is recovered. If they are distinct and dependent, the constants encode external information coupled to \(f_0\) through \(C_{fg}\).

\subsection{Laplacian--Dirichlet eigenbasis implementation for the differential ensemble}\label{appendix:marginal-finite-basis-likelihood}

This subsection records the Stan implementation used for the single-curve differential ensemble. The notation follows the main text. The non-marginalized model samples the finite-rank coefficients directly. The marginalized model integrates out those coefficients and reconstructs them in generated quantities.

Let \(Y=(Y_1,\ldots,Y_n)^\top\) denote the anchor observations at times \(t_1,\ldots,t_n\), and let \(y\) be the observed realization. Let \(\tau_1,\ldots,\tau_G\) be prediction locations. The computational domain is \([-L,L]\), with \(t_i,\tau_g,t_0\in[-L,L]\). The retained eigensystem is the one used in \cref{sec:finite-rank},
\[
\lambda_{k,L}=\left(\frac{k\pi}{2L}\right)^2,
\qquad
\phi_{k,L}(t)=L^{-1/2}\sin\left(\frac{k\pi(t+L)}{2L}\right),
\qquad
k=1,\ldots,K.
\]
The requested order \(r\) must be admissible for the kernel. Squared exponential kernels allow every finite \(r\). A Mat\'ern kernel with smoothness \(\nu\) is allowed only when \(r<\nu\).

For \(p\in\{-r,\ldots,r\}\), define the transformed basis functions as in the main text,
\[
\psi_k^{[p]}(t)=(\mathcal A_{t_0}^{[p]}\phi_{k,L})(t).
\]
With finite-rank coefficients \(\beta=(\beta_1,\ldots,\beta_K)^\top\), set
\[
g_p(t)=\sum_{k=1}^K\beta_k\psi_k^{[p]}(t).
\]
The implemented ensemble is
\[
f_p(t)=
\begin{cases}
g_p(t), & p\le 0,\\[4pt]
g_p(t)+\displaystyle\sum_{\ell=1}^p
\kappa_\ell\,\frac{(t-t_0)^{p-\ell}}{(p-\ell)!}, & p>0,
\end{cases}
\]
so \(f_p(t_0)=\kappa_p\) for \(p=1,\ldots,r\). The observations satisfy
\[
Y_i\mid f_0,\sigma_y \overset{\indep}{\sim}
\mathcal N(f_0(t_i),\sigma_y^2),
\qquad i=1,\ldots,n.
\]

Let \(\Phi_{ik}=\phi_{k,L}(t_i)\). Write \(S^{\mathrm{unit}}_\rho\) for the spectral density of the selected covariance family with unit marginal amplitude and length-scale \(\rho\), and define
\[
v_k(\rho)=S^{\mathrm{unit}}_\rho\left(\sqrt{\lambda_{k,L}}\right)^{1/2}.
\]
The observed-design normalization is the average finite-basis prior variance at the observed inputs for a unit-amplitude process:
\[
\bar V_\rho
=
\max\left(
\frac1n\sum_{i=1}^n\sum_{k=1}^K v_k(\rho)^2\Phi_{ik}^2,\,
10^{-10}
\right).
\]
Thus \(\bar V_\rho\) depends only on \(\rho\), the observed input locations, and the chosen finite basis, and it rescales the process amplitude below so that \(1-\eta\) represents the latent-process variance fraction on the observed design.
Let \(\vartheta=(\sigma_{\mathrm{tot}},\eta,\rho)\), where \(\sigma_{\mathrm{tot}}>0\) is the total response scale and \(\eta\in(0,1)\) is the residual variance fraction. The non-marginalized parameterization uses
\[
\zeta\sim\mathcal N_K(0,I_K),
\qquad
\beta_k=\alpha_\vartheta v_k(\rho)\zeta_k,
\qquad
\alpha_\vartheta=\sigma_{\mathrm{tot}}\sqrt{\frac{1-\eta}{\bar V_\rho}},
\qquad
\sigma_y=\sigma_{\mathrm{tot}}\sqrt\eta.
\]
Equivalently, conditional on \(\vartheta\),
\[
\beta\sim
\mathcal N_K(0,D_\vartheta),
\qquad
D_\vartheta
=
\sigma_{\mathrm{tot}}^2
\diag\left(\frac{1-\eta}{\bar V_\rho}v_1(\rho)^2,\ldots,
\frac{1-\eta}{\bar V_\rho}v_K(\rho)^2\right),
\]
and
\[
Y\mid \beta,\vartheta
\sim
\mathcal N_n(\Phi\beta,\sigma_{\mathrm{tot}}^2\eta I_n).
\]

The priors are computed from \(y\) and the observation times. Define the empirical response scale as
\[
s_y=
\max\left\{
\operatorname{median}_{i}|y_i-\operatorname{median}_j y_j|,
\operatorname{sd}(y),
10^{-3}
\right\}.
\]
Let \(\Delta_{\mathrm{med}}\) be the median gap between sorted observation times, with fallback \(0.05L\) when \(n=1\). The length-scale prior is
\[
\log\rho\sim\mathcal N(m_\rho,s_\rho^2),
\qquad
m_\rho=\frac12(\log q_{\rho,\mathrm{lo}}+\log q_{\rho,\mathrm{hi}}),
\qquad
s_\rho=
\frac{\log q_{\rho,\mathrm{hi}}-\log q_{\rho,\mathrm{lo}}}{2\Phi^{-1}(0.95)},
\]
where
\[
q_{\rho,\mathrm{hi}}=2L,
\qquad
q_{\rho,\mathrm{lo}}=\min\{\max(\Delta_{\mathrm{med}},0.05L),L\}.
\]
The remaining priors are
\[
\sigma_{\mathrm{tot}}\sim t^+_4(0,s_y),
\qquad
\eta\sim\operatorname{Beta}(2,2).
\]
The implementation uses independent normal priors for the integration constants,
\[
\kappa_j\sim\mathcal N(\mu_{\kappa,j},\tau_{\kappa,j}^2),
\qquad j=1,\ldots,r.
\]
The values \(\mu_{\kappa,j}\) and \(\tau_{\kappa,j}\) may be supplied as data. If the default prior is requested, they are set to
\[
\mu_{\kappa,j}=0,
\qquad
\tau_{\kappa,j}=\frac12s_yL^j,
\qquad j=1,\ldots,r.
\]
Because only \(f_0\) is observed in this model, independent integration constants do not enter the anchor likelihood and remain prior-driven unless separate boundary information is added.

The marginal implementation uses the same priors and integrates out \(\zeta\). Define
\[
A_\vartheta
=
\Phi\diag\left(\sqrt{\frac{1-\eta}{\bar V_\rho}}\,v(\rho)\right),
\qquad
S_\vartheta=\eta I_n+A_\vartheta A_\vartheta^\top .
\]
Then
\[
Y\mid\vartheta
\sim
\mathcal N_n(0,\sigma_{\mathrm{tot}}^2S_\vartheta),
\]
with observation-space log likelihood
\[
\ell_n(\vartheta)
=
-\frac12\left(
n\log(2\pi)
+2n\log\sigma_{\mathrm{tot}}
+\log |S_\vartheta|
+\sigma_{\mathrm{tot}}^{-2}y^\top S_\vartheta^{-1}y
\right).
\]
This is an exact rewrite of the explicit coefficient model. It removes \(\zeta\) from the sampled parameter block. Since \(\boldsymbol\kappa\) is independent of the anchor likelihood, the marginal implementation draws both \(\boldsymbol\kappa\) and \(\beta\) later in generated quantities.

The marginal likelihood is evaluated by one of two algebraically equivalent branches.
The dense branch factors the \(n\times n\) matrix \(S_\vartheta\), with time cost \(O(n^2K+n^3)\) and memory cost \(O(n^2)\). When \(K<n\) and \(\eta>0\), the same likelihood can be evaluated through
\[
B_\vartheta=I_K+\eta^{-1}A_\vartheta^\top A_\vartheta,
\qquad
|S_\vartheta|=\eta^n|B_\vartheta|.
\]
If \(L_B L_B^\top=B_\vartheta\), then
\[
y^\top S_\vartheta^{-1}y
=
\eta^{-1}y^\top y
-
\left\|
L_B^{-1}\frac{A_\vartheta^\top y}{\eta}
\right\|^2 .
\]
This coefficient-space branch costs \(O(nK^2+K^3)\) time and \(O(K^2)\) memory. The implementation uses it only when \(K<n\) and \(\eta>10^{-8}\), otherwise it uses the dense branch. No jitter is added, so both branches evaluate the same finite-basis likelihood.

After fitting the marginal model, finite-basis coefficients are drawn from \(\beta\mid y,\vartheta\). With
\[
\Sigma_Y=\Phi D_\vartheta\Phi^\top+\sigma_{\mathrm{tot}}^2\eta I_n,
\]
Gaussian conditioning gives
\[
\beta\mid y,\vartheta
\sim
\mathcal N_K\!\left(
D_\vartheta\Phi^\top\Sigma_Y^{-1}y,\,
D_\vartheta-D_\vartheta\Phi^\top\Sigma_Y^{-1}\Phi D_\vartheta
\right).
\]
The implementation samples this law with Matheron's rule. Draw
\[
\beta_{\mathrm{prior}}\sim\mathcal N_K(0,D_\vartheta),
\qquad
\epsilon_{\mathrm{prior}}\sim\mathcal N_n(0,\sigma_{\mathrm{tot}}^2\eta I_n),
\qquad
y_{\mathrm{prior}}=\Phi\beta_{\mathrm{prior}}+\epsilon_{\mathrm{prior}},
\]
and set
\[
\beta_{\mathrm{draw}}
=
\beta_{\mathrm{prior}}
+D_\vartheta\Phi^\top\Sigma_Y^{-1}(y-y_{\mathrm{prior}}).
\]
This draw has distribution \(\beta\mid y,\vartheta\) and requires only an \(n\times n\) solve. The generated quantities then draw \(\kappa_j\) from the independent normal prior above and assemble \(f_{-r},\ldots,f_r\) from the transformed basis matrices \(\psi_k^{[p]}\) and the polynomial integration-constant terms.

\subsection{TARTARE calibration procedure}\label{appendix:tartare-procedure}

\begin{tartarenote}{Inputs}
\begin{itemize}[nosep]
\item kernel specification: covariance family and fixed smoothness parameters, which determine the spectral density \(S_\theta\) and the exact covariance blocks induced by \(C_\theta\);
\item representation order \(r\) and monitoring set \(\mathcal M\subseteq\{-r,\ldots,r\}\), where \(p\in\mathcal M\) denotes the ensemble level \(f_p\);
\item normalized length-scale grid \(\mathcal U=(u_i)\), where \(u_i=\rho_iW^{-1}\) and \(W\) is the observation window half-width;
\item candidate range grid \(\mathcal C\), grid size \(G\), floor \(c_{\mathrm{floor}}\), and tolerances \(\varepsilon_{\mathrm{tail}}\), \(\tau_{\mathrm{joint}}\), \(\tau_{\mathrm{pair}}\), \(\tau_{\mathrm{var}}\), and optional \(\tau_{t_0}\).
\end{itemize}
\end{tartarenote}

\begin{tartarenote}{Definitions and conventions}
Write \(\theta\) for the fixed kernel specification together with a length-scale \(\rho\), and set the global marginal variance scale to one, since it cancels in the normalized calibration criteria. In the v1 implementation, the spectral tail order is the row order \(q_\star=r\), not re-estimated from the monitored subset \(\mathcal M\). For derivative-only rows, the row order equals the derivative order being calibrated, so \(q_\star=q\); for full differential-ensemble rows, \(q_\star=r\). Define
\[
R_{q_\star}(\Omega;\rho)
=
\int_{\Omega}^{\infty}\omega^{2q_\star}S_\theta(\omega;\rho)\,\diff\omega
\left(\int_0^{\infty}\omega^{2q_\star}S_\theta(\omega;\rho)\,\diff\omega\right)^{-1}.
\]
Let \(\mathcal G_{\mathrm{joint}}\) be the \(G\)-point grid on \([-(1-\alpha_r),1-\alpha_r]\), where \(\alpha_r=0.15\) for marginal Mat\'ern rows and otherwise \(\alpha_r=\min(0.05+0.02r,0.20)\). Let \(\mathcal G_{\mathrm{full}}\) be the untrimmed \(G\)-point grid on \([-1,1]\), augmented by \(t_0\). For integral rows, \(t_0=-1\), \(\delta_{t_0}=0.05\), and at least \(n_{t_0}=3\) near-\(t_0\) points are required.
\end{tartarenote}

\begin{tartarebox}{Offline calibration}
\small
\begin{enumerate}[label=\textbf{O\arabic*.},leftmargin=1.8em,itemsep=4pt,topsep=2pt]
\item Set \(m\) from the scale-normalized tail rule with v1 safety factor \(s=1.25\), so \(R_{q_\star}\!\left(\pi m(2\rho)^{-1};\rho\right)\leq \varepsilon_{\mathrm{tail}}\).
\item Calibrate over the prespecified normalized grid \(\mathcal U=(u_i)\); no length-scale is estimated in the offline step.
\item For each \(u_i\in\mathcal U\), set the calibration scale to \(W=1\), so the grid value \(u_i=\rho_iW^{-1}\) corresponds to \(\rho_i=u_i\).
\item Construct the exact joint covariance \(K_{\mathrm{ex}}\) on \(\mathcal G_{\mathrm{joint}}\), using closed-form covariance blocks induced by \(C_\theta\) for all ordered pairs of monitored levels in \(\mathcal M\). The covariance contribution from integration constants is retained exactly in the model and is excluded from this HSGP calibration error. If \(\mathcal M\) includes positive levels, also construct \(K_{\mathrm{ex}}^{\mathrm{full}}\) on \(\mathcal G_{\mathrm{full}}\) for the near-\(t_0\) check.
\item For each \(c\in\mathcal C\), set \(L=cW\), \(K=\left\lceil mcu_i^{-1}\right\rceil\), and \(\Omega_K=\pi K(2L)^{-1}\). Construct the HSGP joint covariance \(K_{\mathrm H}\) on \(\mathcal G_{\mathrm{joint}}\), using the transformed HSGP basis for each level in \(\mathcal M\) and the spectral weights \(S_\theta(\omega_k;\rho_i)\).
\item Let \(\eta=10^{-12}\max_j|K_{\mathrm{ex},jj}|\) and \(\Lambda_{\mathcal M}=\diag\left(\sqrt{\max\{|K_{\mathrm{ex},jj}|,\eta\}}\right)\). Compute
\[
E_{\mathrm{joint}}=\norm{\Lambda_{\mathcal M}^{-1}(K_{\mathrm H}-K_{\mathrm{ex}})\Lambda_{\mathcal M}^{-1}}_{\mathrm F}\norm{\Lambda_{\mathcal M}^{-1}K_{\mathrm{ex}}\Lambda_{\mathcal M}^{-1}}_{\mathrm F}^{-1}.
\]
\item For each monitored ordered pair \((p,q)\), compute
\[
E_{p,q}=\norm{\Lambda_p^{-1}(K_{\mathrm H}^{p,q}-K_{\mathrm{ex}}^{p,q})\Lambda_q^{-1}}_{\mathrm F}\norm{\Lambda_p^{-1}K_{\mathrm{ex}}^{p,q}\Lambda_q^{-1}}_{\mathrm F}^{-1},
\]
where \(\Lambda_p\) and \(\Lambda_q\) contain exact marginal standard deviations for levels \(p\) and \(q\). Set \(E_{\mathrm{pair}}=\max_{p,q\in\mathcal M}E_{p,q}\).
\item Compute the maximum marginal variance error
\[
E_{\mathrm{var}}=\max_j |K_{\mathrm H,jj}-K_{\mathrm{ex},jj}|\max\{|K_{\mathrm{ex},jj}|,\eta\}^{-1}.
\]
\item If \(\mathcal M\) includes positive levels, construct \(K_{\mathrm H}^{\mathrm{full}}\) on \(\mathcal G_{\mathrm{full}}\). Let \(\mathcal M_{\mathrm{int}}=\mathcal M\cap\{1,\ldots,r\}\) be the monitored \(t_0\)-based integral levels, and for each \(p\in\mathcal M_{\mathrm{int}}\) define
\[
\mathcal G_{t_0}^p=\{t\in\mathcal G_{\mathrm{full}}:t_0<t\leq t_0+\delta_{t_0},\, |K_{\mathrm{ex}}^{p,p}(t,t)|>10^{-12}\},
\]
falling back to the threshold \(10^{-15}\max_{t\in\mathcal G_{\mathrm{full}}}|K_{\mathrm{ex}}^{p,p}(t,t)|\) if fewer than \(n_{t_0}\) points are found. This is the anchor-band diagnostic in the implementation. Compute
\[
E_{t_0}=\max_{\substack{p\in\mathcal M_{\mathrm{int}}\\ t\in\mathcal G_{t_0}^p}} |K_{\mathrm H}^{\mathrm{full},p,p}(t,t)-K_{\mathrm{ex}}^{\mathrm{full},p,p}(t,t)|\,|K_{\mathrm{ex}}^{\mathrm{full},p,p}(t,t)|^{-1}.
\]
\item Mark \(c\) acceptable if \(E_{\mathrm{joint}}\leq\tau_{\mathrm{joint}}\), \(E_{\mathrm{pair}}\leq\tau_{\mathrm{pair}}\), \(E_{\mathrm{var}}\leq\tau_{\mathrm{var}}\), and, when applicable, \(E_{t_0}\leq\tau_{t_0}\). Record \(c_i^\star\), the smallest acceptable \(c\) for \(u_i\).
\item Fit only the coefficient \(c_{\mathcal M}\). Let \(\mathcal J=\{i:c_i^\star>c_{\mathrm{floor}}+10^{-12}\}\), use \(c_{\mathcal M}=c_{\mathrm{floor}}\) if \(\mathcal J\) is empty, and otherwise set
\[
c_{\mathcal M}=\max\{c_{\mathrm{floor}},\lceil10\max_{i\in\mathcal J}c_i^\star u_i^{-1}\rceil10^{-1}\}.
\]
Define \(c(u)=\max\{c_{\mathrm{floor}},c_{\mathcal M}u\}\).
\item Validate the envelope on dense or refined grids; accept the calibration only if all required guards remain below tolerance.
\end{enumerate}
\end{tartarebox}

\begin{tartarebox}{Online adaptive use}
\small
\begin{enumerate}[label=\textbf{A\arabic*.},leftmargin=1.8em,itemsep=4pt,topsep=2pt]
\item Normalize the observed input window to half-width \(W\). Integral rows use the corresponding physical left-end reference point \(t_0=-W\).
\item Choose the completed calibration entry for the kernel, representation order, and monitored set. Keep the specified prior on \(\rho\) fixed throughout Phase A.
\item Initialize the working length-scale at \(\rho_{\mathrm{work}}=\rho_{\mathrm{start}}\), with \(\rho_{\mathrm{start}}=0.5W\) in the implementation. For the current \(\rho_{\mathrm{work}}\), set \(u=\rho_{\mathrm{work}}W^{-1}\), \(c=c(u)\), \(L=cW\), \(K=\left\lceil mcu^{-1}\right\rceil\), and \(\Omega_K=\pi K(2L)^{-1}\).
\item Fit the model with the fixed length-scale prior, and let \(\rho_\alpha\) be the posterior \(\alpha_\rho\)-quantile of \(\rho\).
\item Compute \(\ell_{\min}=mcW/K\). Accept the Phase A basis only if \(\rho_\alpha-\delta_\rho\geq \ell_{\min}\) and \(R_{q_\star}(\Omega_K;\rho_\alpha)\leq\varepsilon_{\mathrm{tail}}\).
\item If either Phase A check fails, set the next working length-scale from the posterior length-scale summary, recompute \(u,c,L,K,\Omega_K\), and refit. Stop when both checks pass or when the Phase A iteration cap is reached.
\end{enumerate}

\begin{enumerate}[label=\textbf{B\arabic*.},leftmargin=1.8em,itemsep=4pt,topsep=2pt]
\item Start from the accepted Phase A basis.
\item Increase \(K\) according to the refinement rule. Recompute \(c\), \(L\), and \(\Omega_K\) from the current posterior length-scale summary, and enlarge \(c\) further if the boundary diagnostic requires it.
\item Refit the model and recompute the online diagnostics.
\item Require the Phase A length-scale and spectral-tail checks to continue to hold.
\item Require either the high-frequency fraction of derivative-weighted posterior basis energy, computed from squared basis coefficients weighted by \(\lambda_k^{q_\star}\), to fall below tolerance, or the total derivative-weighted posterior basis energy to be stable relative to the previous fit.
\item Require posterior length-scale stability, leave-one-out expected log predictive density stability, passing sampler diagnostics, stable posterior mean and standard-deviation curves for every monitored level in \(\mathcal M\), and stable posterior means and standard deviations for any unfixed integration constants.
\item First apply the curve-stability checks on a trimmed interior prediction grid, then repeat the endpoint check on the full grid.
\item Accept the online basis only after the configured number of consecutive enlarged-basis refits pass all required checks.
\end{enumerate}
\end{tartarebox}

\input{tables/tartare_calibration_constants.tex}

\subsection{Additional TARTARE simulation results}\label{appendix:additional-simulation-analyses}

\subsubsection{Posterior covariance diagnostics}\label{appendix:additional-simulation-diagnostics}

\Cref{tab:tartare-d2-vs-full-diagnostics} compares posterior covariance errors for T-D2, with \(\mathcal M=\{-2\}\), and T-all, with \(\mathcal M=\{-2,-1,0,1,2\}\), using the same replicated datasets as \cref{tab:tartare-simulation-summary}. The diagnostics are whitened relative Frobenius errors against the exact-covariance posterior: \(E_{\mathrm{joint}}\) for the stacked covariance, \(E_{\mathrm{pair,max}}\) for the worst pairwise block, and \(E_{\mathrm{var,max}}\) for marginal variances. Positive changes favor full-ensemble calibration.

\begin{table}[htbp]
\centering
\caption{\label{tab:tartare-d2-vs-full-diagnostics}Posterior covariance diagnostics comparing TARTARE calibration targeted at the second derivative (D2) with calibration targeted at the full second-order ensemble (Full). Completed columns give the number of accepted simulation replicates contributing to each diagnostic.}
\resizebox{\textwidth}{!}{\input{tables/supplement_tartare_d2_vs_full_diagnostics.tex}}
\end{table}

Across the eight configurations, T-D2 and T-all give nearly identical posterior covariance errors. Changes in \(E_{\mathrm{joint}}\) and \(E_{\mathrm{pair,max}}\) are about one percent or less. This reflects the structure of the second-order ensemble: \(f_{-2}\) is the hardest component for the finite basis because differentiation amplifies high-frequency error, while the anchor and integral levels are easier once \(f_{-2}\) is resolved.

The pairwise diagnostic tests whether full-ensemble calibration improves cross-level dependence. In these simulations, it does not materially reduce \(E_{\mathrm{pair,max}}\), and inspection of individual cross-block errors gives the same conclusion. This is a global Frobenius diagnostic, so it does not rule out differences for localized covariance features or downstream functionals.

TARTARE is a finite-grid acceptance rule, not an optimizer of the posterior diagnostics in this table. Adding all ensemble levels can increase the calibration burden without lowering posterior covariance error when \(f_{-2}\) is already the limiting target. In these scenarios, \(\mathcal M=\{-2\}\) captures nearly all of the covariance-approximation gain, while full-ensemble monitoring remains a conservative choice when several levels are primary targets.

\subsubsection{Boundary-constant sensitivity}\label{appendix:boundary-constant-sensitivity}

This simulation isolates the role of positive-index integration constants. Anchor observations update the canonical integral component induced by \(f_0\), but they do not identify independent boundary constants. Integral-level uncertainty and bias can therefore be dominated by the boundary component unless the constants are fixed correctly or informed by boundary measurements.

The simulation reused the second-order design grid from the main study after removing duplicate calibration-target rows. This gave eight covariance-design settings: SE and Mat\'ern \(7/2\), \(\rho\in\{0.35,0.65\}\), \(n\in\{25,50\}\), \(n_{\mathrm{pred}}=101\), \(\sigma_y=0.10\), and truth \((\kappa_1,\kappa_2)=(0.60,-0.40)\). For each setting, \(1000\) anchor datasets were generated on \([-1,1]\). The analysis did not refit Stan models. It combined exact canonical integral covariances with analytic Gaussian boundary-constant conditions.

The \(22\) boundary settings were evaluated for all eight covariance-design settings and \(1000\) replicated anchor datasets per design. They consist of two fixed baselines, \(\kappa=(0.60,-0.40)\) and \(\kappa=(0,0)\), four independent-prior settings \(\kappa_j\sim N(0,\tau^2)\), and sixteen direct-observation settings \(z_j=\kappa_j+\eta_j\), \(\eta_j\sim N(0,\omega^2)\). Here \(\tau\in\{0.10,0.25,0.50,1.00\}\) and \(\omega\in\{0.02,0.05,0.10,0.25\}\). Constant posteriors were computed by Gaussian conditioning, then propagated through the polynomial integral terms.

\begin{table}[htbp]
\centering
\caption{\label{tab:boundary-constant-kappa-summary}Boundary-constant posterior summaries in the sensitivity simulation. Entries aggregate over the boundary settings within each condition type, eight design settings, and \(1000\) seeds per design. Coverage is the proportion of nominal 95\% intervals containing the generating value. The SD/prior SD column is reported for the prior-based conditions and is averaged over the prior-standard-deviation grid.}
\resizebox{\textwidth}{!}{\input{tables/boundary_constant_kappa_summary.tex}}
\end{table}

\Cref{tab:boundary-constant-kappa-summary} shows the identifiability pattern. Known constants are degenerate at the truth, fixed wrong constants are degenerate at the wrong value, and independent priors are unchanged by anchor data. Direct boundary observations reduce the average posterior SD from \(0.463\) to \(0.086\), but average coverage remains below nominal because the grid includes tight priors centered away from the truth.

\begin{table}[htbp]
\centering
\caption{\label{tab:boundary-constant-variance-decomposition}Mean variance decomposition for the integral levels in the boundary-constant sensitivity simulation. Variance columns are multiplied by \(10^3\). The canonical variance is induced by the anchor-process posterior covariance. The constant variance is the polynomial contribution from posterior uncertainty in \((\kappa_1,\kappa_2)\).}
\resizebox{\textwidth}{!}{\input{tables/boundary_constant_variance_decomposition.tex}}
\end{table}

\Cref{tab:boundary-constant-variance-decomposition} shows that boundary uncertainty dominates positive-index variance. Under independent priors, constants contribute \(98.3\%\) of mean \(I_1\) variance and \(99.5\%\) of mean \(I_2\) variance. Direct boundary observations reduce the contribution sharply, but it still accounts for \(78.1\%\) and \(90.2\%\), respectively.

\begin{table}[htbp]
\centering
\caption{\label{tab:boundary-constant-integral-error}Integrated error and pointwise interval summaries for the positive-index levels. IRMSE is the integrated root mean squared error over the prediction grid. Coverage and interval width are averaged over prediction locations, replicated datasets, and design settings.}
\resizebox{\textwidth}{!}{\input{tables/boundary_constant_integral_error_summary.tex}}
\end{table}

\Cref{tab:boundary-constant-integral-error} separates uncertainty calibration from bias. When the constants are known, the remaining error is small, with IRMSE \(0.0209\) for \(I_1\), \(0.0161\) for \(I_2\), and pointwise coverage about \(0.94\). Fixing the constants incorrectly leaves the same narrow canonical intervals but shifts the integral levels, increasing IRMSE to \(0.6003\) for \(I_1\) and \(0.4003\) for \(I_2\) and producing essentially no coverage. Independent priors without boundary data have the same posterior mean error as the fixed-wrong condition, because the prior mean is zero, but they produce much wider intervals. Those wide intervals partially recover coverage, especially for \(I_2\), at the cost of interval widths that are roughly twenty-one to forty-four times larger than the known-constant widths. Direct boundary observations give the best compromise among the non-oracle conditions: they reduce IRMSE to about \(0.135\) for both levels and keep intervals much narrower than the independent-prior intervals, although average coverage remains below nominal for \(I_1\) because the grid includes prior settings that are too concentrated around zero.

\begin{table}[htbp]
\centering
\caption{\label{tab:boundary-constant-location-summary}Location-specific error summaries for the integral levels. The anchor is the boundary point where constants enter directly, the midpoint is the center of the prediction interval, and the far endpoint is the opposite end of the interval. Coverage is for nominal 95\% intervals.}
\resizebox{\textwidth}{!}{\input{tables/boundary_constant_location_summary.tex}}
\end{table}

\Cref{tab:boundary-constant-location-summary} shows how the boundary error propagates through the polynomial part of the integral representation. For \(I_1\), an error in \(\kappa_1\) shifts the level across the domain, so the fixed-wrong RMSE remains approximately \(0.60\) at the anchor, midpoint, and far endpoint. For \(I_2\), the location effect is polynomial: the fixed-wrong RMSE is \(0.40\) at the anchor, \(0.201\) at the midpoint, and \(0.801\) at the far endpoint. The known-constant condition has near-nominal coverage away from the anchor, where the degenerate boundary value makes coverage numerically brittle, and the direct-observation condition again falls between the oracle and prior-only cases. The far endpoint of \(I_2\) is the most sensitive location because uncertainty in \(\kappa_1\) and \(\kappa_2\) accumulates through the second integral.

Boundary constants are statistical components of the positive-index ensemble, not HSGP details. Anchor-only data cannot identify independent constants. Fixing them incorrectly produces biased, overconfident integral summaries. Independent priors represent boundary uncertainty but do not remove boundary bias. Direct boundary measurements reduce uncertainty and error, although coverage can still suffer when the prior scale is too tight. Applied analyses should state the boundary specification explicitly and treat prior scales on integration constants as sensitivity parameters.

\subsection{Motorcycle covariance priors and diagnostics}\label{appendix:additional-application-information}

\noindent The motorcycle model uses Mat\'ern \(7/2\) anchor kernels for both the mean ensemble and the log standard deviation ensemble. It follows the finite-basis scale parameterization described in the single-curve Stan implementation, applied separately to the mean and log standard deviation ensembles. The reported covariance parameters are $(\alpha_\mu,\rho_\mu,\sigma_\mu)$ and $(\alpha_h,\rho_h,\sigma_h)$ where \(\alpha\) is the reconstructed process amplitude, \(\rho\) is the length-scale, and \(\sigma\) is the residual scale. These parameters were estimated jointly with the latent basis coefficients. Let \(s_y\) denote the empirical scale used for the corresponding response, and let \(\bar V_\rho\) be the observed-design average unit-amplitude HSGP variance for the corresponding ensemble. The Stan scale parameters are \(\sigma_{\mathrm{tot}}\) and \(\eta\), with
\[
\sigma_{\mathrm{tot}}^2=\alpha^2\bar V_\rho+\sigma^2,
\qquad
\eta=\frac{\sigma^2}{\sigma_{\mathrm{tot}}^2}.
\]
Their priors are
\[
\sigma_{\mathrm{tot}}\sim \mathrm{half}\text{-}t_4(0,s_y),
\qquad
\eta\sim \mathrm{Beta}(2,2),
\]
and the length-scale has prior
\[
\rho\sim \mathrm{LogNormal}(m_\rho,s_\rho),
\]
where \(m_\rho\) is the log-location parameter and \(s_\rho\) is the log-scale parameter. The table below reports the reconstructed parameters \(\alpha=\sigma_{\mathrm{tot}}\sqrt{(1-\eta)/\bar V_\rho}\) and \(\sigma=\sigma_{\mathrm{tot}}\sqrt{\eta}\). In the implementation, \(L\) is the calibrated HSGP computational half-width. Let \(\Delta_{\mathrm{med}}\) be the median gap between sorted centered observation times, using \(0.05L\) as the fallback spacing, and set
\[
q_{\rho,\mathrm{hi}}=2L,
\qquad
q_{\rho,\mathrm{lo}}=\min\{\max(\Delta_{\mathrm{med}},0.05L),L\}.
\]
The lognormal parameters are
\[
m_\rho=\frac12(\log q_{\rho,\mathrm{lo}}+\log q_{\rho,\mathrm{hi}}),
\qquad
s_\rho=\frac{\log q_{\rho,\mathrm{hi}}-\log q_{\rho,\mathrm{lo}}}{2\Phi^{-1}(0.95)}.
\]
Sampling used \(4\) chains with \(2000\) warmup and \(2000\) post-warmup iterations per chain, \(\texttt{adapt\_delta}=0.95\), and \(\texttt{max\_treedepth}=12\).

\Cref{tab:motorcycle-covariance-summary} reports posterior summaries and Monte Carlo efficiency diagnostics for the mean and log standard deviation ensembles. \Cref{fig:motorcycle-covariance-trace} shows the corresponding chain traces.
\input{tables/motorcycle_covariance_summary.tex}

\begin{figure}[htbp]
\centering
\includegraphics[width=\textwidth]{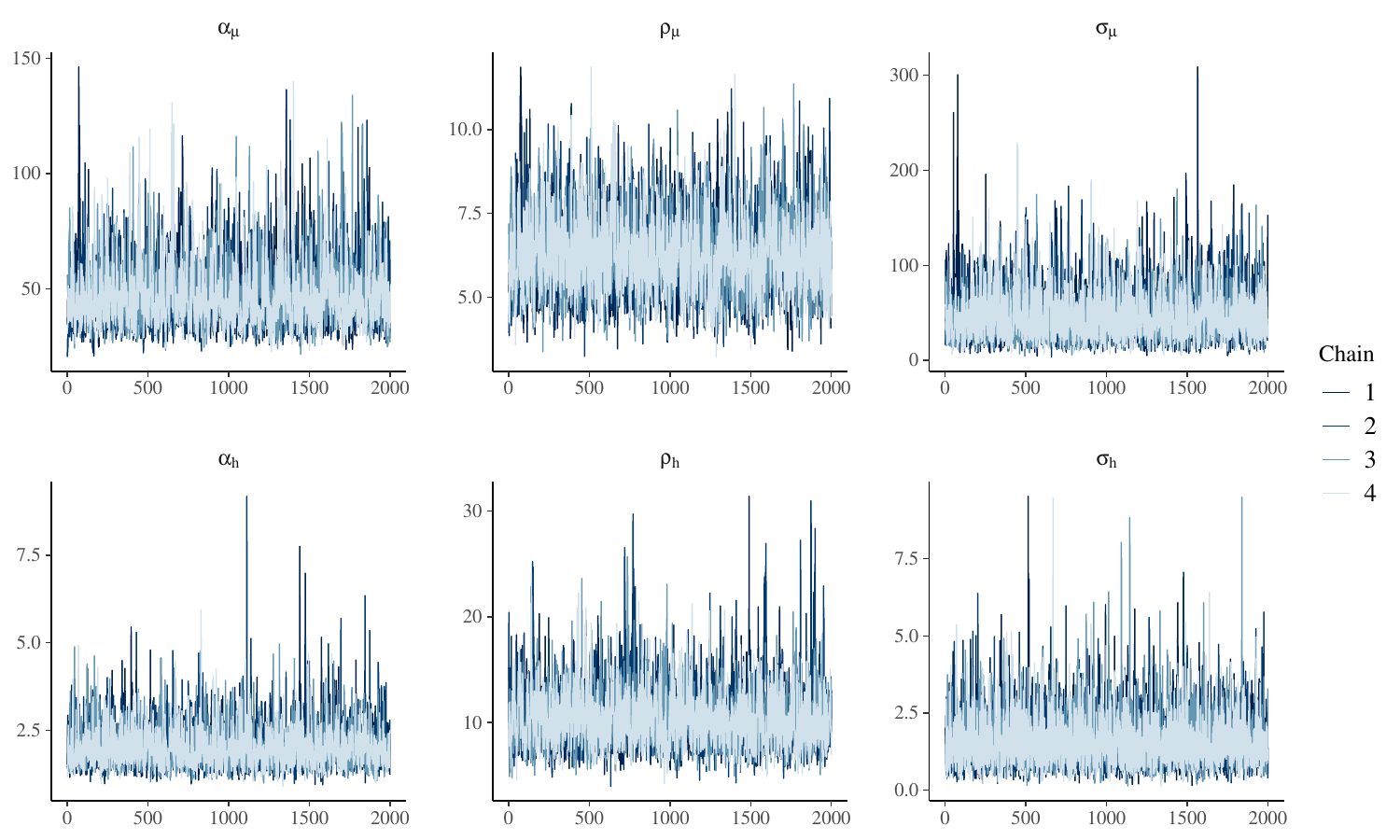}
\caption{Posterior trace plots for the motorcycle covariance hyperparameters. Panels show amplitude, length-scale, and residual-scale parameters for the mean ensemble \((\mu)\) and log standard deviation ensemble \((h)\).}
\label{fig:motorcycle-covariance-trace}
\end{figure}

\end{document}

%% file: tables/tartare_simulation_summary.tex
\begin{tabular}{
rrrrrrr|rrrrrr|rrrrr
}
\multicolumn{7}{c|}{} & \multicolumn{6}{c|}{IRMSE} & \multicolumn{5}{c}{Coverage (\%)} \\
$\rho$ & Kernel & n & Method & Failed (
\hline
0.65 & SE & 25 & Exact & -- & -- & -- & 1.00 & 1.14 & 0.21 & 0.04 & 0.02 & 0.02 & 95.26 & 95.05 & 95.23 & 94.48 & 94.05 \\
0.65 & SE & 25 & HSGP & 5.60 & 15 & 1.76 & 1.58 & 4.35 & 0.44 & 0.05 & 0.02 & 0.02 & 91.46 & 94.25 & 96.98 & 96.53 & 96.04 \\
0.65 & SE & 25 & T-f & 3.50 & 22 & 4.42 & 1.10 & 1.41 & 0.24 & 0.05 & 0.02 & 0.02 & 93.20 & 94.57 & 96.22 & 96.14 & 95.66 \\
0.65 & SE & 25 & T-D2 & 6.70 & 52 & 5.32 & 1.09 & 1.38 & 0.24 & 0.04 & 0.02 & 0.02 & 94.43 & 95.44 & 96.43 & 96.24 & 95.83 \\
0.65 & SE & 25 & T-all & 7.40 & 54 & 5.33 & 1.08 & 1.38 & 0.24 & 0.04 & 0.02 & 0.02 & 94.42 & 95.46 & 96.45 & 96.34 & 95.89 \\
0.65 & SE & 25 & Spline & -- & -- & -- & 1.81 & 7.44 & 0.49 & 0.05 & 0.02 & 0.02 & -- & -- & -- & -- & -- \\
\hline
0.65 & SE & 50 & Exact & -- & -- & -- & 1.00 & 1.00 & 0.17 & 0.03 & 0.02 & 0.01 & 94.88 & 94.85 & 94.94 & 93.62 & 93.48 \\
0.65 & SE & 50 & HSGP & 5.20 & 15 & 1.62 & 1.68 & 4.58 & 0.40 & 0.04 & 0.02 & 0.01 & 89.19 & 92.14 & 95.22 & 94.18 & 93.81 \\
0.65 & SE & 50 & T-f & 0.70 & 21 & 4.27 & 1.08 & 1.23 & 0.20 & 0.03 & 0.02 & 0.01 & 92.96 & 94.20 & 95.41 & 94.16 & 94.02 \\
0.65 & SE & 50 & T-D2 & 2.50 & 47 & 5.23 & 1.08 & 1.23 & 0.20 & 0.03 & 0.02 & 0.01 & 93.34 & 94.40 & 95.44 & 94.34 & 94.13 \\
0.65 & SE & 50 & T-all & 2.90 & 48 & 5.23 & 1.08 & 1.23 & 0.20 & 0.03 & 0.02 & 0.01 & 93.42 & 94.42 & 95.44 & 94.36 & 94.14 \\
0.65 & SE & 50 & Spline & -- & -- & -- & 1.93 & 8.98 & 0.41 & 0.04 & 0.02 & 0.01 & -- & -- & -- & -- & -- \\
\hline
0.65 & Mat7/2 & 25 & Exact & -- & -- & -- & 1.00 & 3.73 & 0.38 & 0.05 & 0.02 & 0.02 & 94.74 & 94.73 & 95.00 & 94.55 & 93.95 \\
0.65 & Mat7/2 & 25 & T-f & 2.00 & 36 & 4.84 & 1.08 & 4.09 & 0.42 & 0.06 & 0.03 & 0.02 & 85.56 & 91.79 & 95.87 & 96.41 & 96.21 \\
0.65 & Mat7/2 & 25 & T-D2 & 6.20 & 178 & 4.00 & 1.06 & 4.06 & 0.41 & 0.06 & 0.02 & 0.02 & 94.53 & 95.50 & 96.81 & 96.49 & 96.21 \\
0.65 & Mat7/2 & 25 & T-all & 7.00 & 256 & 5.64 & 1.06 & 4.06 & 0.41 & 0.06 & 0.02 & 0.02 & 93.91 & 95.07 & 96.70 & 96.46 & 96.30 \\
0.65 & Mat7/2 & 25 & Spline & -- & -- & -- & 1.40 & 10.21 & 0.63 & 0.06 & 0.02 & 0.02 & -- & -- & -- & -- & -- \\
\hline
0.65 & Mat7/2 & 50 & Exact & -- & -- & -- & 1.00 & 3.34 & 0.31 & 0.04 & 0.02 & 0.01 & 95.10 & 95.23 & 95.10 & 93.64 & 93.54 \\
0.65 & Mat7/2 & 50 & T-f & 2.10 & 39 & 4.44 & 1.04 & 3.64 & 0.34 & 0.04 & 0.02 & 0.01 & 87.33 & 92.58 & 95.24 & 94.36 & 94.31 \\
0.65 & Mat7/2 & 50 & T-D2 & 1.30 & 168 & 3.82 & 1.04 & 3.65 & 0.33 & 0.04 & 0.02 & 0.01 & 95.27 & 95.41 & 95.85 & 94.24 & 94.38 \\
0.65 & Mat7/2 & 50 & T-all & 2.50 & 236 & 5.38 & 1.04 & 3.64 & 0.33 & 0.04 & 0.02 & 0.01 & 94.97 & 95.19 & 95.79 & 94.40 & 94.52 \\
0.65 & Mat7/2 & 50 & Spline & -- & -- & -- & 1.46 & 11.37 & 0.52 & 0.04 & 0.02 & 0.01 & -- & -- & -- & -- & -- \\
\hline
0.35 & SE & 25 & Exact & -- & -- & -- & 1.00 & 3.77 & 0.41 & 0.05 & 0.02 & 0.02 & 95.08 & 95.08 & 95.00 & 94.50 & 94.06 \\
0.35 & SE & 25 & HSGP & 2.00 & 23 & 1.22 & 1.30 & 7.48 & 0.62 & 0.06 & 0.03 & 0.02 & 90.72 & 94.25 & 97.50 & 96.91 & 96.57 \\
0.35 & SE & 25 & T-f & 0.60 & 26 & 2.38 & 1.08 & 4.37 & 0.46 & 0.06 & 0.02 & 0.02 & 94.63 & 95.56 & 96.85 & 96.58 & 96.25 \\
0.35 & SE & 25 & T-D2 & 1.10 & 42 & 2.87 & 1.08 & 4.36 & 0.46 & 0.06 & 0.02 & 0.02 & 95.01 & 95.75 & 96.90 & 96.75 & 96.33 \\
0.35 & SE & 25 & T-all & 1.20 & 44 & 2.87 & 1.08 & 4.36 & 0.46 & 0.06 & 0.02 & 0.02 & 95.02 & 95.74 & 96.90 & 96.75 & 96.29 \\
0.35 & SE & 25 & Spline & -- & -- & -- & 1.49 & 12.71 & 0.74 & 0.06 & 0.02 & 0.02 & -- & -- & -- & -- & -- \\
\hline
0.35 & SE & 50 & Exact & -- & -- & -- & 1.00 & 3.24 & 0.33 & 0.04 & 0.02 & 0.01 & 94.66 & 94.84 & 95.01 & 93.70 & 93.60 \\
0.35 & SE & 50 & HSGP & 1.20 & 23 & 1.21 & 1.33 & 7.33 & 0.54 & 0.04 & 0.02 & 0.01 & 88.92 & 92.56 & 95.54 & 94.54 & 94.34 \\
0.35 & SE & 50 & T-f & 0.60 & 27 & 2.32 & 1.05 & 3.68 & 0.36 & 0.04 & 0.02 & 0.01 & 94.11 & 94.94 & 95.67 & 94.40 & 94.38 \\
0.35 & SE & 50 & T-D2 & 0.40 & 39 & 2.81 & 1.05 & 3.65 & 0.36 & 0.04 & 0.02 & 0.01 & 94.56 & 95.11 & 95.75 & 94.40 & 94.56 \\
0.35 & SE & 50 & T-all & 0.40 & 40 & 2.81 & 1.05 & 3.66 & 0.36 & 0.04 & 0.02 & 0.01 & 94.56 & 95.10 & 95.77 & 94.40 & 94.54 \\
0.35 & SE & 50 & Spline & -- & -- & -- & 1.54 & 13.19 & 0.59 & 0.05 & 0.02 & 0.01 & -- & -- & -- & -- & -- \\
\hline
0.35 & Mat7/2 & 25 & Exact & -- & -- & -- & 1.00 & 13.33 & 0.77 & 0.07 & 0.02 & 0.02 & 94.67 & 94.74 & 95.04 & 94.58 & 94.16 \\
0.35 & Mat7/2 & 25 & T-f & 2.40 & 40 & 2.64 & 1.06 & 14.27 & 0.84 & 0.07 & 0.03 & 0.02 & 87.44 & 93.33 & 97.14 & 97.36 & 97.25 \\
0.35 & Mat7/2 & 25 & T-D2 & 4.50 & 186 & 2.30 & 1.05 & 14.13 & 0.84 & 0.07 & 0.03 & 0.02 & 94.00 & 95.37 & 97.44 & 97.44 & 97.36 \\
0.35 & Mat7/2 & 25 & T-all & 4.20 & 272 & 3.22 & 1.06 & 14.13 & 0.84 & 0.07 & 0.03 & 0.02 & 93.66 & 95.21 & 97.42 & 97.44 & 97.40 \\
0.35 & Mat7/2 & 25 & Spline & -- & -- & -- & 1.24 & 24.19 & 1.11 & 0.08 & 0.02 & 0.02 & -- & -- & -- & -- & -- \\
\hline
0.35 & Mat7/2 & 50 & Exact & -- & -- & -- & 1.00 & 12.12 & 0.64 & 0.05 & 0.02 & 0.01 & 94.81 & 94.89 & 95.01 & 93.83 & 93.74 \\
0.35 & Mat7/2 & 50 & T-f & 0.50 & 45 & 2.49 & 1.02 & 12.69 & 0.67 & 0.05 & 0.02 & 0.01 & 89.43 & 93.75 & 95.64 & 94.69 & 94.74 \\
0.35 & Mat7/2 & 50 & T-D2 & 0.50 & 180 & 2.17 & 1.02 & 12.64 & 0.67 & 0.05 & 0.02 & 0.01 & 95.02 & 95.38 & 95.94 & 94.79 & 94.76 \\
0.35 & Mat7/2 & 50 & T-all & 0.30 & 253 & 3.04 & 1.02 & 12.63 & 0.67 & 0.05 & 0.02 & 0.01 & 94.79 & 95.27 & 95.85 & 94.75 & 94.67 \\
0.35 & Mat7/2 & 50 & Spline & -- & -- & -- & 1.22 & 22.12 & 0.85 & 0.05 & 0.02 & 0.01 & -- & -- & -- & -- & -- \\
\end{tabular}

%% file: main-proofs.tex

\begin{proof}[Proof of \cref{prop:gaussian-ensemble}]
Write
\begin{align*}
A_p(t) := (\mathcal A_{t_0}^{[p]}f_0)(t),
\qquad
B_p(t) := \mathbf{1}_{\{p>0\}}\sum_{j=1}^{p}\frac{(t-t_0)^{p-j}}{(p-j)!}\kappa_j,
\end{align*}
so that \(f_p(t)=A_p(t)+B_p(t)\). Standard mean-square differentiation theory under \cref{ass:ms-regularity} gives the derivative processes \(D_t^k f_0\), \(k=1,\ldots,r\), with
\begin{align*}
\E[D_t^k f_0(t)] = D_t^k\mu(t),
\qquad
\Cov(D_s^a f_0(s),D_t^b f_0(t))=(D_1^aD_2^b C_\theta)(s,t),
\end{align*}
for the relevant orders \(a,b,k\le r\) \citep[Theorem~2.5.1]{cramer1967stationary}. In particular, the components \(f_{-r},\ldots,f_{-1}\) are square-integrable, and \(f_{-r},\ldots,f_0\) satisfy the chain relations for \(p=-r+1,\ldots,0\).

Since \(\mu\) and \(C_\theta\) are continuous on compact sets,
\begin{align*}
M:=\sup_{u\in[t_0,t_1]}\E[f_0(u)^2]
=\sup_{u\in[t_0,t_1]}\left(\mu(u)^2+C_\theta(u,u)\right)<\infty.
\end{align*}
This gives \(A_0=f_0\in L^2\). For \(p>0\), Cauchy--Schwarz gives, for every \(t\in[t_0,t_1]\),
\begin{align*}
\E[A_p(t)^2]
&\le
\frac{1}{((p-1)!)^2}
\left(\int_{t_0}^t (t-u)^{2p-2}\,du\right)
\left(\int_{t_0}^t \E[f_0(u)^2]du\right)<\infty.
\end{align*}
The term \(B_p(t)\) is a finite linear combination of the Gaussian variables \(\kappa_1,\ldots,\kappa_p\), and is therefore square-integrable. Hence every \(f_p(t)\) is square-integrable.

It remains to verify the positive-index part of the derivative chain. Let \(G_0=f_0\) and \(G_p=\mathcal I_{t_0}^p f_0\) for \(p\ge1\). The continuity of \(\mu\) and \(C_\theta\) implies that \(f_0\) is mean-square continuous. Moreover, the repeated integral identity
\begin{align*}
G_p(t)=\int_{t_0}^t G_{p-1}(u)\,du
\end{align*}
holds in \(L^2\). If \(G_{p-1}\) is mean-square continuous, then, for \(t\in(t_0,t_1)\) and \(h\) such that \(t+h\in[t_0,t_1]\),
\begin{align*}
\frac{G_p(t+h)-G_p(t)}{h}-G_{p-1}(t)
=\frac{1}{h}\int_t^{t+h}\left(G_{p-1}(u)-G_{p-1}(t)\right)\,du,
\end{align*}
and Jensen's inequality yields
\begin{align*}
\E\!\left[\left|\frac{G_p(t+h)-G_p(t)}{h}-G_{p-1}(t)\right|^2\right]
&\le
\frac{1}{|h|}\int_{\min(t,t+h)}^{\max(t,t+h)}
\E\!\left[|G_{p-1}(u)-G_{p-1}(t)|^2\right]du
\to0.
\end{align*}
Thus \(G_p\) is mean-square differentiable with derivative \(G_{p-1}\), and hence is mean-square continuous. Induction gives this for all \(p=1,\ldots,r\). Since \(B_p\) has ordinary derivative \(B_{p-1}\) for \(p\ge1\), with \(B_0\equiv0\), the full processes satisfy the mean-square derivative relation \(D_t f_p=f_{p-1}\) for \(p=1,\ldots,r\). Finally, \(G_p(t_0)=0\) and $f_p(t_0)=B_p(t_0)=\kappa_p$ for $p=1,\ldots,r$ almost surely, so the boundary identities hold.

We next prove joint Gaussianity. Fix \((p_1,\tau_1),\ldots,(p_n,\tau_n)\) with \(p_i\in\{-r,\ldots,r\}\) and \(\tau_i\in[t_0,t_1]\), and set \(\mathbf A=(A_{p_1}(\tau_1),\ldots,A_{p_n}(\tau_n))^\top\). Each coordinate of \(\mathbf A\) is an \(L^2\)-limit of finite linear functionals of finite evaluation vectors of \(f_0\): finite differences for negative indices, the process value itself for index zero, and Riemann sums for positive indices. Finite evaluation vectors of \(f_0\) are Gaussian, finite linear transformations preserve Gaussianity, and the \(L^2\)-limit is Gaussian by the Cram\'er--Wold argument. The same approximations show that \(\mathbf A\) may be chosen measurable with respect to the completed \(\sigma\)-field generated by \(f_0\). The corresponding vector \(\mathbf B=(B_{p_1}(\tau_1),\ldots,B_{p_n}(\tau_n))^\top\) is a linear transformation of \(\boldsymbol\kappa\), hence Gaussian, and it is independent of \(\mathbf A\) because \(\boldsymbol\kappa\) is independent of \(f_0\). Therefore \(\mathbf A+\mathbf B\) is Gaussian. Since the finite collection was arbitrary, \(\mathbf f\) is a Gaussian process.

The mean formula follows by applying expectation to the same \(L^2\) limits for negative indices and by Fubini's theorem for positive indices: $\E[A_p(t)] = (\mathcal A_{t_0}^{[p]}\mu)(t)$ for $p=-r,\ldots,r$. Adding \(\E[B_p(t)]\) gives
\begin{align*}
\E[f_p(t)]
&=
(\mathcal A_{t_0}^{[p]}\mu)(t)
+
\mathbf{1}_{\{p>0\}}
\sum_{j=1}^{p}\frac{(t-t_0)^{p-j}}{(p-j)!}(\boldsymbol\mu_\kappa)_j.
\end{align*}
If \(\mu\equiv0\) and \(\boldsymbol\mu_\kappa=\mathbf0\), this mean is zero for every component and time point; together with square-integrability, the derivative chain, and the boundary identities proved above, this gives a centered \(r\)th-order differential ensemble in the sense of \cref{def:ensemble}.

It remains to compute the covariance. Independence of \(f_0\) and \(\boldsymbol\kappa\) gives
\begin{align*}
\Cov(f_p(s),f_q(t))=\Cov(A_p(s),A_q(t))+\Cov(B_p(s),B_q(t)).
\end{align*}
The polynomial part is
\begin{align*}
\Cov(B_p(s),B_q(t))
&=
\mathbf{1}_{\{p>0\}}\mathbf{1}_{\{q>0\}}
\sum_{j=1}^{p}\sum_{m=1}^{q}
\frac{(s-t_0)^{p-j}}{(p-j)!}
\frac{(t-t_0)^{q-m}}{(q-m)!}
(\boldsymbol\Sigma_\kappa)_{jm}.
\end{align*}
For the anchor part, the relevant kernels are continuous on compact sets, so the covariance--Fubini interchanges below are justified by the \(L^2\) bounds above. The standard derivative covariance formula gives the claim when \(p,q\le0\). When \(p>0\) and \(q\le0\), continuity of \(D_2^{-q}C_\theta\) and the \(L^2\)-integral definition give
\begin{align*}
\Cov(A_p(s),A_q(t))
&=\frac{1}{(p-1)!}\int_{t_0}^s (s-u)^{p-1}
\Cov(f_0(u),D_t^{-q}f_0(t))\,du\\
&=\frac{1}{(p-1)!}\int_{t_0}^s (s-u)^{p-1}(D_2^{-q}C_\theta)(u,t)\,du\\
&=(\mathcal I_{1,t_0}^pD_2^{-q}C_\theta)(s,t).
\end{align*}
The case \(p\le0\) and \(q>0\) similarly gives
\begin{align*}
\Cov(A_p(s),A_q(t))
&=\frac{1}{(q-1)!}\int_{t_0}^t (t-v)^{q-1}(D_1^{-p}C_\theta)(s,v)\,dv\\
&=(D_1^{-p}\mathcal I_{2,t_0}^qC_\theta)(s,t).
\end{align*}
Finally, when \(p,q>0\), continuity of \(C_\theta\) and Fubini's theorem give
\begin{align*}
\Cov(A_p(s),A_q(t))
&=\frac{1}{(p-1)!(q-1)!}
\int_{t_0}^s\int_{t_0}^t
(s-u)^{p-1}(t-v)^{q-1}C_\theta(u,v)\,dv\,du\\
&=(\mathcal I_{1,t_0}^p\mathcal I_{2,t_0}^qC_\theta)(s,t).
\end{align*}
All four cases are exactly
\begin{align*}
\Cov(A_p(s),A_q(t))=(\mathcal A_{1,t_0}^{[p]}\mathcal A_{2,t_0}^{[q]}C_\theta)(s,t).
\end{align*}
Combining this identity with the polynomial covariance gives the stated covariance function.
\end{proof}

\begin{proof}[Proof of \cref{prop:samplepaths}]
Let \(\Omega_C\) be a probability-one event on which \(f_0^\ast(\cdot,\omega)\in C^r([t_0,t_1])\). Define pathwise candidate versions, with arbitrary values assigned off \(\Omega_C\) for negative indices, by
\begin{align*}
\tilde f_p(t)=
\begin{cases}
(\mathcal A_{t_0}^{[p]}f_0^\ast)(t), & p\le 0,\\
(\mathcal A_{t_0}^{[p]}f_0^\ast)(t)+\displaystyle\sum_{j=1}^{p}\frac{(t-t_0)^{p-j}}{(p-j)!}\kappa_j, & p>0.
\end{cases}
\end{align*}
First identify the negative-index versions. Fix \(k=1,\ldots,r\) and \(t\in[t_0,t_1]\). Under \cref{ass:ms-regularity}, \(f_{-k}(t)\) is the \(L^2\) limit of \(k\)th difference quotients of \(f_0\), using one-sided quotients at endpoints. Replacing \(f_0\) by its version \(f_0^\ast\) does not change these quotients almost surely for each fixed step. On \(\Omega_C\), the same quotients converge to \(D_t^k f_0^\ast(t)\). Uniqueness of limits in probability gives \(D_t^k f_0^\ast(t)=f_{-k}(t)\) almost surely. Hence \(D_t^k f_0^\ast\) is a version of \(f_{-k}\). For \(p=0\), \(\tilde f_0=f_0^\ast\) is a version of \(f_0\). For \(p=-k<0\), the preceding argument shows that \(D_t^k f_0^\ast=(\mathcal A_{t_0}^{[-k]}f_0^\ast)\) is a version of \(f_{-k}\).

It remains to identify the positive-index candidates with the \(L^2\)-defined processes in \cref{prop:gaussian-ensemble}. Fix \(p>0\) and \(t\in[t_0,t_1]\), and set
\begin{align*}
K_{p,t}(u):=\frac{(t-u)^{p-1}}{(p-1)!}\mathbf 1_{[t_0,t]}(u).
\end{align*}
Let
\begin{align*}
J_t:=\int_{t_0}^{t_1}K_{p,t}(u)f_0(u)\,du,
\qquad
J_t^\ast:=\int_{t_0}^{t_1}K_{p,t}(u)f_0^\ast(u)\,du.
\end{align*}
The joint-measurability and integrability assumptions make these samplewise integrals well defined. Since \(K_{p,t}\) is bounded and \(\sup_{u\in[t_0,t_1]}\E[f_0(u)^2]<\infty\) under \cref{ass:ms-regularity}, \(J_t\in L^2(\Omega)\). For every \(Z\in L^2(\Omega)\), Fubini's theorem and Cauchy--Schwarz yield
\begin{align*}
\E[ZJ_t]
&=\int_{t_0}^{t_1}K_{p,t}(u)\E[Zf_0(u)]\,du.
\end{align*}
Thus the samplewise integral \(J_t\) represents the weak \(L^2\)-integral \(\bigl(\mathcal I_{t_0}^{p}f_0\bigr)(t)\) used in \cref{prop:gaussian-ensemble}.

Because \(f_0^\ast\) is a version of \(f_0\), Tonelli's theorem gives
\begin{align*}
\E|J_t^\ast-J_t|
&\le \E\int_{t_0}^{t_1}|K_{p,t}(u)|\,|f_0^\ast(u)-f_0(u)|\,du\\
&=\int_{t_0}^{t_1}|K_{p,t}(u)|\,\E|f_0^\ast(u)-f_0(u)|\,du
=0.
\end{align*}
Hence \(J_t^\ast=J_t\) almost surely. Adding the same polynomial in \(\kappa_1,\ldots,\kappa_p\) shows that, for this fixed \(t\), \(\tilde f_p(t)\) equals the \(f_p(t)\) constructed in \cref{prop:gaussian-ensemble} almost surely. Therefore each \(\tilde f_p\) is a version of \(f_p\).

Replace \(f_p\) by \(\tilde f_p\) for all \(p=-r,\ldots,r\) and drop the tildes. On the full-probability event \(\Omega_\ast:=\Omega_C\), the displayed pathwise representation in the proposition holds for every \(p\) and every \(t\in[t_0,t_1]\) by construction. It remains only to check the pathwise derivative chain on this event. If \(p\le0\) and \(p\ge-r+1\), ordinary differentiation of the \(C^r\) path \(f_0^\ast(\cdot,\omega)\) gives
\[
\frac{d}{dt}(\mathcal A_{t_0}^{[p]}f_0^\ast)(t,\omega)=(\mathcal A_{t_0}^{[p-1]}f_0^\ast)(t,\omega).
\]
If \(p>0\), the fundamental theorem of calculus gives
\[
\frac{d}{dt}(\mathcal A_{t_0}^{[p]}f_0^\ast)(t,\omega)=(\mathcal A_{t_0}^{[p-1]}f_0^\ast)(t,\omega),
\]
where \(p=1\) gives \(f_0^\ast(t,\omega)\). Differentiating the polynomial boundary term gives
\[
\frac{d}{dt}\sum_{j=1}^{p}\frac{(t-t_0)^{p-j}}{(p-j)!}\kappa_j(\omega)
=\mathbf 1_{\{p>1\}}\sum_{j=1}^{p-1}\frac{(t-t_0)^{p-1-j}}{(p-1-j)!}\kappa_j(\omega),
\]
which is the boundary term in \(f_{p-1}\). At \(t=t_0\), the integral part vanishes and the same polynomial term gives \(f_p(t_0,\omega)=\kappa_p(\omega)\) for every \(p=1,\ldots,r\). Combining the operator and polynomial identities proves \(\frac{d}{dt}f_p(t,\omega)=f_{p-1}(t,\omega)\) for every \(p=-r+1,\ldots,r\) and \(t\in(t_0,t_1)\).
\end{proof}

\begin{proof}[Proof of \cref{prop:integrals}]
The derivative identity follows by differentiating the sine function \(a\) times:
\[
\frac{d^a}{dt^a}\sin\{\omega_{k,L}(t+L)\}
=
\omega_{k,L}^{a}\sin\{\omega_{k,L}(t+L)+a\pi/2\}.
\]
Multiplying by \(L^{-1/2}\) gives the displayed formula for signed orders \(m=-a\le0\), including \(m=0\).

For positive signed orders, direct integration gives
\begin{align*}
(\mathcal I_{t_0}^{1}\phi_{k,L})(t)
&=
L^{-1/2}\int_{t_0}^t \sin\{\omega_{k,L}(s+L)\}\,\mathrm ds\\
&=
\frac{L^{-1/2}}{\omega_{k,L}}
\left[
\cos\{\omega_{k,L}(t_0+L)\}-\cos\{\omega_{k,L}(t+L)\}
\right],
\end{align*}
since \(L>0\) and \(k\in\mathbb N\) imply \(\omega_{k,L}\neq0\).

It remains to prove the recursion for \(q\ge2\). Write \(f=\phi_{k,L}\) and \(\omega=\omega_{k,L}\). Then \(f''=-\omega^2 f\), and Cauchy's formula for repeated integration gives
\[
(\mathcal I_{t_0}^q f'')(t)
=
\frac{1}{(q-1)!}\int_{t_0}^t (t-s)^{q-1}f''(s)\,\mathrm ds.
\]
Two integrations by parts yield
\begin{align*}
(\mathcal I_{t_0}^q f'')(t)
&=
-\frac{(t-t_0)^{q-1}}{(q-1)!}f'(t_0)
+
\frac{1}{(q-2)!}\int_{t_0}^t (t-s)^{q-2}f'(s)\,\mathrm ds\\
&=
(\mathcal I_{t_0}^{q-2}f)(t)
-
\frac{(t-t_0)^{q-2}}{(q-2)!}f(t_0)
-
\frac{(t-t_0)^{q-1}}{(q-1)!}f'(t_0),
\end{align*}
where the second equality also covers \(q=2\), with \(\mathcal I_{t_0}^0\) interpreted as the identity. Since \(f''=-\omega^2 f\),
\[
(\mathcal I_{t_0}^q f)(t)
=
-\omega^{-2}(\mathcal I_{t_0}^{q-2}f)(t)
+
\frac{(t-t_0)^{q-2}}{(q-2)!\,\omega^2}f(t_0)
+
\frac{(t-t_0)^{q-1}}{(q-1)!\,\omega^2}f'(t_0).
\]
Finally,
\[
f(t_0)=L^{-1/2}\sin\{\omega_{k,L}(t_0+L)\},
\qquad
f'(t_0)=L^{-1/2}\omega_{k,L}\cos\{\omega_{k,L}(t_0+L)\}.
\]
Substituting these two endpoint values and returning to \(f=\phi_{k,L}\), \(\omega=\omega_{k,L}\), and \(\psi_k^{[q]}=\mathcal A_{t_0}^{[q]}\phi_{k,L}\) gives the displayed recursion for every \(q\ge2\).
\end{proof}

\begin{proof}[Proof of \cref{prop:spectral-approximation-bounds}]
We prove a common estimate for the anchor case and for the differentiated case. Fix nonnegative integers $a,b$ and put $h:=\frac{\pi}{2L}$ and $\omega_{k,L}=kh$ with $a=b=0$ and $g=S_\theta$ for the anchor-kernel bound, and with $g(\omega)=\omega^{a+b}S_\theta(\omega)$ and $\omega>0$ for the derivative-block bound. In the latter case, the required boundedness, differentiability, and integrability properties of $g$ are exactly those in \cref{ass:subsection-weighted-spectral}; in the anchor case they are the explicit assumptions on $S_\theta$. Since $S_\theta$ is the spectral density of a real stationary covariance function, it is even and nonnegative, and hence so is the weight $g$ on $(0,\infty)$. All constants below may depend on \(g\), \((a,b)\), and the fixed envelope \(\Lambda\), but not on \(K\), \(L\), \(s\), or \(t\).

The assumptions imply that $g$ has a continuous one-sided extension at zero and is absolutely continuous on compact subintervals of $[0,\infty)$. Moreover,
\[
    \|g\|_\infty<\infty,\qquad
    \int_0^\infty g(\omega)\,\mathrm d\omega<\infty,
    \qquad
    \int_0^\infty |g'(\omega)|\,\mathrm d\omega<\infty .
\]
Under the Fourier inversion representation of the derivative block, dominated convergence gives
\[
D_1^aD_2^b C_\theta(s,t)
=
\frac{1}{2\pi}\int_{\mathbb R}
(\mathrm i\omega)^a(-\mathrm i\omega)^b e^{\mathrm i\omega(s-t)}S_\theta(\omega)\,\mathrm d\omega.
\]
Using the evenness of $S_\theta$, this becomes
\[
F_{a,b}(s,t)
:=
\frac{1}{\pi}\int_0^\infty
    g(\omega)
    \cos\!\left(\omega(s-t)+\frac{(a-b)\pi}{2}\right)
\,\mathrm d\omega .
\]
For $a=b=0$ this is simply the Fourier representation of $C_\theta(s,t)$. Thus \(F_{a,b}\) is the exact derivative covariance block represented in half-line cosine form.

Since the series defining $\widetilde C_{\theta,L,K}$ is finite, termwise differentiation is immediate. By \cref{prop:integrals},
\[
D_t^a\phi_{k,L}(t)
=
L^{-1/2}\omega_{k,L}^a
\sin\!\left(\omega_{k,L}(t+L)+\frac{a\pi}{2}\right),
\]
and the product-to-sum identity gives
\[
D_1^aD_2^b\widetilde C_{\theta,L,K}(s,t)
=
A_{L,K}^{[a,b]}(s,t)-B_{L,K}^{[a,b]}(s,t),
\]
where
\[
A_{L,K}^{[a,b]}(s,t)
:=
\frac{h}{\pi}\sum_{k=1}^K
    g(kh)
    \cos\!\left(kh(s-t)+\frac{(a-b)\pi}{2}\right)
\]
and, since $2Lh=\pi$,
\[
B_{L,K}^{[a,b]}(s,t)
:=
\frac{h}{\pi}\sum_{k=1}^K
    (-1)^k g(kh)
    \cos\!\left(kh(s+t)+\frac{(a+b)\pi}{2}\right).
\]

It remains to bound the Riemann sum error in $A_{L,K}^{[a,b]}$ and the reflected term $B_{L,K}^{[a,b]}$. The fixed compact-domain envelope is used twice: the main Riemann term is controlled through \(|s-t|\), while the reflected alternating term is controlled through \(|s+t|\). Let $\tau=s-t$ and define
\[
    r_\tau(\omega)
    :=g(\omega)\cos\!\left(\omega\tau+\frac{(a-b)\pi}{2}\right).
\]
Because $(s,t)\in[t_0,t_1]^2\subset[-\Lambda,\Lambda]^2$, we have $|\tau|\le 2\Lambda$, and hence
\[
\int_0^\infty |r_\tau'(\omega)|\,\mathrm d\omega
\le
\int_0^\infty |g'(\omega)|\,\mathrm d\omega
+2\Lambda\int_0^\infty g(\omega)\,\mathrm d\omega
=:M_{g,\Lambda}<\infty .
\]
For each $K$,
\begin{align*}
\left|F_{a,b}(s,t)-A_{L,K}^{[a,b]}(s,t)\right|
&\le
\frac{1}{\pi}
\left|\int_0^{Kh} r_\tau(\omega)\,\mathrm d\omega
      -h\sum_{k=1}^K r_\tau(kh)\right|
+
\frac{1}{\pi}\int_{Kh}^\infty g(\omega)\,\mathrm d\omega  \\
&\le
\frac{hM_{g,\Lambda}}{\pi}
+
\frac{1}{\pi}\int_{Kh}^\infty g(\omega)\,\mathrm d\omega .
\end{align*}
The second inequality is the elementary right endpoint Riemann sum bound obtained by applying absolute continuity on each interval $[(k-1)h,kh]$.

For the reflected term, set $x=s+t$ and
\[
    q_x(\omega)
    :=g(\omega)\cos\!\left(\omega x+\frac{(a+b)\pi}{2}\right).
\]
Again $|x|\le 2\Lambda$, so $\int_0^\infty |q_x'|\le M_{g,\Lambda}$ and $\|q_x\|_\infty\le\|g\|_\infty$. Pair consecutive terms in the alternating sum, with \(N=\lfloor K/2\rfloor\). If \(K\) is even, all terms are paired. If \(K\) is odd, one endpoint remains and is bounded by \(\|g\|_\infty\). Hence
\[
\left|\sum_{k=1}^K(-1)^k q_x(kh)\right|
\le
\sum_{j=1}^N |q_x(2jh)-q_x((2j-1)h)|+\|g\|_\infty
\le
M_{g,\Lambda}+\|g\|_\infty .
\]
Therefore
\[
\sup_{(s,t)\in[t_0,t_1]^2}
\left|B_{L,K}^{[a,b]}(s,t)\right|
\le
\frac{h}{\pi}(M_{g,\Lambda}+\|g\|_\infty)
\le
\frac{C_{g,\Lambda}}{L},
\]
where $C_{g,\Lambda}<\infty$ is independent of $K$, $L$, $s$, and $t$.

Combining the two estimates and using $Kh=\pi K/(2L)$ gives
\[
\sup_{(s,t)\in[t_0,t_1]^2}
\left|F_{a,b}(s,t)-D_1^aD_2^b\widetilde C_{\theta,L,K}(s,t)\right|
\le
\frac{E_{a,b}}{L}
+
\frac{1}{\pi}\int_{\pi K/(2L)}^\infty g(\omega)\,\mathrm d\omega,
\]
for a finite constant $E_{a,b}$ that may depend on \(g\), \((a,b)\), and \(\Lambda\), but not on $K$, $L$, $s$, and $t$. The reflected-term bound is absorbed into the \(E_{a,b}/L\) term, while the remaining half-line spectral tail has coefficient \(1/\pi\). Taking $a=b=0$ gives the anchor-kernel bound, with $E=E_{0,0}$ and $g=S_\theta$. For a fixed pair $(a,b)$ satisfying \cref{ass:subsection-weighted-spectral}, the same display gives the derivative-block bound.

Finally, if $K,L\to\infty$ and $K/L\to\infty$, then $L^{-1}\to0$ and the lower limit $\pi K/(2L)$ in the tail integral tends to infinity. Since the relevant weight $g$ is integrable, the tail integral tends to zero. This proves both uniform convergence statements.
\end{proof}

\begin{proof}[Proof of \cref{prop:finite-grid-posterior-convergence}]
The \(Y,Y\) block is positive definite because it includes \(\sigma^2 I_n\) with \(\sigma^2>0\). In particular, \(\det(\Sigma_{YY})\neq0\). Since \(\widetilde\Sigma_{YY}^{(\ell)}\to\Sigma_{YY}\), continuity of the determinant implies that \(\det(\widetilde\Sigma_{YY}^{(\ell)})\neq0\) for all sufficiently large \(\ell\), and continuity of matrix inversion gives \((\widetilde\Sigma_{YY}^{(\ell)})^{-1}\to\Sigma_{YY}^{-1}\). Applying the finite-dimensional Gaussian conditioning formulas and using continuity of addition and matrix multiplication gives the stated limits.
\end{proof}

%% file: tables/tartare_calibration_constants.tex
\begingroup
\footnotesize
\setlength{\tabcolsep}{3pt}
\begin{longtable}{p{0.12\textwidth}cp{0.36\textwidth}rr}
\caption{Completed TARTARE calibration constants for the active default entries.}\label{tab:tartare-calibration-constants}\\
Kernel & \(r\) & \(\mathcal M\) & \(m\) & \(c_{\mathcal M}\)\\
\midrule
\endfirsthead
\caption[]{Completed TARTARE calibration constants for the active default entries (continued).}\\
Kernel & \(r\) & \(\mathcal M\) & \(m\) & \(c_{\mathcal M}\)\\
\midrule
\endhead
\midrule
\multicolumn{5}{r}{\emph{Continued on next page}}\\
\endfoot
\endlastfoot
SE & 0 & \(\{0\}\) & 2.050 & 6.90\\
SE & 1 & \(\{-1\}\) & 2.680 & 8.00\\
SE & 1 & \(\{-1,0,1\}\) & 2.680 & 8.30\\
SE & 2 & \(\{-2\}\) & 3.091 & 8.30\\
SE & 2 & \(\{-2,-1,0,1,2\}\) & 3.091 & 8.30\\
SE & 3 & \(\{-3\}\) & 3.420 & 8.00\\
SE & 3 & \(\{-3,-2,-1,0,1,2,3\}\) & 3.420 & 8.30\\
SE & 4 & \(\{-4\}\) & 3.704 & 6.90\\
SE & 4 & \(\{-4,-3,-2,-1,0,1,2,3,4\}\) & 3.704 & 8.30\\
\midrule
Mat3/2 & 0 & \(\{0\}\) & 4.648 & 8.60\\
\midrule
Mat5/2 & 0 & \(\{0\}\) & 3.209 & 8.00\\
Mat5/2 & 1 & \(\{-1\}\) & 9.658 & 8.30\\
Mat5/2 & 1 & \(\{-1,0,1\}\) & 9.658 & 9.70\\
\midrule
Mat7/2 & 0 & \(\{0\}\) & 2.785 & 7.70\\
Mat7/2 & 1 & \(\{-1\}\) & 5.664 & 8.30\\
Mat7/2 & 1 & \(\{-1,0,1\}\) & 5.664 & 9.30\\
Mat7/2 & 2 & \(\{-2\}\) & 14.444 & 6.70\\
Mat7/2 & 2 & \(\{-2,-1,0,1,2\}\) & 14.444 & 9.30\\
\midrule
Mat9/2 & 0 & \(\{0\}\) & 2.586 & 7.40\\
Mat9/2 & 1 & \(\{-1\}\) & 4.543 & 8.30\\
Mat9/2 & 1 & \(\{-1,0,1\}\) & 4.543 & 9.00\\
Mat9/2 & 2 & \(\{-2\}\) & 7.906 & 7.20\\
Mat9/2 & 2 & \(\{-2,-1,0,1,2\}\) & 7.906 & 9.00\\
Mat9/2 & 3 & \(\{-3\}\) & 19.183 & 4.30\\
Mat9/2 & 3 & \(\{-3,-2,-1,0,1,2,3\}\) & 19.183 & 9.00\\
\midrule
Mat11/2 & 0 & \(\{0\}\) & 2.472 & 7.40\\
Mat11/2 & 1 & \(\{-1\}\) & 4.029 & 8.30\\
Mat11/2 & 1 & \(\{-1,0,1\}\) & 4.029 & 9.00\\
Mat11/2 & 2 & \(\{-2\}\) & 6.093 & 7.40\\
Mat11/2 & 2 & \(\{-2,-1,0,1,2\}\) & 6.093 & 9.00\\
Mat11/2 & 3 & \(\{-3\}\) & 10.098 & 4.90\\
Mat11/2 & 3 & \(\{-3,-2,-1,0,1,2,3\}\) & 10.098 & 9.00\\
Mat11/2 & 4 & \(\{-4\}\) & 23.905 & 3.50\\
Mat11/2 & 4 & \(\{-4,-3,-2,-1,0,1,2,3,4\}\) & 23.905 & 9.00\\
\end{longtable}
\par\smallskip
\noindent\footnotesize\emph{Note.} The monitoring set \(\mathcal M\subseteq\{-r,\ldots,r\}\) lists ensemble levels \(f_p\): \(f_0\) is the anchor process, \(f_{-j}\) is the \(j\)th derivative level, and \(f_j\) is the \(j\)-fold anchored integral level for \(j>0\). Rows with positive monitored indices use the left-boundary normalized integration anchor. All rows use \(\varepsilon_{\mathrm{tail}}=0.01\) and \(u\in[0.05,1]\). The shared acceptance tolerances are \(\tau_{\mathrm{joint}}=0.01\), \(\tau_{\mathrm{pair}}=0.02\), \(\tau_{\mathrm{var}}=0.02\), and, for rows with positive monitored indices, \(\tau_{t_0}=0.02\). Displayed constants are rounded.
\par
\endgroup

%% file: tables/supplement_tartare_d2_vs_full_diagnostics.tex
\begin{tabular}{llrrrrrrrrrrrr}
 & & & \multicolumn{2}{c}{Completed} & \multicolumn{3}{c}{$E_{\mathrm{joint}}$} & \multicolumn{3}{c}{$E_{\mathrm{pair,max}}$} & \multicolumn{3}{c}{$E_{\mathrm{var,max}}$} \\
\cmidrule(lr){4-5}\cmidrule(lr){6-8}\cmidrule(lr){9-11}\cmidrule(lr){12-14}
$\rho$ & Kernel & $n$ & D2 & Full & D2 & Full & Change & D2 & Full & Change & D2 & Full & Change \\
\midrule
0.65 & SE & 25 & 967 & 966 & 0.160 & 0.160 & 0.2\% & 0.515 & 0.514 & 0.1\% & 0.361 & 0.361 & 0.2\% \\
0.65 & SE & 50 & 989 & 988 & 0.155 & 0.155 & 0.1\% & 0.540 & 0.539 & 0.2\% & 0.281 & 0.282 & 0.0\% \\
0.65 & Mat7/2 & 25 & 963 & 956 & 0.177 & 0.176 & 0.6\% & 0.916 & 0.911 & 0.5\% & 0.372 & 0.374 & -0.3\% \\
0.65 & Mat7/2 & 50 & 989 & 992 & 0.172 & 0.173 & -0.4\% & 1.098 & 1.089 & 0.8\% & 0.326 & 0.325 & 0.3\% \\
0.35 & SE & 25 & 996 & 996 & 0.169 & 0.170 & -0.3\% & 0.705 & 0.707 & -0.3\% & 0.285 & 0.285 & -0.1\% \\
0.35 & SE & 50 & 1000 & 1000 & 0.164 & 0.165 & -0.2\% & 0.824 & 0.827 & -0.4\% & 0.258 & 0.260 & -0.9\% \\
0.35 & Mat7/2 & 25 & 989 & 987 & 0.188 & 0.190 & -1.0\% & 1.177 & 1.174 & 0.2\% & 0.335 & 0.330 & 1.3\% \\
0.35 & Mat7/2 & 50 & 996 & 1000 & 0.181 & 0.181 & 0.3\% & 1.599 & 1.596 & 0.2\% & 0.292 & 0.295 & -0.9\% \\
\end{tabular}

%% file: tables/boundary_constant_kappa_summary.tex
\begin{tabular}{llrrrr}
Condition & Constant & Mean & SD & Coverage & SD/prior SD \\
\midrule
Known constants & \(\kappa_1\) & 0.600 & 0.000 & 1.000 & -- \\
Known constants & \(\kappa_2\) & -0.400 & 0.000 & 1.000 & -- \\
Fixed wrong constants & \(\kappa_1\) & 0.000 & 0.000 & 0.000 & -- \\
Fixed wrong constants & \(\kappa_2\) & 0.000 & 0.000 & 0.000 & -- \\
Independent prior & \(\kappa_1\) & 0.000 & 0.463 & 0.500 & 1.000 \\
Independent prior & \(\kappa_2\) & 0.000 & 0.463 & 0.750 & 1.000 \\
Direct boundary observation & \(\kappa_1\) & 0.501 & 0.086 & 0.746 & 0.302 \\
Direct boundary observation & \(\kappa_2\) & -0.335 & 0.086 & 0.805 & 0.302 \\
\end{tabular}

%% file: tables/boundary_constant_variance_decomposition.tex
\begin{tabular}{llrrrr}
Condition & Level & Canonical variance \((10^3)\) & Constant variance \((10^3)\) & Total variance \((10^3)\) & Constant share \\
\midrule
Known constants & \(I_1\) & 0.589 & 0.000 & 0.589 & 0.000 \\
Known constants & \(I_2\) & 0.397 & 0.000 & 0.397 & 0.000 \\
Fixed wrong constants & \(I_1\) & 0.589 & 0.000 & 0.589 & 0.000 \\
Fixed wrong constants & \(I_2\) & 0.397 & 0.000 & 0.397 & 0.000 \\
Independent prior & \(I_1\) & 0.589 & 330.625 & 331.214 & 0.983 \\
Independent prior & \(I_2\) & 0.397 & 771.480 & 771.877 & 0.995 \\
Direct boundary observation & \(I_1\) & 0.589 & 12.049 & 12.638 & 0.781 \\
Direct boundary observation & \(I_2\) & 0.397 & 28.115 & 28.512 & 0.902 \\
\end{tabular}

%% file: tables/boundary_constant_integral_error_summary.tex
\begin{tabular}{llrrrr}
Condition & Level & IRMSE & Mean absolute error & Coverage & Interval width \\
\midrule
Known constants & \(I_1\) & 0.0209 & 0.0178 & 0.937 & 0.087 \\
Known constants & \(I_2\) & 0.0161 & 0.0125 & 0.934 & 0.061 \\
Fixed wrong constants & \(I_1\) & 0.6003 & 0.6001 & 0.000 & 0.087 \\
Fixed wrong constants & \(I_2\) & 0.4003 & 0.3336 & 0.024 & 0.061 \\
Independent prior & \(I_1\) & 0.6003 & 0.6001 & 0.500 & 1.818 \\
Independent prior & \(I_2\) & 0.4003 & 0.3336 & 0.855 & 2.683 \\
Direct boundary observation & \(I_1\) & 0.1352 & 0.1337 & 0.761 & 0.357 \\
Direct boundary observation & \(I_2\) & 0.1352 & 0.1193 & 0.896 & 0.505 \\
\end{tabular}

%% file: tables/boundary_constant_location_summary.tex
\begin{tabular}{lllrr}
Condition & Level & Location & RMSE & Coverage \\
\midrule
Known constants & \(I_1\) & Anchor & 0.000 & 0.000 \\
Known constants & \(I_1\) & Midpoint & 0.024 & 0.944 \\
Known constants & \(I_1\) & Far endpoint & 0.035 & 0.948 \\
Known constants & \(I_2\) & Anchor & 0.000 & 0.000 \\
Known constants & \(I_2\) & Midpoint & 0.014 & 0.948 \\
Known constants & \(I_2\) & Far endpoint & 0.040 & 0.944 \\
Fixed wrong constants & \(I_1\) & Anchor & 0.600 & 0.000 \\
Fixed wrong constants & \(I_1\) & Midpoint & 0.601 & 0.000 \\
Fixed wrong constants & \(I_1\) & Far endpoint & 0.601 & 0.000 \\
Fixed wrong constants & \(I_2\) & Anchor & 0.400 & 0.000 \\
Fixed wrong constants & \(I_2\) & Midpoint & 0.201 & 0.000 \\
Fixed wrong constants & \(I_2\) & Far endpoint & 0.801 & 0.000 \\
Independent prior & \(I_1\) & Anchor & 0.600 & 0.500 \\
Independent prior & \(I_1\) & Midpoint & 0.601 & 0.501 \\
Independent prior & \(I_1\) & Far endpoint & 0.601 & 0.501 \\
Independent prior & \(I_2\) & Anchor & 0.400 & 0.750 \\
Independent prior & \(I_2\) & Midpoint & 0.201 & 1.000 \\
Independent prior & \(I_2\) & Far endpoint & 0.801 & 0.750 \\
Direct boundary observation & \(I_1\) & Anchor & 0.145 & 0.746 \\
Direct boundary observation & \(I_1\) & Midpoint & 0.150 & 0.762 \\
Direct boundary observation & \(I_1\) & Far endpoint & 0.154 & 0.771 \\
Direct boundary observation & \(I_2\) & Anchor & 0.116 & 0.805 \\
Direct boundary observation & \(I_2\) & Midpoint & 0.121 & 0.950 \\
Direct boundary observation & \(I_2\) & Far endpoint & 0.255 & 0.826 \\
\end{tabular}

%% file: tables/motorcycle_covariance_summary.tex
\begin{table}

\caption{\label{tab:motorcycle-covariance-summary}Posterior summary for motorcycle covariance parameters from the Stan fit.}
\centering
\begin{tabular}[t]{lrrrrrrrr}
\toprule
Parameter & Mean & Median & SD & 5\% & 95\% & R-hat & ESS (bulk) & ESS (tail)\\
\midrule
$\alpha_{\mu}$ & 46.976 & 44.168 & 13.750 & 30.716 & 72.976 & 1.001 & 1659.614 & 3336.432\\
$\rho_{\mu}$ & 6.300 & 6.188 & 1.141 & 4.613 & 8.293 & 1.001 & 2028.439 & 3317.832\\
$\sigma_{\mu}$ & 46.257 & 40.735 & 26.124 & 15.791 & 95.301 & 1.001 & 4077.915 & 4188.476\\
$\alpha_h$ & 1.998 & 1.893 & 0.569 & 1.313 & 3.022 & 1.001 & 3172.925 & 4510.824\\
$\rho_h$ & 10.520 & 10.145 & 2.733 & 6.895 & 15.273 & 1.004 & 1694.838 & 2025.953\\
\addlinespace
$\sigma_h$ & 1.599 & 1.413 & 0.857 & 0.581 & 3.181 & 1.000 & 5196.910 & 5044.803\\
\bottomrule
\end{tabular}
\end{table}